\documentclass{article}

\usepackage[vmargin=1.18in,hmargin=1.1in]{geometry}
\usepackage{amsmath,amssymb,amsfonts,graphicx,amsthm,nicefrac,mathtools,bm}
\usepackage{hyperref}
\usepackage[numbers]{natbib}
\usepackage[english]{babel}
\usepackage[T1]{fontenc}
\usepackage{bbm,mdframed}
\usepackage[dvipsnames]{xcolor}
\usepackage{xspace}
\usepackage{rotating,multirow}
\usepackage{enumerate,graphics,enumitem}
\usepackage[algo2e]{algorithm2e}
\usepackage{booktabs}

\setlist[itemize]{noitemsep, topsep=0pt}
\setlist[enumerate]{itemsep=5pt, topsep=5pt, leftmargin=25pt}

\newtheorem{theorem}{Theorem}

\definecolor{verylightblue}{rgb}{0.7,0.8,1}
  {\begin{mdframed}[backgroundcolor=verylightblue]\begin{theorem}}%
  {\end{theorem}\end{mdframed}}

\definecolor{verylightgray}{gray}{0.95}
  {\begin{mdframed}[backgroundcolor=verylightgray]\begin{proof}}%
  {\end{proof}\end{mdframed}}

\newtheorem{lemma}{Lemma}
\definecolor{verylightred}{rgb}{1,0.8,0.8}
  {\begin{mdframed}[backgroundcolor=verylightred]\begin{lemma}}%
  {\end{lemma}\end{mdframed}}

\newtheorem{proposition}{Proposition}
  {\begin{mdframed}[backgroundcolor=verylightblue]\begin{proposition}}%
  {\end{proposition}\end{mdframed}}

\newtheorem{fact}{Fact}
\theoremstyle{definition}
\newtheorem{definition}{Definition}
\theoremstyle{remark}

\makeatletter

\newtheorem*{rep@theorem}{\rep@title}
\newcommand{\newreptheorem}[2]
{\newenvironment{rep#1}[1]
{\def\rep@title{#2 \ref{##1}} \begin{rep@theorem}}%
 {\end{rep@theorem}}}
\makeatother
\newreptheorem{theorem}{Theorem}
\newreptheorem{lemma}{Lemma}
\newreptheorem{corollary}{Corollary}
\newreptheorem{proposition}{Proposition}

\newcommand{\figref}[1]{Figure~\ref{fig:#1}}

\newcommand{\secref}[1]{Section~\ref{sec:#1}}

\newcommand{\defref}[1]{Definition~\ref{def:#1}}

\newcommand{\lemref}[1]{Lemma~\ref{lem:#1}}

\newcommand{\factref}[1]{Fact~\ref{fact:#1}}
\newcommand{\thmref}[1]{Theorem~\ref{thm:#1}}

\newcommand{\tabref}[1]{Table~\ref{tab:#1}}

\newcommand{\eqnref}[1]{\eqref{eqn:#1}}

\DeclareMathOperator*{\argmax}{arg\,max}

\newcommand{\PP}[1]{\mathbb{P}\!\left\{{#1}\right\}} 
\newcommand{\EE}[1]{\mathbb{E}\left[{#1}\right]} 

\newcommand{\EEst}[2]{\mathbb{E}\left[{#1}\ \middle| \ {#2}\right]} 
\newcommand{\PPst}[2]{\mathbb{P}\!\left\{{#1}\ \middle| \ {#2}\right\}} 

\def\R{\mathbb{R}}

\newcommand{\ignore}[1]{}

\usepackage{tikz}
\usetikzlibrary{arrows}
\usetikzlibrary{shapes}

\newcommand{\theauthor}{Tijana Zrnic\textsuperscript{1} ~ ~ ~ Daniel L. Jiang\textsuperscript{2} ~ ~ ~Aaditya Ramdas\textsuperscript{3} ~ ~ ~ Michael I. Jordan\textsuperscript{4}\\ \\
Department of Electrical Engineering and Computer Sciences, UC Berkeley\textsuperscript{1,4}\\
Amazon\textsuperscript{2}\\
Department of Statistics and Data Science, Carnegie Mellon University\textsuperscript{3}\\
Department of Statistics, UC Berkeley\textsuperscript{4}\\
{\small \texttt{\{tijana\textsuperscript{1},jordan\textsuperscript{4}\}@eecs.berkeley.edu, jiadanie@amazon.com\textsuperscript{2}, aramdas@cmu.edu\textsuperscript{3}} }
}
\newcommand{\thetitle}{The Power of Batching in Multiple Hypothesis Testing}

\date{\vspace{-1ex}}
\author{\theauthor}
\title{\thetitle}

\usepackage[mmddyy]{datetime}

\newcommand{\nulls}{\mathcal{H}^0}

\newcommand{\fdp}{\textnormal{FDP}}
\newcommand{\fdr}{\textnormal{FDR}}

\newcommand{\mfdr}{\textnormal{mFDR}}

\newcommand{\fdphat}{\widehat{\fdp}}

\newcommand{\bh}{\textnormal{BH}}

\newcommand{\stbh}{\textnormal{St-BH}}
\newcommand{\One}[1]{{\bf{1}}\left\{{#1}\right\}}

\newcommand{\pih}{\widehat{\pi}}

\def\N{\mathbb N}

\def\F{\mathcal{F}}

\def\cR{\mathcal{R}}

\newcommand{\batchbh}{\text{Batch}_{\text{BH}}}
\newcommand{\batchsbh}{\text{Batch}_{\text{St-BH}}}

\makeatletter
\newcommand{\dotfrac}[2]{
\mathchoice
{\ooalign{$\genfrac{}{}{0pt}{0}{#1}{#2}$\cr\leavevmode\cleaders\hb@xt@ .22em{\hss $\displaystyle\cdot$\hss}\hfill\kern\z@\cr}}
{\ooalign{$\genfrac{}{}{0pt}{1}{#1}{#2}$\cr\leavevmode\cleaders\hb@xt@ .22em{\hss $\textstyle\cdot$\hss}\hfill\kern\z@\cr}}
{\ooalign{$\genfrac{}{}{0pt}{2}{#1}{#2}$\cr\leavevmode\cleaders\hb@xt@ .22em{\hss $\scriptstyle\cdot$\hss}\hfill\kern\z@\cr}}
{\ooalign{$\genfrac{}{}{0pt}{3}{#1}{#2}$\cr\leavevmode\cleaders\hb@xt@ .22em{\hss $\scriptscriptstyle\cdot$\hss}\hfill\kern\z@\cr}}
}
\makeatother

\usepackage{algorithm,algorithmic}

\newcommand{\defn}{\ensuremath{:\, =}}

\makeatletter
\long\def\@makecaption#1#2{
        \vskip 0.8ex
        \setbox\@tempboxa\hbox{\small {\bf #1:} #2}
        \parindent 1.5em  
        \dimen0=\hsize
        \advance\dimen0 by -3em
        \ifdim \wd\@tempboxa >\dimen0
                \hbox to \hsize{
                        \parindent 0em
                        \hfil 
                        \parbox{\dimen0}{\def\baselinestretch{0.96}\small
                                {\bf #1.} #2
                                } 
                        \hfil}
        \else \hbox to \hsize{\hfil \box\@tempboxa \hfil}
        \fi
        }
\makeatother

\begin{document}

 \maketitle


\begin{abstract}
    One important partition of algorithms for controlling the false discovery rate (FDR) in multiple testing is into \emph{offline} and \emph{online} algorithms. The first generally achieve significantly higher power of discovery, while the latter allow making decisions sequentially as well as adaptively formulating hypotheses based on past observations. Using existing methodology, it is unclear how one could trade off the benefits of these two broad families of algorithms, all the while preserving their formal FDR guarantees. To this end, we introduce $\batchbh$ and $\batchsbh$, algorithms for controlling the FDR when a possibly infinite sequence of batches of hypotheses is tested by repeated application of one of the most widely used offline algorithms, the Benjamini-Hochberg (BH) method or Storey's improvement of the BH method. We show that our algorithms interpolate between existing online and offline methodology, thus trading off the best of both worlds.
 \end{abstract}

\section{Introduction}
Consider the setting in which a large number of decisions need to be made (e.g., hypotheses to be tested), and one wishes to achieve some form of aggregate control over the quality of these decisions. For binary decisions, a seminal line of research has cast this problem in terms of an error metric known as the \emph{false discovery rate} (FDR)~\citep{BH95}. The FDR has a Bayesian flavor, conditioning on the decision to reject (i.e., conditioning on a ``discovery'') and computing the fraction of discoveries that are false.  This should be contrasted with traditional metrics---such as sensitivity, specificity, Type I and Type II errors---where one conditions not on the decision but rather on the hypothesis---whether the null or the alternative is true.  The scope of research on FDR control has exploded in recent years, with progress on problems such as dependencies, domain-specific constraints, and contextual information.

Classical methods for FDR control are ``offline'' or ``batch'' methods, taking in a single batch of data and outputting a set of decisions for all hypotheses at once. This is a serious limitation in the setting of emerging applications at planetary scale, such as A/B testing in the IT industry \citep{kohavi2017online}, and researchers have responded by developing a range of \emph{online FDR control} methods~\citep{foster2008alpha, aharoni2014generalized, javanmard2016online,ramdas2018saffron,tian2019addis}.  In the online setting, a decision is made at every time step with no knowledge of future tests, and with possibly infinitely many tests to be conducted overall. By construction, online FDR algorithms guarantee that the FDR is controlled during the whole sequence of tests, and not merely at the end. 

Online and offline FDR methods both have their pros and cons. Online methods allow the testing of infinitely many hypotheses, and require less coordination in the setting of multiple decision-makers.  Also, perhaps most importantly, they allow the scientist to choose new hypotheses adaptively, depending on the results of previous tests. On the other hand, offline FDR methods tend to make significantly more discoveries due to the fact that they have access to \emph{all} test statistics before making decisions, and not just to the ones from past tests. That is, online methods are myopic, and this can lead to a loss of statistical power.  Moreover, the decisions of offline algorithms are \emph{stable}, in the sense that they are invariant to any implicit ordering of hypotheses; this is not true of online algorithms, whose discovery set can vary drastically depending on the ordering of hypotheses \citep{foster2008alpha}.

By analogy with batch and online methods in gradient-based optimization, these considerations suggest investigating an intermediate notion of ``mini-batch,'' hoping to exploit and manage some form of tradeoff between methods that are purely batch or purely online.

Managing such a tradeoff is, however, more challenging in the setting of false-discovery-rate control than in the optimization setting. Indeed, consider a naive approach that would run offline algorithms on different batches of hypotheses in an online fashion. Unfortunately, such a method violates the assumptions behind FDR control, yielding uncontrolled, possibly meaningless FDR guarantees. To illustrate this point, \figref{naive} plots the  performance of the Benjamini-Hochberg (BH) algorithm \cite{BH95} and Storey's improved version of the BH algorithm (Storey-BH) \cite{Storey02, Storey04}, run repeatedly under the same FDR level $0.05$ on different batches of hypotheses. We observe that the FDR can be much higher than the nominal value.

\begin{figure}
  \centering
  \includegraphics[width=0.49\columnwidth]{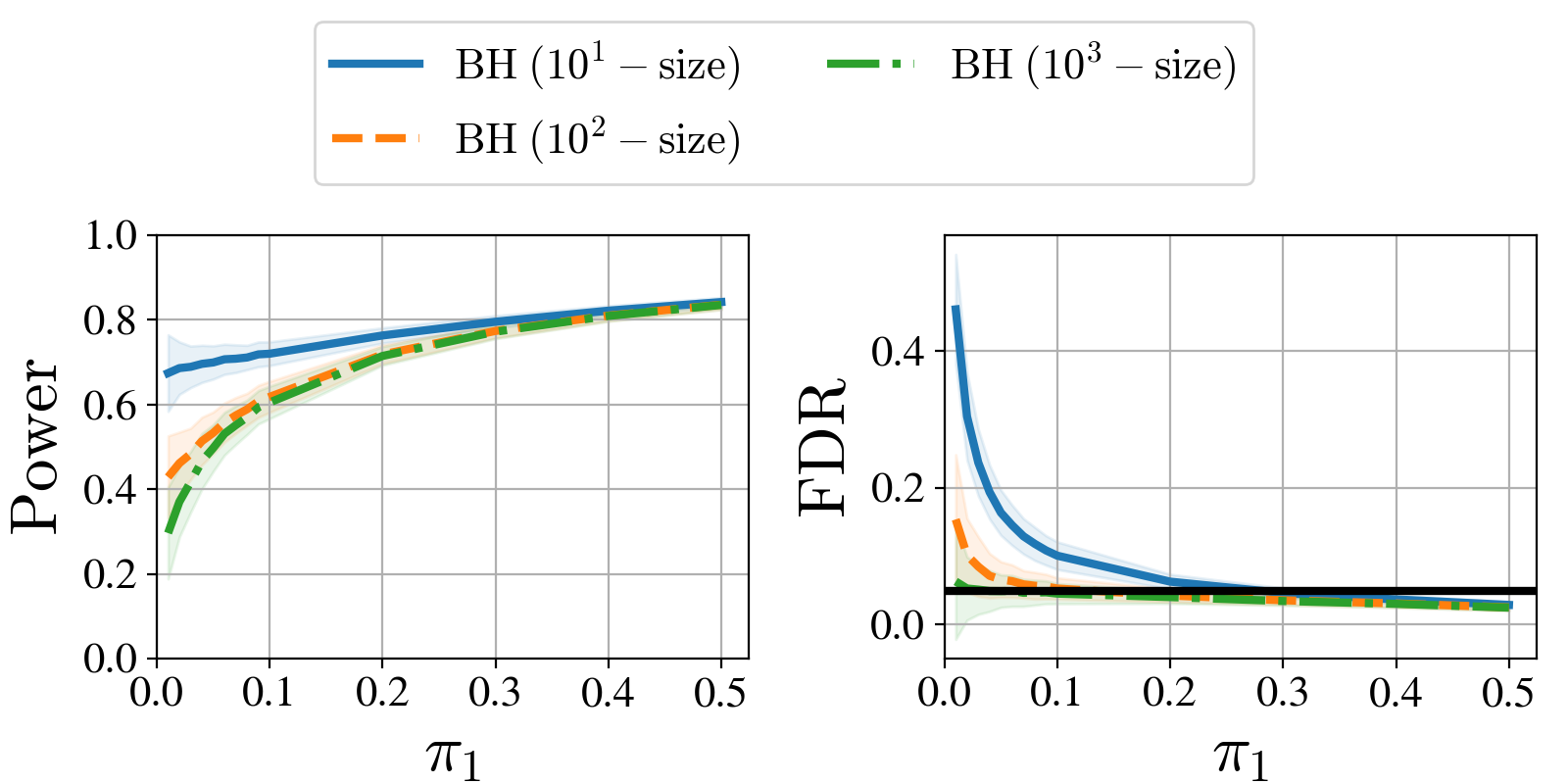}
  \includegraphics[width=0.49\columnwidth]{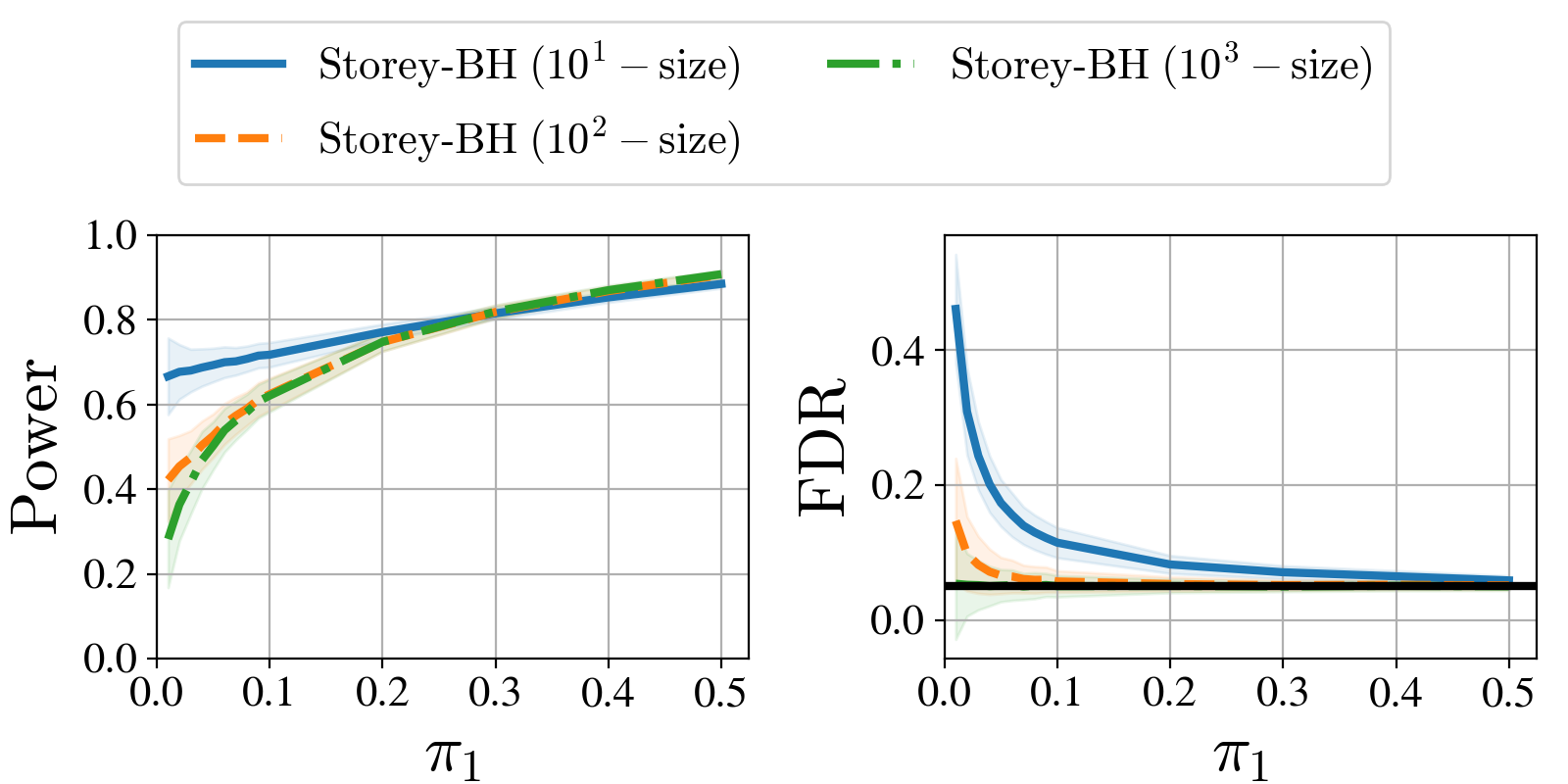}
  \caption{Statistical power and FDR versus probability of non-null hypotheses, $\pi_1$, for naively composed BH (left) and Storey-BH (right), at batch sizes 10, 100, and 1000. The total number of hypotheses is 3000, and the target FDR is 0.05.}
  \label{fig:naive}
\end{figure}

In this paper, we develop FDR procedures which are appropriate for multiple batches of tests. We allow testing of possibly infinitely many batches in an online fashion. We refer to this setting as \emph{online batch testing}. More precisely, we improve the widely-used BH algorithm \citep{BH95} and a variant that we refer to Storey-BH \citep{Storey02,Storey04}, such that their repeated composition does not violate the desired FDR guarantees. We refer to these sequential, FDR-preserving versions of BH and Storey-BH as $\batchbh$ and $\batchsbh$\footnote{Both $\batchbh$ and $\batchsbh$ have been incorporated into the \href{https://dsrobertson.github.io/onlineFDR}{onlineFDR package} \cite{onlineFDR}.}, respectively. As is the case for state-of-the-art online algorithms, our procedures allow testing an infinite sequence of batches of adaptively chosen hypotheses, but they also enjoy a higher power of discovery than those algorithms. Finally, since they consist of compositions of offline FDR algorithms with provable guarantees, they immediately imply FDR control over each constituent batch, and not just over the whole sequence of tests. This property has value in settings with natural groupings of hypotheses, where the scientist might be interested in the overall FDR, but also the FDR over certain subgroups of hypotheses.

\subsection{Outline}

In \secref{notation}, we present preliminaries and sketch the main ideas behind our proofs. In \secref{bh}, we define the $\batchbh$ family of algorithms and state its FDR guarantees. In \secref{sbh}, we do the same for $\batchsbh$ algorithms. 
While $\batchbh$ and $\batchsbh$ make conservative dependence assumptions on the $p$-value sequence, in \secref{prds} we describe algorithms for online batch testing of positively dependent $p$-values. In \secref{experiments}, we demonstrate the performance of our methods on synthetic data. In \secref{credit}, we demonstrate the performance of our methods on a real fraud detection data set. In the Appendix, we give a short overview of some related work, provide all the proofs, as well as additional experimental results.


\section{Preliminaries}
\label{sec:notation}

We introduce a formal description of the testing process, together with some preliminaries.

At every time $t\in\N$, a batch of $n_t$ hypotheses is tested using a pre-specified offline FDR procedure. We consider two such procedures, the BH and Storey-BH procedures, which we review in the Appendix for the reader's convenience. The batches arrive sequentially, in a stream; at the time of testing the $t$-th batch, no information about future batches needs to be available, such as their size or their number. For each hypothesis, there is unknown ground truth that says whether the hypothesis is null or non-null. Denote the set of hypotheses in the $t$-th batch by $\mathbf{H}_t \defn \{H_{t,1},\dots,H_{t,n_t}\}$. Each hypothesis has a $p$-value associated with it. Let $\mathbf{P}_t$ denote the $p$-values corresponding to the $t$-th batch of hypotheses, given by $\mathbf{P}_t \defn \{P_{t,1},\dots,P_{t,n_t}\}$, where $P_{t,j}$ is the $j$-th $p$-value in batch $t$. Denote by $\nulls_t$ the indices corresponding to null hypotheses in batch $t$, and let $\cR_t$ denote the indices of rejections, or \emph{discoveries}, in batch $t$:
\begin{gather*}
    \nulls_t \defn \{i: H_{t,i} \text{ is null}\},~ ~  \cR_t \defn \{i: H_{t,i} \text{ is rejected}\}.
\end{gather*}
We will also informally say that a $p$-value is rejected, if its corresponding hypothesis is rejected.

We now define the \emph{false discovery rate (FDR) up to time} $t$:
$$\fdr(t) \defn \EE{\fdp(t)}\defn \EE{\frac{\sum_{s=1}^t|\nulls_s \cap \cR_s|}{(\sum_{s=1}^t|\cR_s|)\vee 1}},$$
where $\fdp(t)$ denotes a random quantity called the \emph{false discovery proportion} up to time $t$. To simplify notation, we also define $R_t \defn |\cR_t|$. In real applications, it does not suffice to merely control the FDR (which we can do by making no discoveries, which results in $\fdr = 0$); rather, we also need to achieve high statistical \emph{power}:
$$\text{Power}(t) \defn \EE{\frac{\sum_{s=1}^t|([n_s]\setminus \nulls_s) \cap \cR_s|}{\sum_{s=1}^t|([n_s]\setminus \nulls_s)|}},$$
where $[n_s]\setminus \nulls_s$ are the non-null hypotheses in batch $s$.

The goal of the $\batchbh$ procedure is to achieve high power, while guaranteeing $\fdr(t)\leq \alpha$ for a pre-specified level $\alpha\in(0,1)$ and for all $t\in\N$. To do so, the algorithm adaptively determines a \textit{test level} $\alpha_t$ based on information about past batches of tests, and tests $\mathbf{P}_t$ under FDR level $\alpha_t$ using the standard BH method. The $\batchsbh$ method operates in a similar way, the difference being that it uses the Storey-BH method for every batch, as opposed to BH.

Define $R_t^+$ to be the maximum ``augmented'' number of rejections in batch $t$, if one $p$-value in $\mathbf{P}_t$ is ``hallucinated'' to be equal to zero, and all other $p$-values and level $\alpha_t$ are held fixed; the maximum is taken over the choice of the $p$-value which is set to zero. More formally, let $\mathcal{A}_t$ denote a map from a set of $p$-values $\mathbf{P}_t$ (and implicitly, a level $\alpha_t$) to a set of rejections $\cR_t$. Hence, $R_t = |\mathcal{A}_t(\mathbf{P}_t)|$. In our setting, $\mathcal{A}_t$ will be the BH algorithm in the case of $\batchbh$ and Storey-BH algorithm in the case of $\batchsbh$. Then, $R_t^+$ is defined as 
\begin{equation}
    R_t^+ \defn \max_{i\in[n_t]} |\mathcal{A}_t(\mathbf{P}_t\setminus P_{t,i}\cup 0)|.
\end{equation}
Note that $R_t^+$ could be as large as $n_t$ in general. For an extreme example, let $n_t = 3$, $\mathbf{P}_t:=\{2\alpha/3, \alpha, 4\alpha/3\}$, and consider $\mathcal{A}_t$ being the BH procedure. Then $R_t=0$, while $R_t^+=3$. However, such ``adversarial'' p-values are unlikely to be encountered in practice and we typically expect $R_t^+$ to be roughly equal to $R_t + 1$. In other words, we expect that when an unrejected $p$-value is set to 0, it will be a new rejection, but typically will not result in other rejections as well. This intuition is confirmed by our experiments, where we plot $R_t^+ - R_t$ for $\batchbh$ with different batch sizes and observe that this quantity concentrates around 1. These plots are available in \figref{rdiff} in the Appendix.

Let the natural filtration induced by the testing process be denoted
\[
\F^t \defn \sigma(\mathbf{P}_1,\dots,\mathbf{P}_t),
\] 
which is the $\sigma$-field of all previously observed $p$-values.
Naturally, we require $\alpha_t$ to be $\F^{t-1}$-measurable; the test level at time $t$ is only allowed to depend on information seen before $t$. It is worth pointing out that this filtration is different from the corresponding filtration in prior online FDR work, which was typically of the form $\sigma(R_1,\dots,R_t)$. The benefits of this latter, smaller filtration arise when proving \emph{modified} FDR ($\mfdr$) guarantees, which we do not consider in this paper. Moreover, a richer filtration allows more freedom in choosing $\alpha_t$, making our choice of $\F^t$ a natural one.

For the formal guarantees of $\batchbh$ and $\batchsbh$, we will require the procedures to be \textit{monotone}. Let $(\{P_{1,1},\dots,P_{1,n_1}\},\dots,\{P_{t,1},\dots,P_{t,n_t}\})$ and $(\{\tilde P_{1,1},\dots,\tilde P_{1,n_1}\},\dots,\{\tilde P_{t,1},\dots,\tilde P_{t,n_t}\})$ be two sequences of $p$-value batches,
which are identical in all entries but $(s,i)$, for some $s\leq t$: $\tilde P_{s,i}<P_{s,i}$. Then,  \[
\text{a procedure is monotone if } \sum_{r=s+1}^t R_r \leq \sum_{r=s+1}^t \tilde R_r.\]
Intuitively, this condition says that making any of the tested $p$-values smaller can only make the overall number of rejections larger. A similar assumption appears in online FDR literature \citep{javanmard2016online,ramdas2018saffron, zrnic2018asynchronous, tian2019addis}. In general, whether or not a procedure is monotone is a property of the $p$-value distribution; notice, however, that monotonicity can be assessed empirically (it does not depend on the unknown ground truth). One way to ensure monotonicity is to make $\alpha_t$ a coordinate-wise non-increasing function of $(P_{1,1},\dots,P_{1,n_1},P_{2,1},\dots,P_{t-1,n_{t-1}})$. In the Appendix, we give examples of monotone strategies.

Finally, we review a basic property of null $p$-values. If a hypothesis $H_{t,i}$ is truly null, then the corresponding $p$-value $P_{t,i}$ stochastically dominates the uniform distribution, or is \emph{super-uniformly distributed}, meaning that:
\begin{gather*}
    \text{If the hypothesis } H_{t,i} \text{ is null},   \text{ then } \PP{P_{t,i} \leq u} \leq u \text{ for all } u \in [0, 1].
\end{gather*}

\subsection{Algorithms via Empirical FDP Estimates}

We build on Storey's interpretation of the BH procedure \citep{Storey02} as an empirical Bayesian procedure, based on empirical estimates of the false discovery proportion. In this section, we give a sketch of this idea, as it is at the core of our algorithmic constructions. The steps presented below are not fully rigorous, but are simply meant to develop intuition.

When an algorithm decides to reject a hypothesis, there is generally no way of knowing if the rejected hypothesis is null or non-null. Consequently, it is impossible for the scientist to know the achieved FDP. However, by exploiting the super-uniformity of null $p$-values, it is possible to estimate the behavior of the FDP \emph{on average}. More explicitly, there are tools that utilize only the information available to the scientist to upper bound the average FDP, that is the FDR.

We sketch this argument for the $\batchbh$ procedure here, formalizing the argument in \thmref{fdrbatchbh}.  \thmref{fdrbatchsbh} gives an analogous proof for the $\batchsbh$ procedure.

By definition, the FDR is equal to
$$\EE{\frac{\sum_{s=1}^t|\nulls_s \cap \cR_s|}{(\sum_{r=1}^t|\cR_r|)\vee 1}} = \sum_{s=1}^t \EE{\frac{ \sum_{i\in\nulls_s} \One{P_{s,i}\leq \frac{\alpha_s}{n_s}R_s}}{(\sum_{r=1}^t|\cR_r|)\vee 1}},$$
where we use the definition of the BH procedure. If the $p$-values are independent, we will show that it is valid to upper bound this expression by inserting an expectation in the numerator, approximately as
$$\sum_{s=1}^t \EE{\frac{ \sum_{i\in\nulls_s} \PPst{P_{s,i}\leq \frac{\alpha_s}{n_s}R_s}{\alpha_s,R_s}}{(\sum_{r=1}^t|\cR_r|)\vee 1}}.$$
Invoking the super-uniformity of null $p$-values (and temporarily ignoring dependence between $P_{s,i}$ and $R_s$), we get
$$\sum_{s=1}^t \EE{\frac{ |\nulls_s| \frac{\alpha_s}{n_s}R_s}{(\sum_{r=1}^t|\cR_r|)\vee 1}} \leq \EE{\frac{ \sum_{s=1}^t \alpha_s R_s}{(\sum_{r=1}^t|\cR_r|)\vee 1}}.$$
Suppose we define $\fdphat_{\batchbh}(t)\approx \frac{ \sum_{s=1}^t \alpha_s R_s}{(\sum_{r=1}^t|\cR_r|)\vee 1}$. This quantity is purely \emph{empirical}; each term is known to the scientist. Hence, by an appropriate choice of $\alpha_s$ at each step, one can ensure that $\fdphat_{\batchbh}(t)\leq \alpha$ for all $t$.  But by the sketch given above, this would immediately imply $\fdr\leq\alpha$, as desired. This proof sketch is the core idea behind our algorithms.

It is important to point out that there is not a single way of ensuring $\fdphat_{\batchbh}(t)\leq\alpha$; this approach gives rise to a whole family of algorithms. Naturally, the choice of $\alpha_s$ can be guided by prior knowledge or importance of a given batch, as long as the empirical estimate is controlled under $\alpha$.

%
%

\section{Online Batch FDR Control via \texorpdfstring{$\batchbh$}{Batch BH}}
\label{sec:bh}

In this section, we define the $\batchbh$ class of algorithms and state our main technical result regarding its FDR guarantees. 

\begin{definition}[$\batchbh$]
The $\batchbh$ procedure is any rule for assigning test levels $\alpha_s$ such that
$$\fdphat_{\batchbh}(t) \defn \sum_{s\leq t} \alpha_s \frac{R_s^+}{R_s^+ + \sum_{r\leq t, r\neq s} R_r}$$
is always controlled under a pre-determined level $\alpha$.
\end{definition}

Note that if we were to approximate $R_s^+$ by $R_s$, we would arrive exactly at the estimate derived in the proof sketch of the previous section.

This way of controlling $\fdphat_{\batchbh}(t)$ interpolates between prior offline and online FDR approaches. First, suppose that there is only one batch. Then, the user is free to pick $\alpha_1$ to be any level less than or equal to $\alpha$, in which case it makes sense to simply pick $\alpha$. On the other hand, if every batch is of size one we have $R_s^+ = 1$, hence the FDP estimate reduces to
\begin{align*}
    \fdphat_{\batchbh}(t) &= \sum_{s\leq t} \frac{\alpha_s}{1 + \sum_{r\leq t,r\neq s} R_r} \leq \frac{\sum_{s\leq t} \alpha_s}{\sum_{r\leq t} R_r} \defn \fdphat_{\text{LORD}}(t),
\end{align*}
where the intermediate inequality is almost an equality whenever the total number of rejections is non-negligible. The quantity $\fdphat_{\text{LORD}}(t)$ is an estimate of FDP that is implicitly used in an existing online algorithm known as LORD \citep{javanmard2016online}, as detailed in Ramdas et al.~\cite{RYWJ17}. Thus, $\batchbh$ can be seen as a generalization of both BH and LORD, simultaneously allowing arbitrary batch sizes (like BH) and an arbitrary number of batches (like LORD).


%
%
%
%

We now state our main formal result regarding FDR control of $\batchbh$. As suggested in \secref{notation}, together with the requirement that $\fdphat_{\batchbh}(t)\leq\alpha$ for all $t\in\N$, we also need to guarantee that the procedure is monotone. We present multiple such procedures in the Appendix, and discuss some useful heuristics in \secref{experiments}.

\begin{theorem}
\label{thm:fdrbatchbh}
If all null $p$-values in the sequence are independent of each other and the non-nulls, and the $\batchbh$ procedure is monotone, then it provides anytime $\fdr$ control: for every $t\in\N$, $\fdr(t)\leq\alpha$.
\end{theorem}

We defer the proof of \thmref{fdrbatchbh} to the Appendix.



\section{Online Batch FDR Control via \texorpdfstring{$\batchsbh$}{Batch St-BH}}
\label{sec:sbh}

In addition to the FDR level $\alpha$, the Storey-BH algorithm also requires a user-chosen constant $\lambda\in(0,1)$ as a parameter. This extra parameter allows the algorithm to be more adaptive to the data at hand, constructing a better FDP estimate \citep{Storey02}. We revisit this estimate in the Appendix.

Thus, our extension of Storey-BH, $\batchsbh$, requires a user-chosen constant $\lambda_t\in(0,1)$ as an input to the algorithm at time $t\in\N$. Unless there is prior knowledge of the $p$-value distribution, it is a reasonable heuristic to simply set $\lambda_t = 0.5$ for all $t$. If one suspects, however, that there are many non-nulls which yield strong signal, a smaller value of $\lambda_t$ is appropriate.

Denote by $\max_t\defn \argmax_i\{P_{t,i}:i\in[n_t]\}$ the index corresponding to the maximum $p$-value in batch $t$. With this,
we can define the $\batchsbh$ family of algorithms as follows.

\begin{definition}
\label{def:sbh}
The $\batchsbh$ procedure is any rule for assigning test levels $\alpha_s$, such that
$$\fdphat_{\batchsbh}(t)\defn \sum_{s \leq t} \frac{ \alpha_s R_s^+ \One{P_{s,\max_s}>\lambda_s}}{R_s^+ + \sum_{r\leq t,r\neq s} R_r}$$
is controlled under a pre-determined level $\alpha$.
\end{definition}

Just like $\batchbh$, $\batchsbh$ likewise interpolates between existing offline and online FDR procedures. If there is a single batch of tests, the user can pick the test level $\alpha_1$ to be at most $\alpha$, in which case it makes sense to simply pick $\alpha$. On the other end of the spectrum, in the fully online setting, $\batchsbh$ reduces to the SAFFRON procedure \citep{ramdas2018saffron}. Indeed, the FDP estimate reduces to:
\begin{align*}
	\fdphat_{\batchsbh}(t) &= \sum_{s \leq t} \frac{ \alpha_s \One{P_{s,1}>\lambda_s}}{1 + \sum_{r\leq t,r\neq s} R_r} \leq \frac{ \sum_{s \leq t} \alpha_s \One{P_{s,1}>\lambda_s}}{\sum_{r\leq t} R_r} \defn \fdphat_{\text{SAFFRON}}(t).
	\end{align*}
Since in the original paper, SAFFRON's FDP estimate was written in a slightly different, albeit equivalent form, we point out a subtle difference in the meaning of ``$\alpha_s$'' for Storey-BH and SAFFRON. For SAFFRON, $\alpha_s$ denotes the decision threshold for $P_{s,1}$, while in the batch setting, $\alpha_s$ is the Storey-BH level. If Storey-BH is applied to a single $p$-value under level $\alpha_s$, then it is rejected if and only if $P_{s,1} \leq (1-\lambda_s)\alpha_s$. This difference should be kept in mind when comparing $\fdphat_{\batchsbh}(t)$ to the usual form of $\fdphat_{\text{SAFFRON}}(t)$ (which we review in the Appendix).


%
%

We are now ready to state our main result for $\batchsbh$. Just like $\batchbh$, the $\batchsbh$ procedure requires monotonicity to control the FDR (as per the argument outlined in \secref{notation}). We describe multiple monotone versions of $\batchsbh$ in the Appendix, and discuss some useful heuristics in \secref{experiments}.

\begin{theorem}
\label{thm:fdrbatchsbh}
If the null $p$-values in the sequence are independent of each other and the non-nulls, and the $\batchsbh$ procedure is monotone, then it provides anytime $\fdr$ control: for every $t\in\N$, $\fdr(t)\leq\alpha$.
\end{theorem}

The proof of \thmref{fdrbatchsbh} is presented in the Appendix.


\section{Online Batch FDR Control under Positive Dependence}
\label{sec:prds}

The guarantees of $\batchbh$ and $\batchsbh$ presented thus far relied on independence between $p$-values. In this section we generalize $\batchbh$ to one natural form of dependence, namely \emph{positive dependence} \citep{BY01}. We call this modification $\batchbh^{\text{PRDS}}$, and it controls FDR when the $p$-values in one batch are positively dependent, and independent across batches. Such a setting might occur in multi-armed clinical trials where different treatments are tested against a common control arm \citep{robertson2018online}.

First we establish the definition of positive dependence we consider.

\begin{definition}
	Let $\mathcal{D}\subseteq [0,1]^n$ be any non-decreasing set, meaning that $x\in\mathcal{D}$ implies $y\in\mathcal{D}$, for all $y$ such that $y_i\geq x_i$, for all $i\in[n]$. We say that a vector of $p$-values $\mathbf{P}=(P_1,\dots,P_n)$ satisfies positive regression dependency on a subset (PRDS), or positive dependence for short, if for any null index $i\in\nulls$ and arbitrary non-decreasing set $\mathcal{D}\subseteq [0,1]^n$, the function $t\mapsto \PPst{P\in\mathcal{D}}{P_i\leq t}$ is non-decreasing over $t\in(0,1]$.
	\end{definition}
	
This definition has been a common formulation of positive dependence in prior FDR works, e.g. \citep{BY01, blanchard2008two, ramdas2019unified}. Clearly, independent $p$-values satisfy PRDS. A non-trivial example is given for Gaussian observations. Suppose $\mathbf{P} = (\Phi(Z_1),\dots,\Phi(Z_n))$, where $(Z_1,\dots,Z_n)$ is a multivariate Gaussian vector with covariance matrix $\Sigma$. Then, $\mathbf{P}$ satisfies PRDS if and only if $\Sigma_{i,j}\geq 0$ for all $i\in\nulls$ and $j\in[n]$.

 Now we are ready to define the FDP estimate of $\batchbh^{\text{PRDS}}$.

\begin{definition}
\label{def:batchbhprds}
The $\batchbh^{\text{PRDS}}$ procedure is any rule for assigning test levels $\alpha_t$ such that
$$\fdphat_{\batchbh^{\text{PRDS}}}(t) = \sum_{s\leq t}\alpha_s \frac{n_s}{n_s + \sum_{r<s}R_r}$$
is controlled under $\alpha$ for all $t\in\N$.
\end{definition}

Below is an example update rule that satisfies \defref{batchbhprds}.

\begin{algorithm}[H]
\SetAlgoLined
\SetKwInOut{Input}{Input}
\Input{FDR level $\alpha$, non-negative sequence $\{\gamma_s\}_{s=1}^\infty$ such that $\sum_{s=1}^\infty \gamma_s = 1$.}
Set $\alpha_1 = \gamma_1 \alpha$;\newline
 \For{$t=1,2,\dots$}{
 Run the BH procedure under level $\alpha_t$ on batch $\mathbf{P}_t$;\newline
Set $\alpha_{t+1} = \alpha \frac{\gamma_{t+1}}{n_{t+1}} (n_{t+1} + \sum_{s=1}^t R_s)$;
 }
 \caption{The $\batchbh^{\text{PRDS}}$ algorithm}
 \label{alg:batchprds}
\end{algorithm}

We state our main FDR guarantees for $\batchbh^{\text{PRDS}}$ below. Our proof relies on a ``super-uniformity lemma'', similar to several lemmas in prior work that consider PRDS $p$-values \citep{blanchard2008two,BY01,ramdas2019unified}. We prove both this lemma and \thmref{fdrPRDS} later in the Appendix.

\begin{theorem}
\label{thm:fdrPRDS}
Suppose that every batch of $p$-values $\mathbf{P}_t$ satisfies PRDS, and additionally that $P_{t,i}$ and $\{\mathbf{P}_s:s\in \mathcal{I}\}$ are independent whenever $t\not\in\mathcal{I}$, for all $i\in\nulls_t$. Then, the $\batchbh^{\text{PRDS}}$ procedure provides anytime $\fdr$ control: for every $t\in\N$, $\fdr(t)\leq\alpha$.
\end{theorem}

In other words, $\batchbh^{\text{PRDS}}$ ensures FDR control when $p$-values are independent across different batches, and positively dependent within each batch. \thmref{fdrPRDS} is a generalization of an earlier result which states that the BH algorithm controls FDR under PRDS \cite{BY01}.

In online FDR control, handling dependence has generally proved challenging. \citet{javanmard2016online} have proposed procedures which control the FDR under arbitrary dependence, however their updates imply an essentially alpha-spending (online Bonferroni) type correction which controls a more stringent criterion called the family-wise error rate \citep{gordon1983discrete}. Their earlier algorithm called LOND \cite{javanmard2015online} was recently proved to control the FDR under PRDS \cite{zrnic2018asynchronous}, and is a more powerful alternative for the fully online setting than the arbitrary depenence procedure. Indeed, $\batchbh^{\text{PRDS}}$ is a minibatch generalization of the LOND algorithm. Finally, it is worth pointing out that the notion of positive dependence we consider in this paper resembles local dependence proposed by \citet{zrnic2018asynchronous}, although their solutions only control modified FDR (mFDR).


\section{Numerical Experiments}
\label{sec:experiments}

We compare the performance of $\batchbh$ and $\batchsbh$ with two state-of-the-art online FDR algorithms: LORD \citep{javanmard2016online, RYWJ17} and SAFFRON \citep{ramdas2018saffron}. Specifically, we analyze the achieved power and FDR of the compared methods. In this section, we run the procedures on synthetic data, while in \secref{credit} we study a real fraud detection data set.

As explained in prior literature \citep{ramdas2018saffron}, LORD and BH are non-adaptive methods, while SAFFRON and Storey-BH adapt to the tested $p$-values through the parameter $\lambda_t$. For this reason, we keep comparisons fair by comparing $\batchbh$ with LORD, and $\batchsbh$ with SAFFRON.

The choice of $\lambda_t$ should generally depend on the number and strength of non-null $p$-values the analyst expects to see in the sequence. As suggested in previous works on similar adaptive methods \citep{Storey02,Storey04, ramdas2018saffron}, it is reasonable to set $\lambda_t \equiv 0.5$ if no prior knowledge is assumed.

As discussed in \secref{notation}, there are different ways of assigning $\alpha_i$ such that the appropriate FDP estimate is controlled under $\alpha$. Moreover, as we argued in \secref{bh} and \secref{sbh}, this needs to be done in a monotone way to guarantee FDR control for an arbitrary $p$-value distribution. In the experimental sections of this paper, however, we resort to a heuristic. Enforcing monotonicity uniformly across all distributions diminishes the power of FDR methods. Hence, we apply algorithms which control the corresponding FDP estimates and are expected to be monotone under natural $p$-value distributions, however possibly not for adversarially chosen ones. In the Appendix we test the monotonicity of these procedures empirically, and demonstrate that it is indeed satisfied with overwhelming probability. We now present the specific algorithms that we studied.

\begin{algorithm}[H]
\SetAlgoLined
\SetKwInOut{Input}{Input}
\Input{FDR level $\alpha$, non-negative sequence $\{\gamma_s\}_{s=1}^\infty$ such that $\sum_{s=1}^\infty \gamma_s = 1$.}
Set $\alpha_1 = \gamma_1 \alpha$;\newline
 \For{$t=1,2,\dots$}{
 Run the BH method at level $\alpha_t$ on batch $\mathbf{P}_t$;\newline
  Set $\alpha_{t+1} = \left(\sum_{s\leq t+1}\gamma_s \alpha - \sum_{s\leq t} \alpha_s \frac{R_s^+}{R_s^+ + \sum_{r\neq s, r\leq t} R_r}\right)\frac{n_{t+1} + \sum_{s\leq t} R_s}{n_{t+1}}$;
 }
 \caption{The $\batchbh$ algorithm}
\label{alg:defaultbatchbh}
\end{algorithm}

\begin{algorithm}[H]
\SetAlgoLined
\SetKwInOut{Input}{Input}
\Input{FDR level $\alpha$, non-negative sequence $\{\gamma_s\}_{s=1}^\infty$ such that $\sum_{s=1}^\infty \gamma_s = 1$.}
Set $\alpha_1 = \gamma_1 \alpha$;\newline
 \For{$t=1,2,\dots$}{
 Run the Storey-BH procedure at level $\alpha_t$ with parameter $\lambda_t$ on batch $\mathbf{P}_t$;\newline
  Set $\alpha_{t+1} = \left(\sum_{s\leq t+1} \gamma_s \alpha - \sum_{s\leq t} \alpha_s\One{P_{s,\max_s}>\lambda_s} \frac{R_s^+}{R_s^+ + \sum_{r\neq s, r\leq t} R_r}\right)\frac{n_{t+1} + \sum_{s\leq t} R_s}{n_{t+1}}$;
 }
 \caption{The $\batchsbh$ algorithm}
 \label{alg:defaultbatchsbh}
\end{algorithm}

The reason why we add a sequence $\{\gamma_s\}_{s=1}^\infty$ as a hyperparameter is to prevent $\alpha_t$ from vanishing. If we immediately invest the whole error budget $\alpha$, i.e. we set $\gamma_1 = 1$ and $\gamma_s = 0, s\neq 1$, then $\alpha_t$ might be close to 0 for small batches, given that $R_t^+$ could be close to $n_t$. For this reason, for the smallest batch size we consider (which is 10), we pick $\gamma_s \propto s^{-2}$. Similar error budget investment strategies have been considered in prior work \cite{ramdas2018saffron, tian2019addis}. For larger batch sizes, $R_t^+$ is generally much smaller than $n_t$, so for all other batch sizes we invest more aggressively by picking $\gamma_1 = \gamma_2 = \frac{1}{2}$, $\gamma_s = 0$, $s\not\in \{1,2\}$. This is analogous to the default choice of ``initial wealth'' for LORD and SAFFRON of $\frac{\alpha}{2}$, which we also use in our experiments. We only adapt our choice of $\{\gamma_s\}_{s=1}^\infty$ to the batch size, as that is information available to the scientist. In general, one can achieve better power if $\{\gamma_s\}_{s=1}^\infty$ is tailored to parameters such as the number of non-nulls and their strength, but given that such information is typically unknown, we keep our hyperparameters agnostic to such specifics.

In the Appendix we prove \factref{defaultbatchbhvalid}, which states the Algorithm \ref{alg:defaultbatchbh} controls the appropriate FDP estimate. We omit the analogous proof for Algorithm \ref{alg:defaultbatchsbh} due to the similarity of the two proofs.

\begin{fact}
\label{fact:defaultbatchbhvalid}
The algorithm given in Algorithm \ref{alg:defaultbatchbh} maintains $\fdphat_{\batchbh}(t)\leq\alpha$.
\end{fact}

We test for the means of a sequence of $T=3000$ independent Gaussian observations. Under the null, the mean is $\mu_0=0$. Under the alternative, the mean is $\mu_1$, whose distribution differs in two settings that we studied. For each index $i\in\{1,\dotsc,T\}$, the observation $Z_i$ is distributed according to
\begin{align*}
    Z_i \sim
    \begin{cases}
        N(\mu_0, 1), \text{with probability } 1-\pi_1, \\
        N(\mu_1, 1), \text{with probability } \pi_1.
    \end{cases}
\end{align*}

In all experiments we set $\alpha=0.05$. All plots display the average and one standard deviation around the average of power or FDR, against $\pi_1\in\{0.01,0.02,\dots,0.09\}\cup\{0.1,0.2,0.3,0.4,0.5\}$ (interpolated for in-between values). All quantities are averaged over 500 independent trials.

\subsection{Constant Gaussian Means}
\label{sec:mean3}

In this setting, we choose the mean under the alternative to be constant, $\mu_1=3$. Each observation is converted to a one-sided $p$-value as $P_i=\Phi(-Z_i)$, where $\Phi$ is the standard Gaussian CDF.

\paragraph{Non-adaptive procedures.}

\figref{mean3_bbh_pi1s} (left) compares the statistical power and FDR of $\batchbh$ and LORD as functions of $\pi_1$. Across almost all values of $\pi_1$, the online batch procedures outperform LORD, with the exception of $\batchbh$ with the smallest considered batch size, for small values of $\pi_1$.

\paragraph{Adaptive procedures.}

\figref{mean3_bbh_pi1s} (right) compares the statistical power and FDR of $\batchsbh$ and SAFFRON as functions of $\pi_1$. The online batch procedures dominate SAFFRON for all values of $\pi_1$. The difference in power is especially significant for $\pi_1 \leq 0.1$, which is a reasonable range for the non-null proportion in most real-world applications.

\begin{figure}[h]
  \centering
  \includegraphics[width=0.49\columnwidth]{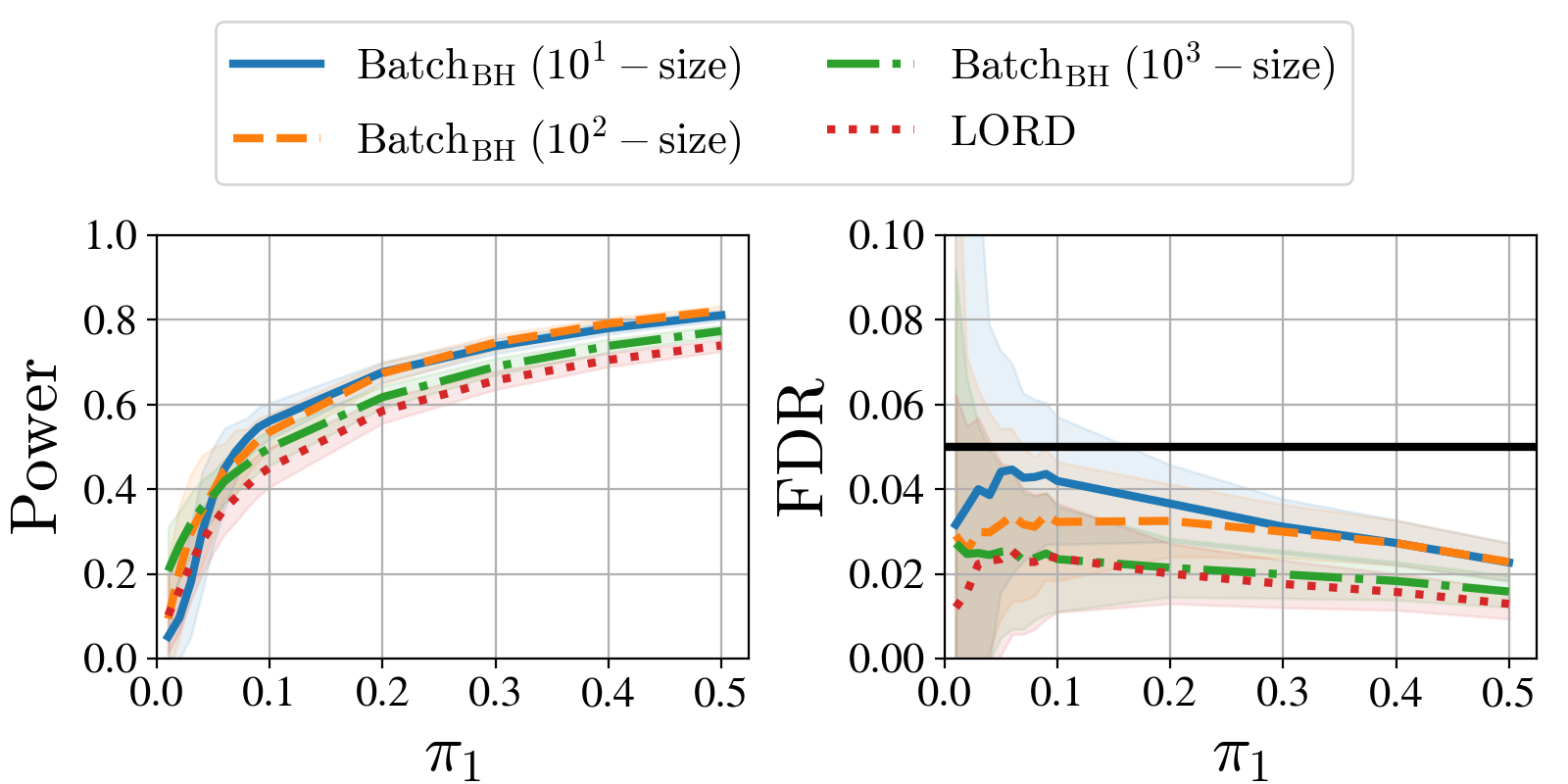}
  \includegraphics[width=0.49\columnwidth]{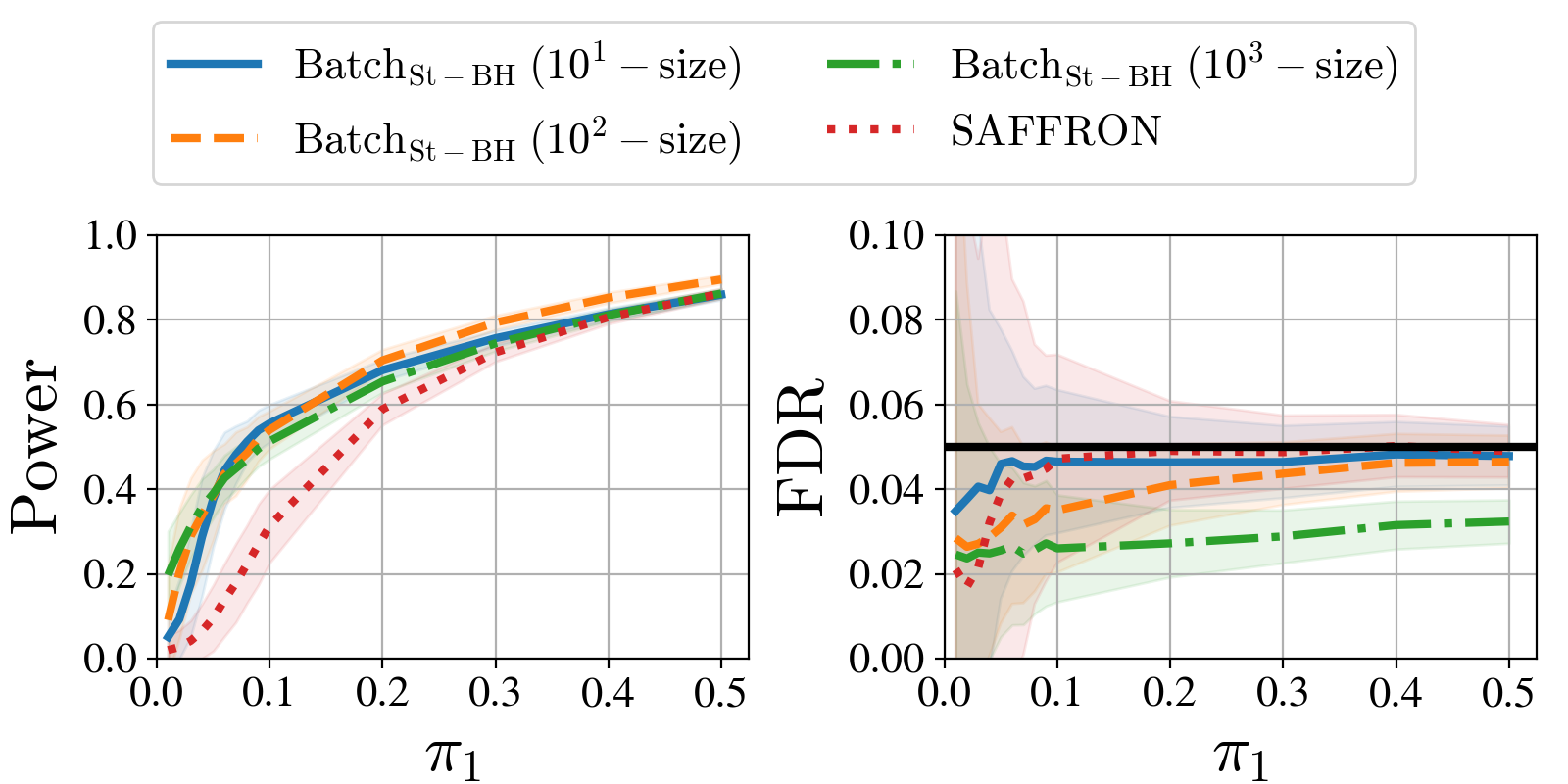}
  \caption{Statistical power and FDR versus probability of non-null hypotheses $\pi_1$ for $\batchbh$ (at batch sizes 10, 100, and 1000) and LORD (left), and $\batchsbh$ (at batch sizes 10, 100, and 1000) and SAFFRON (right). The observations under the null are $N(0,1)$, and the observations under the alternative are $N(3,1)$.}
  \label{fig:mean3_bbh_pi1s}
\end{figure}

\paragraph{Naively composed procedures.} \figref{mean3_bh_pi1s} shows the statistical power and FDR versus $\pi_1$ for BH (left) and Storey-BH (right) naively run in a batch setting where each individual batch is run using test level $\alpha = 0.05$. Although there is a significant boost in power, the FDR is generally much higher than the desired value for reasonably small $\pi_1$; this is not true of batch size 1000 because only 3 batches are composed, where we know that in the worst case $\fdr\leq 3\alpha$.

\begin{figure}[h]
  \centering
  \includegraphics[width=0.49\columnwidth]{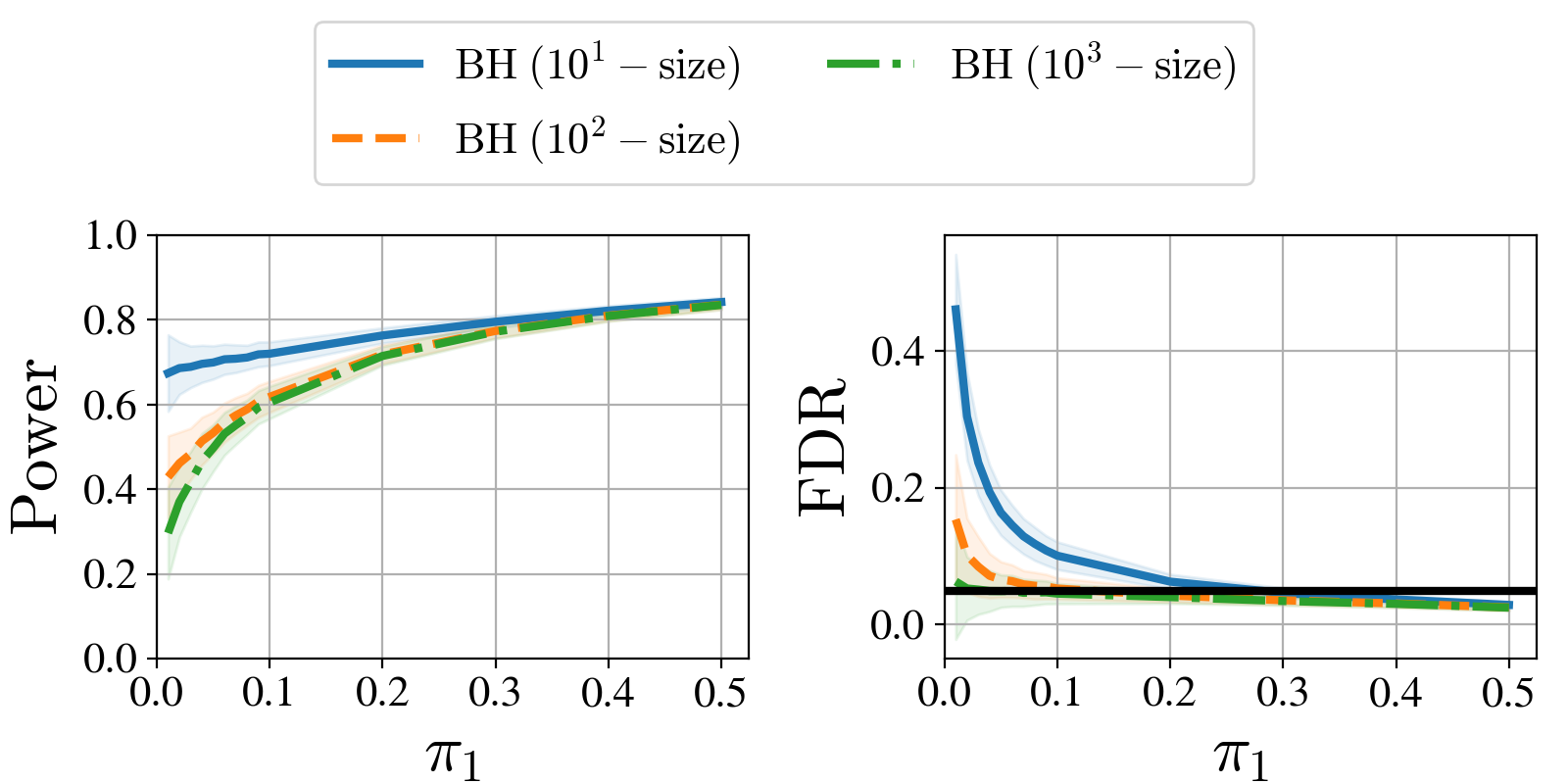}
  \includegraphics[width=0.49\columnwidth]{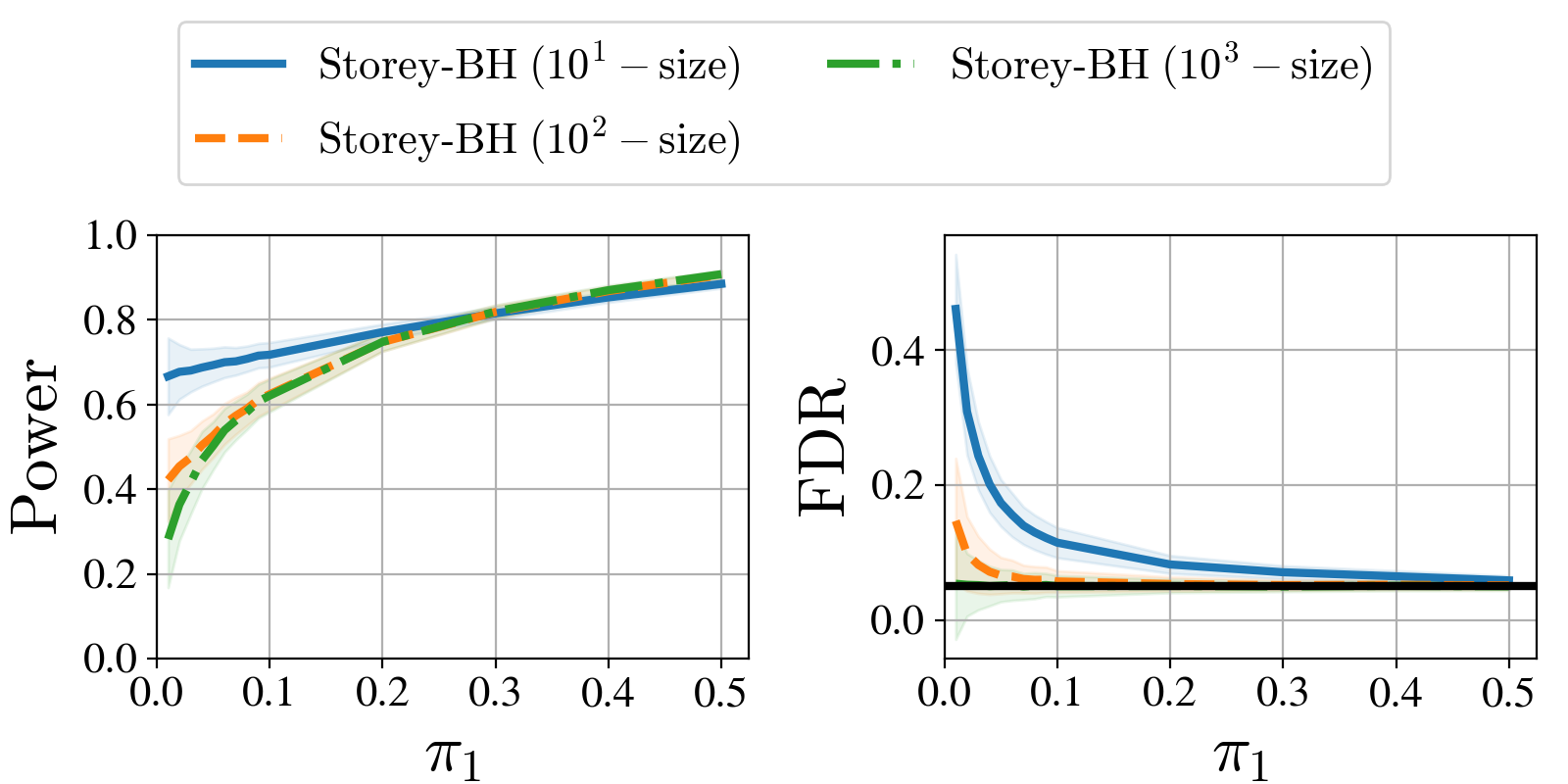}
  \caption{Statistical power and FDR versus probability of non-null hypotheses $\pi_1$ for naively composed BH and Storey-BH (at batch sizes 10, 100, and 1000). The observations under the null are $N(0,1)$, and the observations under the alternative are $N(3,1)$.}
  \label{fig:mean3_bh_pi1s}
\end{figure}

\subsection{Random Gaussian Alternative Means}
\label{sec:mean0}

Now we consider random alternative means; we let $\mu_1\sim N(0, 2\log T)$. Unlike the previous setting, this is a hard testing problem in which non-nulls are barely detectable \citep{javanmard2016online}. Each observation is converted to a two-sided $p$-value as $P_i=2\Phi(-|Z_i|)$, where $\Phi$ is the standard Gaussian CDF.

\paragraph{Non-adaptive procedures.}

\figref{mean0_bbh_pi1s} (left) compares the statistical power and FDR of $\batchbh$ and LORD as functions of $\pi_1$. Again, for most values of $\pi_1$ all batch procedures outperform LORD.

\paragraph{Adaptive procedures.}

\figref{mean0_bbh_pi1s} (right) compares the statistical power and FDR of $\batchsbh$ and SAFFRON as functions of $\pi_1$. For high values of $\pi_1$, all procedures behave similarly, while for small values of $\pi_1$ the batch procedures dominate.

\begin{figure}[h]
  \centering
  \includegraphics[width=0.49\columnwidth]{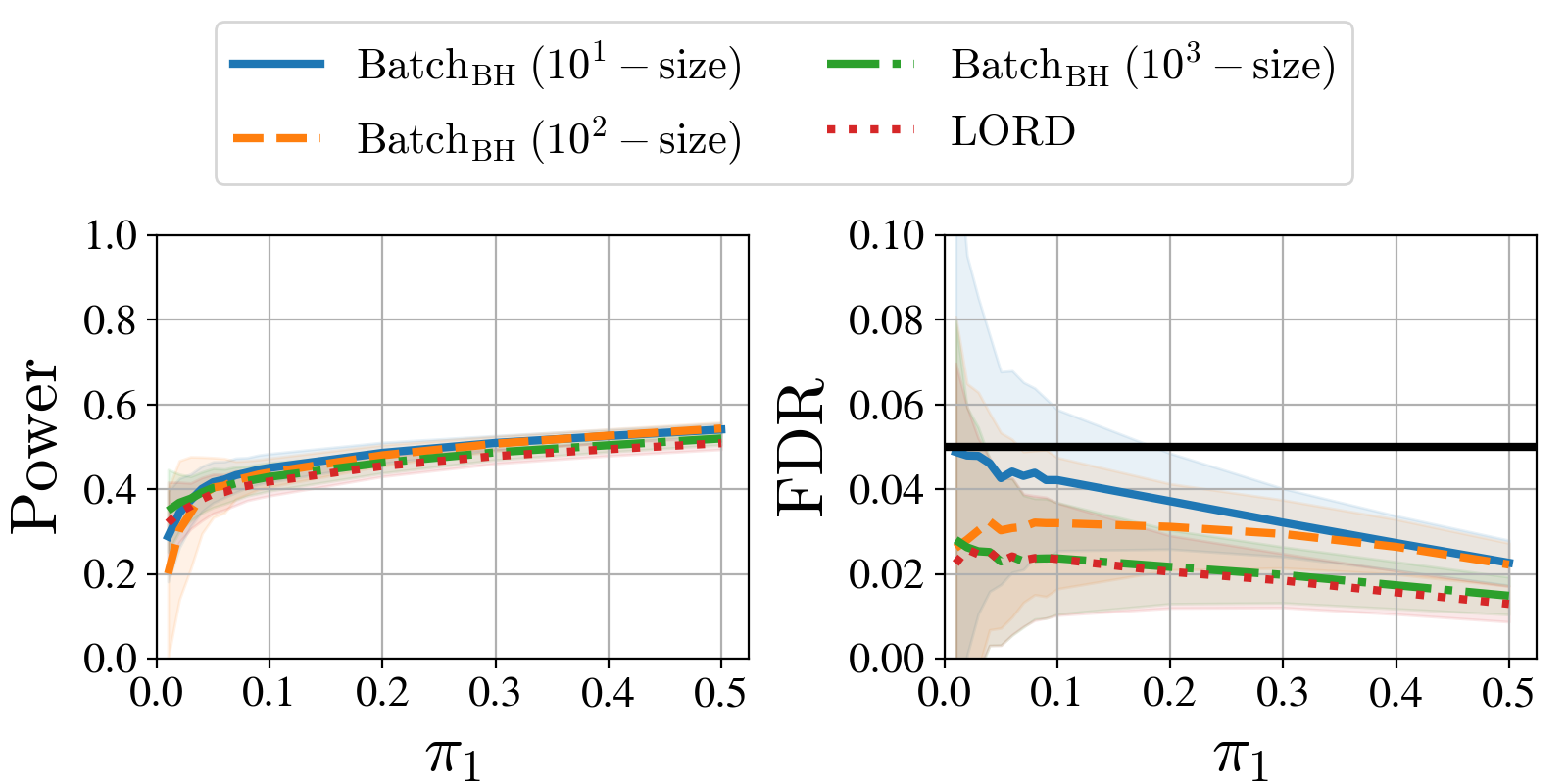}
  \includegraphics[width=0.49\columnwidth]{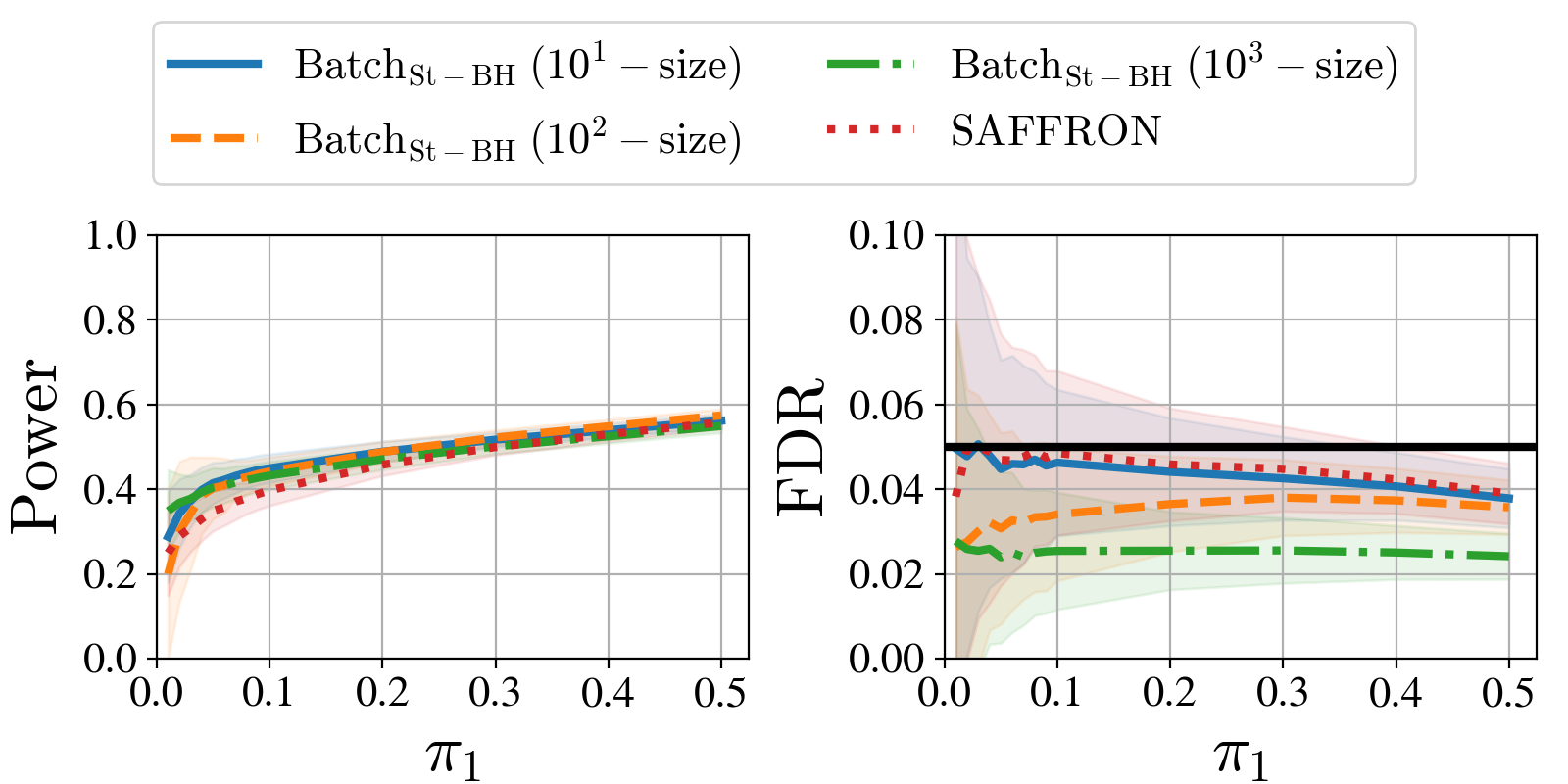}
  \caption{Statistical power and FDR versus probability of non-null hypotheses $\pi_1$ for $\batchbh$ (at batch sizes 10, 100, and 1000) and LORD (left), and $\batchsbh$ (at batch sizes 10, 100, and 1000) and SAFFRON (right). The observations under the null are $N(0,1)$, and the observations under the alternative are $N(\mu_1,1)$ where $\mu_1\sim N(0,2\log T)$.}
  \label{fig:mean0_bbh_pi1s}
\end{figure}

\paragraph{Naively composed procedures.} \figref{mean0_bh_pi1s} shows the statistical power and FDR versus $\pi_1$ for BH (left) and Storey-BH (right) naively run in a batch setting where each individual batch is run using test level $\alpha = 0.05$. In this hard testing problem, there is not as much gain in power, and the FDR is extremely high, as expected.

\begin{figure}[h]
  \centering
  \includegraphics[width=0.49\columnwidth]{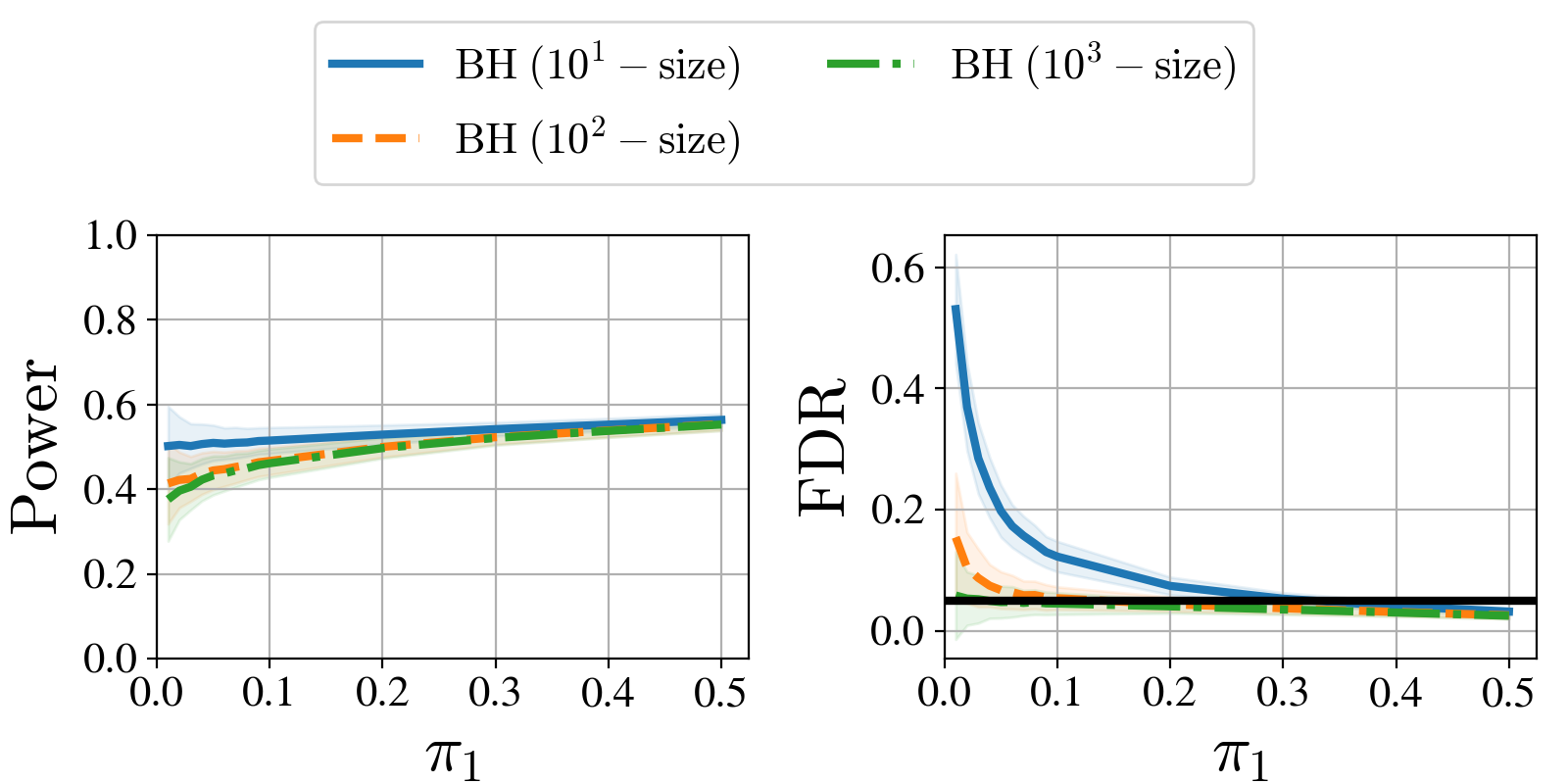}
  \includegraphics[width=0.49\columnwidth]{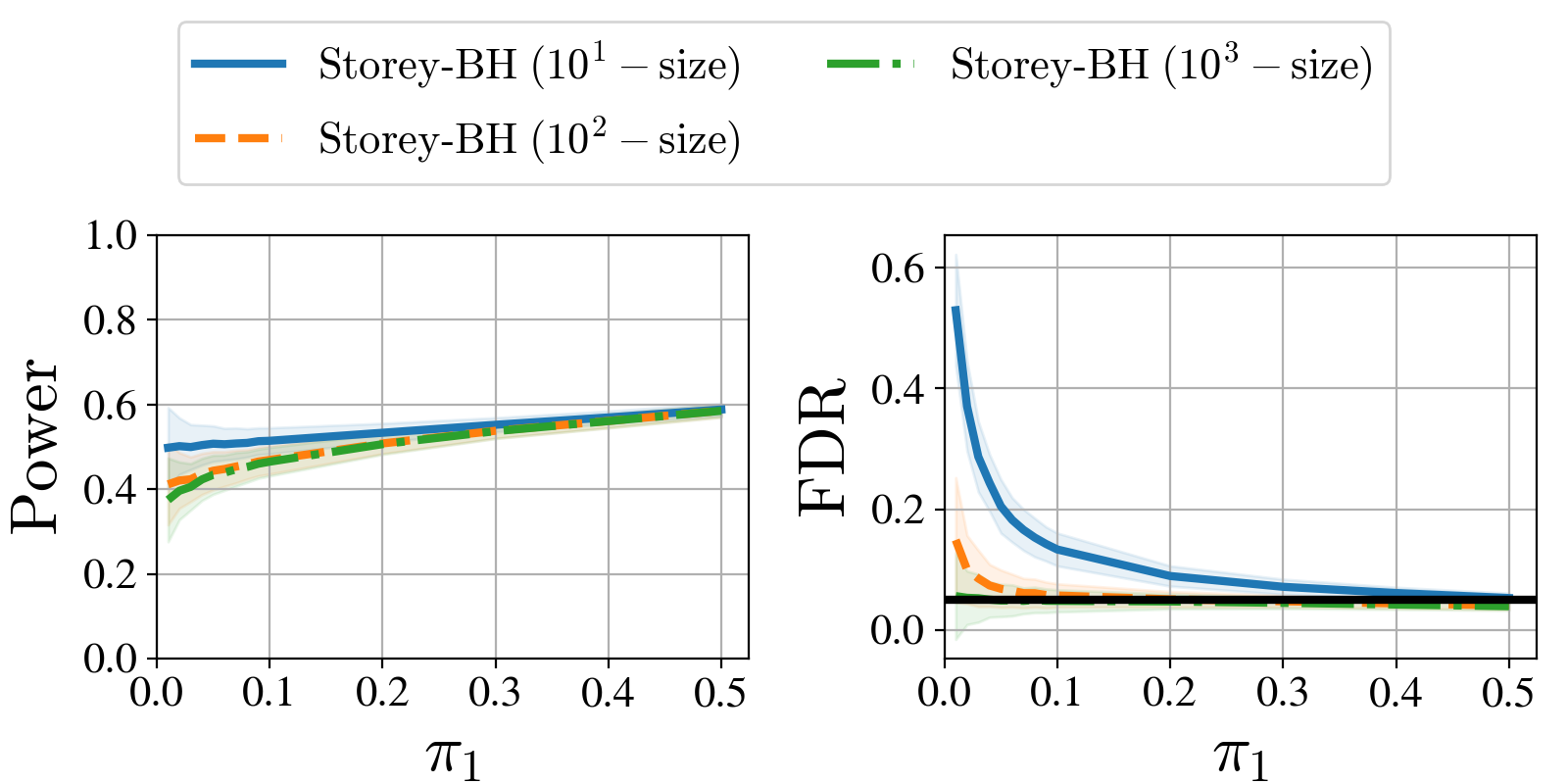}
  \caption{Statistical power and FDR versus probability of non-null hypotheses $\pi_1$ for naively composed BH and Storey-BH (at batch sizes 10, 100, and 1000). The observations under the null are $N(0,1)$, and the observations under the alternative are $N(\mu_1,1)$ where $\mu_1\sim N(0,2\log T)$.
  }
  \label{fig:mean0_bh_pi1s}
\end{figure}

\section{Detecting Credit Card Fraud}
\label{sec:credit}

We apply our algorithms to real credit card transaction data. Credit card companies test for whether transactions are fraudulent; if the transactions are deemed to be fraudulent, they are denied. However, it is important to control the proportion of transactions that are falsely identified as fraudulent, as these false identifications inconvenience users by declining legitimate transactions.

We use a dataset released by the Machine Learning Group of Universit\'e Libre de Bruxelles for a Kaggle competition\footnote{https://www.kaggle.com/mlg-ulb/creditcardfraud}\citep{pozzolo2015fraud}. The dataset comprises 492 fraudulent transactions and 284,315 legitimate transactions. For each transaction, the null hypothesis is that the transaction is not fraudulent, which means that the proportion of non-nulls $\pi_1$ is approximately 0.173\%. Such asymmetry between the proportion of nulls and non-nulls is typical in applications of FDR methods.

Each transaction in the dataset has 28 principal component analysis (PCA) features, the monetary value of the transaction, and a binary label indicating whether the transaction is fraudulent. The PCA features are provided instead of the original features for confidentiality. For each transaction $i$, let $y_i\in\{0, 1\}$ denote whether the transaction is fraudulent ($y_i=1$ denotes a fraudulent transaction) and let $x_i\in \R^{29}$ denote the vector of the transaction's PCA features and monetary value.

In a similar fashion to \citet{javanmard2016online}, we randomly partition the transactions into the subsets Train1 (60\% of the transactions), Train2 (20\% of the transactions), and Test (20\% of the transactions). We then fit a logistic regression model to Train1. In particular, for $i$ in Train1, we model the probability that transaction $i$ is fraudulent as
\begin{align*}
    \PP{Y_i=1 | X_i=x_i} &= \sigma(\theta^T x_i),
\end{align*}
where $\sigma(x)=\frac{1}{1+\text{e}^{-x}}$.

For each $i$ in Train2 and each $j$ in Test, we compute $q_i=\sigma(\theta^T x_i)$ and $q_j^{\text{Test}}=\sigma(\theta^T x_j)$. Let $T_0$ denote the subset of Train2 that are non-fraudulent transactions. We construct the $p$-value $P_j$ as
\begin{align*}
    P_j = \frac{1}{n_0} |\{ i\in T_0 : q_i > q_j^{\text{Test}} \}|.
\end{align*}

We compare hypothesis testing procedures on the $p$-values of the Test subset. We set $\alpha=0.1$ and we set all other hyperparameters the same way as in previous experiments. We use 100 random splits of the transactions into Train1, Train2, and Test in order to compute the average and one standard deviation around the average of power and FDR.

For both the non-adaptive and the adaptive methods, we observe higher power for online batch procedures than for standard online procedures, across several batch sizes of different orders of magnitude. Our findings are summarized in \tabref{fraud_bbh} and \tabref{fraud_bsbh}. However, as observed in experiments on synthetic data as well, we do not observe a monotone relationship between batch size and power; we return to this point in the discussion below.


\begin{table}[h]
\centering
\caption{Non-Adaptive Algorithms on Real Data}
\vspace{.1in}
\begin{tabular}{@{}lll@{}}
	\toprule
	{} &              Power &                FDR \\
	\midrule
	$\batchbh$ ($10^1$-size) &  0.242 $\pm$ 0.053 &  0.126 $\pm$ 0.075 \\
	$\batchbh$ ($10^2$-size) &  0.299 $\pm$ 0.100 &  0.102 $\pm$ 0.067 \\
	$\batchbh$ ($10^3$-size) &  0.260 $\pm$ 0.086 &  0.082 $\pm$ 0.064 \\
	LORD                     &  0.231 $\pm$ 0.051 &  0.082 $\pm$ 0.067 \\
	\bottomrule
\end{tabular}
\label{tab:fraud_bbh}
\end{table}


\begin{table}[h]
\centering
\caption{Adaptive Algorithms on Real Data}
\vspace{.1in}
\begin{tabular}{@{}lll@{}}
	\toprule
	{} &              Power &                FDR \\
	\midrule
	$\batchsbh$ ($10^1$-size) &  0.240 $\pm$ 0.052 &  0.125 $\pm$ 0.081 \\
	$\batchsbh$ ($10^2$-size) &  0.291 $\pm$ 0.098 &  0.096 $\pm$ 0.065 \\
	$\batchsbh$ ($10^3$-size) &  0.246 $\pm$ 0.075 &  0.074 $\pm$ 0.063 \\
	SAFFRON                   &  0.211 $\pm$ 0.041 &  0.137 $\pm$ 0.086 \\
	\bottomrule
\end{tabular}
\label{tab:fraud_bsbh}
\end{table}


\section{Discussion}
\label{sec:discussion}

In this paper, we have presented algorithms for FDR control in online batch settings; at every time step, a batch of decisions is made via the BH or Storey-BH algorithm, and batches arrive sequentially, in a stream. We discuss several possible extensions of this framework.

\paragraph{Alpha-investing version of $\batchsbh$.} In the definition of $\batchsbh$, we considered deterministic values of $\lambda_t$ for simplicity. By imposing a monotonicity constraint on $\lambda_t$ \citep{ramdas2018saffron}, one could generalize $\batchsbh$ to handle random $\lambda_t$ as well. In particular, this would lead to a batch generalization of alpha-investing \citep{foster2008alpha}, in which $\lambda_t = \alpha_t$.

\paragraph{Asynchronous online batch testing.} \citet{zrnic2018asynchronous} consider the setting of asynchronous online testing, in which one conducts a possibly infinite number of sequential experiments which could, importantly, be running in parallel. They generalize multiple online FDR algorithms to handle this so-called \emph{asynchronous testing} problem. Using their technical tools, namely the idea of conflict sets, one can adjust $\batchbh$ and $\batchsbh$ to operate in an asynchronous manner.

\paragraph{ADDIS algorithm.} \citet{tian2019addis} have presented an adaptive online FDR algorithm called ADDIS that
was designed with the goal of improving the power of online FDR methods when the null $p$-values are conservative. The same paper also gives the offline analog of ADDIS. Using our proof technique, one can design online batch corrections for the offline counterpart of ADDIS, thus interpolating between the two algorithms of Tian and Ramdas.

\paragraph{Batch size versus power.} As our experiments indicate, it is not clear that bigger batch sizes give better power. Intuitively, if a batch is very large, say of size $n$, the slope $\alpha/n$ of the BH procedure is very conservative, and it might be better to split up the tests into multiple batches. It would be of great importance for the practitioner to conduct a rigorous analysis of the relationship between batch size and power.

\paragraph{mFDR control.} Many treatments of online FDR have focused on mFDR guarantees (together with FDR guarantees), mostly due to simplicity of the proofs, but also because mFDR can be a reasonable error metric in some settings. Indeed, in the online batch setting, mFDR is potentially a reasonable target measure, because mFDR, unlike FDR, is preserved under composition; if two disjoint batches of tests are guaranteed to achieve $\mfdr\leq \alpha$, pooling their results also ensures $\mfdr\leq \alpha$. This favorable property has been recognized in prior work \citep{van2008controlling}. Unfortunately, the BH algorithm controls $\mfdr$ only asymptotically \citep{genovese2002operating, sun2007oracle}. Moreover, how closely it controls mFDR depends on its ``stability,'' as we show in the Appendix. In fact it has been noted that BH is not stable \citep{gordon2007control}, making FDR our preferred choice of metric.

\subsection*{Acknowledgments}
AR thanks Adel Javanmard for a discussion during the early phases of this work. We thank David Robertson for pointing out a simplification of the $\batchsbh$ FDP estimate from an earlier version of this manuscript.

\bibliography{FDR}

\begin{thebibliography}{23}
\providecommand{\natexlab}[1]{#1}
\providecommand{\url}[1]{\texttt{#1}}
\expandafter\ifx\csname urlstyle\endcsname\relax
  \providecommand{\doi}[1]{doi: #1}\else
  \providecommand{\doi}{doi: \begingroup \urlstyle{rm}\Url}\fi

\bibitem[Aharoni and Rosset(2014)]{aharoni2014generalized}
Ehud Aharoni and Saharon Rosset.
\newblock Generalized $\alpha$-investing: definitions, optimality results and
  application to public databases.
\newblock \emph{Journal of the Royal Statistical Society, Series B (Statistical
  Methodology)}, 76\penalty0 (4):\penalty0 771--794, 2014.

\bibitem[Benjamini and Hochberg(1995)]{BH95}
Yoav Benjamini and Yosef Hochberg.
\newblock Controlling the false discovery rate: a practical and powerful
  approach to multiple testing.
\newblock \emph{Journal of the Royal Statistical Society, Series B
  (Methodological)}, 57\penalty0 (1):\penalty0 289--300, 1995.

\bibitem[Benjamini and Yekutieli(2001)]{BY01}
Yoav Benjamini and Daniel Yekutieli.
\newblock The control of the false discovery rate in multiple testing under
  dependency.
\newblock \emph{The Annals of Statistics}, 29\penalty0 (4):\penalty0
  1165--1188, 2001.

\bibitem[Blanchard and Roquain(2008)]{blanchard2008two}
Gilles Blanchard and Etienne Roquain.
\newblock Two simple sufficient conditions for {FDR} control.
\newblock \emph{Electronic journal of Statistics}, 2:\penalty0 963--992, 2008.

\bibitem[Dal~Pozzolo et~al.(2015)Dal~Pozzolo, Caelen, Johnson, and
  Bontempi]{pozzolo2015fraud}
Andrea Dal~Pozzolo, Olivier Caelen, Reid~A Johnson, and Gianluca Bontempi.
\newblock Calibrating probability with undersampling for unbalanced
  classification.
\newblock In \emph{2015 IEEE Symposium Series on Computational Intelligence},
  pages 159--166, 2015.

\bibitem[Foster and Stine(2008)]{foster2008alpha}
Dean Foster and Robert Stine.
\newblock $\alpha$-investing: a procedure for sequential control of expected
  false discoveries.
\newblock \emph{Journal of the Royal Statistical Society, Series B (Statistical
  Methodology)}, 70\penalty0 (2):\penalty0 429--444, 2008.

\bibitem[Genovese and Wasserman(2002)]{genovese2002operating}
Christopher Genovese and Larry Wasserman.
\newblock Operating characteristics and extensions of the false discovery rate
  procedure.
\newblock \emph{Journal of the Royal Statistical Society: Series B (Statistical
  Methodology)}, 64\penalty0 (3):\penalty0 499--517, 2002.

\bibitem[Gordon et~al.(2007)Gordon, Glazko, Qiu, and
  Yakovlev]{gordon2007control}
Alexander Gordon, Galina Glazko, Xing Qiu, and Andrei Yakovlev.
\newblock Control of the mean number of false discoveries, bonferroni and
  stability of multiple testing.
\newblock \emph{The Annals of Applied Statistics}, 1\penalty0 (1):\penalty0
  179--190, 2007.

\bibitem[Gordon~Lan and DeMets(1983)]{gordon1983discrete}
KK~Gordon~Lan and David~L DeMets.
\newblock Discrete sequential boundaries for clinical trials.
\newblock \emph{Biometrika}, 70\penalty0 (3):\penalty0 659--663, 1983.

\bibitem[Javanmard and Montanari(2015)]{javanmard2015online}
Adel Javanmard and Andrea Montanari.
\newblock On online control of false discovery rate.
\newblock \emph{arXiv preprint arXiv:1502.06197}, 2015.

\bibitem[Javanmard and Montanari(2018)]{javanmard2016online}
Adel Javanmard and Andrea Montanari.
\newblock Online rules for control of false discovery rate and false discovery
  exceedance.
\newblock \emph{The Annals of Statistics}, 46\penalty0 (2):\penalty0 526--554,
  2018.

\bibitem[Kohavi and Longbotham(2017)]{kohavi2017online}
Ron Kohavi and Roger Longbotham.
\newblock Online controlled experiments and a/b testing.
\newblock \emph{Encyclopedia of machine learning and data mining}, pages
  922--929, 2017.

\bibitem[Ramdas et~al.(2017)Ramdas, Yang, Wainwright, and Jordan]{RYWJ17}
Aaditya Ramdas, Fanny Yang, Martin Wainwright, and Michael Jordan.
\newblock Online control of the false discovery rate with decaying memory.
\newblock In \emph{Advances In Neural Information Processing Systems}, pages
  5655--5664, 2017.

\bibitem[Ramdas et~al.(2018)Ramdas, Zrnic, Wainwright, and
  Jordan]{ramdas2018saffron}
Aaditya Ramdas, Tijana Zrnic, Martin Wainwright, and Michael Jordan.
\newblock {SAFFRON}: an adaptive algorithm for online control of the false
  discovery rate.
\newblock In \emph{Proceedings of the 35th International Conference on Machine
  Learning}, pages 4286--4294, 2018.

\bibitem[Ramdas et~al.(2019)Ramdas, Barber, Wainwright, and
  Jordan]{ramdas2019unified}
Aaditya~K Ramdas, Rina~F Barber, Martin~J Wainwright, and Michael~I Jordan.
\newblock A unified treatment of multiple testing with prior knowledge using
  the p-filter.
\newblock \emph{The Annals of Statistics}, 47\penalty0 (5):\penalty0
  2790--2821, 2019.

\bibitem[Robertson and Wason(2018)]{robertson2018online}
David~S Robertson and James Wason.
\newblock Online control of the false discovery rate in biomedical research.
\newblock \emph{arXiv preprint arXiv:1809.07292v1}, 2018.

\bibitem[Robertson et~al.(2019)Robertson, Liou, Ramdas, Javanmard, Tian, Zrnic,
  and Karp]{onlineFDR}
David~S. Robertson, Lathan Liou, Aaditya Ramdas, Adel Javanmard, Jinjin Tian,
  Tijana Zrnic, and Natasha~A. Karp.
\newblock \emph{onlineFDR: Online error control}, 2019.
\newblock R package 1.9.2.

\bibitem[Storey(2002)]{Storey02}
John Storey.
\newblock A direct approach to false discovery rates.
\newblock \emph{Journal of the Royal Statistical Society, Series B (Statistical
  Methodology)}, 64\penalty0 (3):\penalty0 479--498, 2002.

\bibitem[Storey et~al.(2004)Storey, Taylor, and Siegmund]{Storey04}
John Storey, Jonathan Taylor, and David Siegmund.
\newblock Strong control, conservative point estimation and simultaneous
  conservative consistency of false discovery rates: a unified approach.
\newblock \emph{Journal of the Royal Statistical Society, Series B (Statistical
  Methodology)}, 66\penalty0 (1):\penalty0 187--205, 2004.

\bibitem[Sun and Cai(2007)]{sun2007oracle}
Wenguang Sun and T~Tony Cai.
\newblock Oracle and adaptive compound decision rules for false discovery rate
  control.
\newblock \emph{Journal of the American Statistical Association}, 102\penalty0
  (479):\penalty0 901--912, 2007.

\bibitem[Tian and Ramdas(2019)]{tian2019addis}
Jinjin Tian and Aaditya Ramdas.
\newblock {ADDIS}: adaptive algorithms for online {FDR} control with
  conservative nulls.
\newblock \emph{Advances in Neural Information Processing Systems}, 2019.

\bibitem[van~den Oord(2008)]{van2008controlling}
Edwin~JCG van~den Oord.
\newblock Controlling false discoveries in genetic studies.
\newblock \emph{American Journal of Medical Genetics Part B: Neuropsychiatric
  Genetics}, 147\penalty0 (5):\penalty0 637--644, 2008.

\bibitem[Zrnic et~al.(2018)Zrnic, Ramdas, and Jordan]{zrnic2018asynchronous}
Tijana Zrnic, Aaditya Ramdas, and Michael~I Jordan.
\newblock Asynchronous online testing of multiple hypotheses.
\newblock \emph{arXiv preprint arXiv:1812.05068}, 2018.

\end{thebibliography}
\bibliographystyle{plainnat}


\section{Empirical FDP Estimates in Prior Work}

We give a brief overview of BH and Storey-BH, and we do so in the FDP estimation spirit of \secref{notation}. These derivations were first stated by \citet{Storey02}. Let $\mathbf{P} = \{P_1,\dots,P_n\}$ be a set of tested $p$-values, and $\alpha$ be the target FDR level.
For any threshold $c$, Storey defined
$$\fdphat_{\text{BH}} \defn \frac{nc}{\sum_{i=1}^n \One{P_i\leq c}}.$$
Picking the maximum $c$ such that $\fdphat_{\text{BH}}\leq \alpha$, and rejecting all $p$-values less than such $c$, is a succinct statement of the BH procedure. This is a rederivation of the equivalent rule given by Benjamini and Hochberg, who suggested finding
$$k^* = \max\left\{i\in[n]: P_{(i)}\leq \frac{\alpha}{n}i\right\},$$ where $P_{(i)}$ denotes the $i$-th order statistic of $\mathbf{P}$ in non-decreasing order, and rejecting $P_{(1)},\dots,P_{(k^*)}$. This interpretation inspired Storey to improve upon the BH procedure by defining
$$\fdphat_{\stbh}\defn \frac{ns\pih_0}{\sum_{i=1}^n \One{P_i\leq c}},$$
where $\pih_0 = \frac{1 + \sum_{i=1}^n \One{P_i > \lambda}}{n(1-\lambda)}$, for a user-chosen parameter $\lambda\in(0,1)$. Storey-BH finds the maximum $c$ such that 
$\fdphat_{\stbh}\leq \alpha$, and rejects all $p$-values less than such $c$. The motivation for using Storey-BH is the observation that BH might be overly conservative when there are many non-nulls with a strong signal, because it essentially assumes that $\pih_0 \approx 1$, where $\pih_0$ acts as an estimate of the proportion of nulls in the $p$-value set.

The FDP estimate approach was also taken in more recent, online FDR work \citep{RYWJ17,ramdas2018saffron,zrnic2018asynchronous,tian2019addis}. It started with \citet{RYWJ17} who rederived and improved upon the LORD algorithm \citep{javanmard2016online} by noticing that it implicitly controls
$$\fdphat_{\text{LORD}}(t) = \frac{\sum_{j=1}^t \alpha_j}{\sum_{i=1}^t \One{P_i \leq \alpha_i}},$$
where $P_i$ is a single $p$-value observed at time $i$, and $\alpha_i$ is its corresponding test level.
Inspired by Storey's idea of making the BH procedure less conservative, the SAFFRON algorithm was derived as a rule for controlling the estimate
$$\fdphat_{\text{SAFFRON}}(t) = \frac{\sum_{j=1}^t \frac{\alpha_j}{1-\lambda_j}\One{P_j>\lambda_j}}{\sum_{i=1}^t \One{P_i \leq \alpha_i}},$$
for some sequence of user-chosen parameters $\{\lambda_t\}$. Several different update rules for $\alpha_t$ have been proposed for LORD and SAFFRON, all of which control the respective FDP estimates under the target FDR level $\alpha$; for more details, see the respective papers \citep{javanmard2016online,ramdas2018saffron}.

\section{Proof of \thmref{fdrbatchbh}}

First we introduce some additional notation necessary to state the proof. Let $P^{(-i)}_{s,1},\dots,P^{(-i)}_{s,n_s}$ be a sequence of $p$-values that is identical to $P_{s,1},\dots, P_{s,n_s}$, \emph{except} for $P^{(-i)}_{s,i}$, which is set to 0. Let also $R_s^{(-i)}$ denote the number of rejections had BH under level $\alpha_s$ been run on $P^{(-i)}_{s,1},\dots,P^{(-i)}_{s,n_s}$.

Fix the number of tested batches $t$, and suppose $\fdphat_{\batchbh}(t)\leq \alpha$. We prove that this implies $\fdr(t)\leq\alpha$. Starting by definition,
\begin{align*}
    \fdr(t) &= \EE{\frac{\sum_{r\leq t} |\cR_r \cap \nulls_r|}{1\vee \sum_{s\leq t} R_s}}\\
    &= \sum_{r \leq t} \sum_{i\in \nulls_r} \EE{\frac{\One{P_{r,i} \leq \frac{\alpha_r}{n_r} R_r}}{1\vee \sum_{s\leq t} R_s}}\\
    &= \sum_{r \leq t} \sum_{i\in \nulls_r} \EE{\frac{\One{P_{r,i} \leq \frac{\alpha_r}{n_r} R_r}}{R_r^{(-i)} + \sum_{s\leq t,s\neq r} R_s}},
\end{align*}
where the second equality follows by definition of the BH procedure and the third equality follows by observing that, on the event $\{P_{r,i} \leq \frac{\alpha_r}{n_r}R_r\}$, $R_r = R_r^{(-i)}$.

Now we focus on a fixed index $i\in\nulls_r$, for a fixed batch $r$. Imagine a sequence of batches of $p$-values identical to the original one, only with $P_{r,i}$ deterministically set to 0. Denote the set of rejections in batch $s\in\N$ in this slightly modified sequence by $\tilde R^{(-r,i)}_s$. Notice that $R_s = \tilde R^{(-r,i)}_s$ for all $s<r$, and for $s\geq r$ we have $R_s = \tilde R^{(-r,i)}_s$ if $P_{r,i}$ from the original sequence is rejected. Therefore, on the event $\{P_{r,i} \leq \frac{\alpha_r}{n_r} R_r\}$, $\tilde R^{(-r,i)}_s = R_s$ for all $s\in\N$. This implies
\begin{align*}
    \fdr(t) &= \sum_{r \leq t} \sum_{i\in \nulls_r} \EE{\frac{\One{P_{r,i} \leq \frac{\alpha_r}{n_r} R_r}}{R_r^{(-i)} + \sum_{s\leq t,s\neq r} \tilde R^{(-r,i)}_s}}\\
    &\leq \sum_{r \leq t} \sum_{i\in \nulls_r} \EE{\frac{\One{P_{r,i} \leq \frac{\alpha_r}{n_r} R_r^{(-i)}}}{R_r^{(-i)} + \sum_{s\leq t,s\neq r} \tilde R^{(-r,i)}_s}},
\end{align*}
where the final inequality uses the fact that $R_r^{(-i)}\geq R_r$. Conditional on $\F^{r-1}$, $P_{r,i}$ is independent of all other random variables in the final term, namely $\alpha_r$, $R_r^{(-i)}$ and $\tilde R^{(-r,i)}_s, s\in[t], s\neq r$. This allows us to exploit the super-uniformity of $P_{r,i}$ to obtain
\begin{align*}
    \fdr(t) &\leq \sum_{r \leq t} \sum_{i\in \nulls_r} \EE{\EEst{\frac{\One{P_{r,i} \leq \frac{\alpha_r}{n_r} R_r^{(-i)}}}{R_r^{(-i)} + \sum_{s\leq t,s\neq r} \tilde R^{(-r,i)}_s}}{\F^{r-1}, R_r^{(-i)}}}\\
    &= \sum_{r \leq t} \sum_{i\in \nulls_r} \EE{\EEst{\One{P_{r,i} \leq \frac{\alpha_r}{n_r} R_r^{(-i)}}}{\F^{r-1}, R_r^{(-i)}} \EEst{\frac{1}{R_r^{(-i)} + \sum_{s\leq t,s\neq r} \tilde R^{(-r,i)}_s}}{\F^{r-1}, R_r^{(-i)}}}\\
    &= \sum_{r \leq t} \sum_{i\in \nulls_r} \EE{\frac{\alpha_r}{n_r}\frac{R_r^{(-i)}}{R_r^{(-i)} + \sum_{s\leq t,s\neq r} \tilde R^{(-r,i)}_s}}.
\end{align*}
Since the update rule for $\alpha_r$ is monotone by assumption, we have
\begin{align*}
    \fdr(t) &\leq \sum_{r \leq t} \sum_{i\in \nulls_r} \EE{\frac{\alpha_r}{n_r}\frac{R_r^{(-i)}}{R_r^{(-i)} + \sum_{s\leq t,s\neq r} R_s}}.
\end{align*}
Finally, we use the fact that the function $f(x)=\frac{x}{x+a}$ is a non-decreasing function for $a\geq 0$ to conclude
\begin{align*}
    \fdr(t) &\leq \sum_{r \leq t} \sum_{i\in \nulls_r} \EE{\frac{\alpha_r}{n_r}\frac{R_r^+}{R_r^+ + \sum_{s\leq t,s\neq r} R_s}}\\
    &\leq \sum_{r \leq t} \EE{\alpha_r\frac{R_r^+}{R_r^+ + \sum_{s\leq t,s\neq r} R_s}}\\
    &= \EE{\fdphat_{\batchbh}}\\
    &\leq \alpha,
\end{align*}
where the last inequality is deterministic, by design of the algorithm. This concludes the proof.


\section{Proof of \thmref{fdrbatchsbh}}

As in the proof of \thmref{fdrbatchbh}, we introduce some additional notation necessary to state the proof. Recall that the Storey-BH procedure uses a null proportion estimate of the form
$$\pih_{0,s} = \frac{1 + \sum_{j=1}^{n_s} P_{s,j}}{n_s(1-\lambda_s)}.$$
Let $P^{(-i)}_{s,1},\dots,P^{(-i)}_{s,n_s}$ be a sequence of $p$-values that is identical to $P_{s,1},\dots, P_{s,n_s}$, \emph{except} for $P^{(-i)}_{s,i}$, which is set to 0. Let also $R_s^{(-i)}$ denote the number of rejections had Storey-BH under level $\alpha_s$ been run on $P^{(-i)}_{s,1},\dots,P^{(-i)}_{s,n_s}$. With this, define the ``hallucinated'' null proportion as
$$\pih_{0,s}^{(-i)} = \frac{1 + \sum_{j=1}^{n_s} P^{(-i)}_{s,j}}{n_s(1-\lambda_s)}.$$

Fix the number of tested batches $t$, and suppose $\fdphat_{\batchsbh}(t)\leq \alpha$. We prove that this condition implies $\fdr(t)\leq \alpha$. Starting by definition,
\begin{align*}
    \fdr(t) &= \EE{\frac{\sum_{r\leq t} |\cR_r\cap \nulls_r|}{1\vee \sum_{s\leq t} R_s}}\\
    &= \sum_{r \leq t} \sum_{i\in \nulls_r} \EE{\frac{\One{P_{r,i} \leq \frac{\alpha_r}{\pih_{0,r} n_r} R_r}}{1\vee \sum_{s\leq t} R_s}}\\
    &= \sum_{r \leq t} \sum_{i\in \nulls_r} \EE{\frac{\One{P_{r,i} \leq \frac{\alpha_r}{\pih_{0,r} n_r} R_r}}{R_r^{(-i)} + \sum_{s\leq t,s\neq r} R_s}},
\end{align*}
where the second equality follows by definition of the Storey-BH procedure and the third equality follows by observing that, on the event $\{P_{r,i} \leq \frac{\alpha_r}{n_r\pih_{0,r}}R_r\}$, $R_r = R_r^{(-i)}$.

Now we focus on a fixed $i\in\nulls_r$, for a fixed batch $r$. Imagine a sequence of batches of $p$-values identical to the original one, only with $P_{r,i}$ deterministically set to 0. Denote the set of rejections in batch $s\in\N$ in this slightly modified sequence by $\tilde R^{(-r,i)}_s$. Notice that $R_s = \tilde R^{(-r,i)}_s$ for all $s<r$, and for $s\geq r$ we have $R_s = \tilde R^{(-r,i)}_s$ if $P_{r,i}$ from the original sequence in rejected. Therefore, on the event $\{P_{r,i} \leq \frac{\alpha_r}{n_r\pih_{0,r}} R_r\}$, $\tilde R^{(-r,i)}_s = R_s$ for all $s\in\N$. This implies
\begin{align*}
    \fdr(t) &= \sum_{r \leq t} \sum_{i\in \nulls_r} \EE{\frac{\One{P_{r,i} \leq \frac{\alpha_r}{n_r \pih_{0,r}} R_r}}{R_r^{(-i)} + \sum_{s\leq t,s\neq r} \tilde R^{(-r,i)}_s}}\\
    &\leq \sum_{r \leq t} \sum_{i\in \nulls_r} \EE{\frac{\One{P_{r,i} \leq \frac{\alpha_r}{n_r \pih_{0,r}^{(-i)}} R_r^{(-i)}}}{R_r^{(-i)} + \sum_{s\leq t,s\neq r} \tilde R^{(-r,i)}_s}},
\end{align*}
where the final inequality uses the fact that $R_r^{(-i)}\geq R_r$ and $\pih_{0,r}\geq \pih_{0,r}^{(-i)}$. Note that $P_{r,i}$ is independent of all other random variables in the final term, namely $\alpha_r, R_r^{(-i)}$, $\pih_{0,r}^{(-i)}$ and $\tilde R_s, s\in[t], s\neq r$. This allows us to exploit the super-uniformity of $P_{r,i}$ to obtain
\begin{align*}
    \fdr(t) &\leq \sum_{r \leq t} \sum_{i\in \nulls_r} \EE{\EEst{\frac{\One{P_{r,i} \leq \frac{\alpha_r}{n_r\pih_{0,r}^{(-i)}} R_r^{(-i)}}}{R_r^{(-i)} + \sum_{s\leq t,s\neq r} \tilde R^{(-r,i)}_s}}{\F^{r-1}, R_r^{(i)}, \pih_{0,r}^{(-i)}}}\\
    &= \sum_{r \leq t} \sum_{i\in \nulls_r} \EE{\EEst{\One{P_{r,i} \leq \frac{\alpha_r}{n_r\pih_{0,r}^{(-i)}} R_r^{(-i)}}}{\F^{r-1}, R_r^{(i)}, \pih_{0,r}^{(-i)}} \EEst{\frac{1}{R_r^{(i)} + \sum_{s\leq t,s\neq r} \tilde R^{(-r,i)}_s}}{\F^{r-1}, R_r^{(-i)}, \pih_{0,r}^{(-i)}}}\\
    &= \sum_{r \leq t} \sum_{i\in \nulls_r} \EE{\frac{\alpha_r}{n_r\pih_{0,r}^{(-i)}}\frac{R_r^{(i)}}{R_r^{(-i)} + \sum_{s\leq t,s\neq r} \tilde R^{(-r,i)}_s}}.
\end{align*}
Since the update for $\alpha_r$ is monotone, and since setting a $p$-value to 0 can only increase the number of rejections in a given batch, we have
\begin{align*}
    \fdr(t) \leq \sum_{r \leq t} \sum_{i\in \nulls_r} \EE{\frac{\alpha_r}{n_r \pih_{0,r}^{(-i)}}\frac{R_r^{(-i)}}{R_r^{(-i)} + \sum_{s\leq t,s\neq r} R_s}}.
\end{align*}
Now we use a similar trick of ignoring one $p$-value as given above. Imagine a sequence of $p$-values identical to the original one, however with $P_{r,i}$ deterministically set to 1. Denote the set of rejections in batch $s\in\N$ in this modified sequence by $\tilde R_s^{(+r,i)}$. We have $R_s = R_s^{(+r,i)}$ for $s<r$, and the same holds for $s\geq r$ on the event $\{P_{r,i}>\lambda_r\}$. From this, we can conclude the following
\begin{align*}
    \EE{ \frac{\One{P_{r,i}>\lambda_r}}{(1-\lambda_r)} \frac{\alpha_r}{n_r \pih_{0,r}^{(-i)}}\frac{R_r^{(-i)}}{R_r^{(-i)} + \sum_{s\leq t,s\neq l} R_s}} &=\EE{ \frac{\One{P_{r,i}>\lambda_r}}{(1-\lambda_r)} \frac{\alpha_r}{n_r \pih_{0,r}^{(-i)}}\frac{R_r^{(-i)}}{R_r^{(-i)} + \sum_{s\leq t,s\neq r} \tilde R^{(+r,i)}_s}}\\
    &=\EE{\frac{\One{P_{r,i}>\lambda_r}}{(1-\lambda_r)}}  \EE{\frac{\alpha_r}{n_r \pih_{0,r}^{(-i)}}\frac{R_r^{(-i)}}{R_r^{(-i)} + \sum_{s\leq t,s\neq r} \tilde R^{(+r,i)}_s}}\\
    &\geq \EE{\frac{\alpha_r}{n_r \pih_{0,r}^{(-i)}}\frac{R_r^{(-i)}}{R_r^{(-i)} + \sum_{s\leq t,s\neq r} \tilde R^{(+r,i)}_s}}\\
    &\geq \EE{\frac{\alpha_r}{n_r \pih_{0,r}^{(-i)}}\frac{R_r^{(-i)}}{R_r^{(-i)} + \sum_{s\leq t,s\neq r} R_s}},
\end{align*}
where the first inequality uses super-uniformity of null $p$-values, and the second inequality uses monotonicity of the test level update rule. Therefore, we can write
$$\fdr(t) \leq \sum_{r \leq t} \sum_{i\in \nulls_r} \EE{ \frac{\One{P_{r,i}>\lambda_r}}{(1-\lambda_r)} \frac{\alpha_r}{n_r \pih_{0,r}^{(-i)}}\frac{R_r^{(-i)}}{R_r^{(-i)} + \sum_{s\leq t,s\neq r} R_s}}.$$
Finally, we use the fact that the function $f(x)=\frac{x}{x+a}$ is a non-decreasing function for $a\geq 0$ to conclude
\begin{align*}
    \fdr(t) &\leq \sum_{r \leq t} \sum_{i\in \nulls_r} \EE{ \frac{\One{P_{r,i}>\lambda_r}}{(1-\lambda_r)} \frac{\alpha_r}{n_r \min_j \pih_{0,r}^{(-j)}}\frac{R_r^+}{R_r^+ + \sum_{s\leq t,s\neq r} R_s}}\\
    &\leq \sum_{r \leq t} \EE{\frac{\sum_{i\leq n_r} \One{P_{r,i} > \lambda_r}}{1 + \sum_{j\leq n_r, j\neq\max_r} \One{P_{r,j} > \lambda_r}}\frac{\alpha_r R_r^+}{R_r^+ + \sum_{s\leq t,s\neq r} R_s}}\\
    &= \sum_{r \leq t} \EE{\frac{\alpha_r  R_r^+ \One{P_{r,\max_r} > \lambda_r}}{R_r^+ + \sum_{s\leq t,s\neq r} R_s}}\\
    &= \EE{\fdphat_{\batchsbh}}\\
    &\leq \alpha,
\end{align*}
where once again the last inequality is deterministic by design of the algorithm, thus completing the proof of the theorem.

\section{\texorpdfstring{$\batchbh^{\text{PRDS}}$}{PRDS} Proofs}

To facilitate the proof of FDR control, we prove a a ``super-uniformity lemma'', similar to several lemmas in prior work that consider PRDS $p$-values \citep{blanchard2008two,BY01,ramdas2019unified}.

\begin{lemma}
\label{lem:PRDSlemma}
Let $U\in[0,1]$ and $V\in\N\cup\{0\}$ be random variables that satisfy the following:
\begin{itemize}
    \item $U$ is super-uniform, i.e. $\PP{U\leq u}\leq u$ for $u\in[0,1]$.
    \item $\PPst{V\leq r}{U \leq u}$ is non-decreasing in $u$, for every fixed $r>0$.
    \item $V\leq n$ almost surely.
\end{itemize}
Then, for every $a\geq 0, c>0$, $\EE{\frac{\One{U\leq cV}}{V+a}} \leq \frac{cn}{n+a}$.
\end{lemma}

\begin{proof}
The proof when $a=0$ is given by \citet{blanchard2008two} (Lemma 3.2), so in what follows we assume $a>0$.

We expand the expectation as follows:
\begin{align*}
    \EE{\frac{\One{U\leq cV}}{V+a}} &= \sum_{i=0}^n \frac{1}{i + a} \PP{U\leq ci, V = i}\\
    &= \sum_{i=0}^n \frac{\PP{U\leq ci}}{i + a} \frac{\PP{U\leq ci, V = i}}{\PP{U\leq ci}}\\
    &\leq \sum_{i=0}^n \frac{ci}{i + a} \PPst{V = i}{U\leq ci}\\
    &= \sum_{i=0}^n \frac{ci}{i + a} \left(\PPst{V \leq i}{U\leq ci} - \PPst{V \leq i-1}{U\leq ci}\right)\\
    &= \frac{cn}{n + a}  \sum_{i=0}^n\left(\PPst{V \leq i}{U\leq ci} - \PPst{V \leq i-1}{U\leq ci}\right),
\end{align*}
where the inequality follows by the super-uniformity assumption on $U$. By the second assumption of the lemma, $\PPst{V \leq i-1}{U\leq ci}\geq \PPst{V \leq i-1}{U\leq c(i-1)}$, hence
\begin{align*}
    &\frac{cn}{n + a}  \sum_{i=0}^n\left(\PPst{V \leq i}{U\leq ci} - \PPst{V \leq i-1}{U\leq ci}\right) \\
    \leq~ &\frac{cn}{n + a}  \sum_{i=0}^n\left(\PPst{V \leq i}{U\leq ci} - \PPst{V \leq i-1}{U\leq c(i-1)}\right)\\
    \leq~ &\frac{cn}{n + a},
\end{align*}
which follows by a telescoping sum argument.
\end{proof}

\subsection{Proof of \thmref{fdrPRDS}}

Fix the number of tested batches $t$, and suppose $\fdphat_{\batchbh^{\text{PRDS}}}(t)\leq \alpha$. We prove that this implies $\fdr(t)\leq\alpha$. Starting by definition,
\begin{align*}
    \fdr(t) &= \EE{\frac{\sum_{r\leq t} |\cR_r \cap \nulls_r|}{1\vee \sum_{s\leq t} R_s}}\\
    &= \sum_{r \leq t} \sum_{i\in \nulls_r} \EE{\frac{\One{P_{r,i} \leq \frac{\alpha_r}{n_r} R_r}}{1\vee \sum_{s\leq t} R_s}}\\
    &\leq \sum_{r \leq t} \sum_{i\in \nulls_r} \EE{\frac{\One{P_{r,i} \leq \frac{\alpha_r}{n_r} R_r}}{1\vee (R_r + \sum_{s<r} R_s)}},
\end{align*}
where the second equality follows by definition of the BH procedure and the inequality follows by ignoring all rejections in the denominator after the $r$-th batch.

We now condition on $\F^{r-1}$ to obtain
\begin{align*}
    \sum_{r \leq t} \sum_{i\in \nulls_r} \EE{\frac{\One{P_{r,i} \leq \frac{\alpha_r}{n_r} R_r}}{1\vee (R_r + \sum_{s<r} R_s)}} &= \sum_{r \leq t} \sum_{i\in \nulls_r} \EE{\EEst{\frac{\One{P_{r,i} \leq \frac{\alpha_r}{n_r} R_r}}{1\vee (R_r + \sum_{s<r} R_s)}}{\F^{r-1}}}\\
    &\leq \sum_{r \leq t} \sum_{i\in \nulls_r} \EE{\frac{\frac{\alpha_r}{n_r}n_r}{1\vee (n_r + \sum_{s<r} R_s)}}\\
    &= \sum_{r \leq t} \EE{\frac{n_r\alpha_r}{1\vee (n_r + \sum_{s<r} R_s)}},
\end{align*}
where the inequality applies \lemref{PRDSlemma} and the fact that $\alpha_r$ and $\{R_s, s< r\}$ are measurable with respect to the conditioning, and the final equality uses the fact that $|\nulls_r|\leq n_r$.

Since the final expression is equal to $\EE{\fdphat_{\batchbh^{\text{PRDS}}}}$, we can conclude that
 $$\fdr(t)\leq \EE{\fdphat_{\batchbh^{\text{PRDS}}}} \leq \alpha,$$
 as desired.


\section{Proof of \factref{defaultbatchbhvalid}}

The control of the estimate follows by observing
\begin{align*}
	\fdphat_{\batchbh}(t+1) &\leq \sum_{s\leq t} \alpha_s \frac{R_s^+}{R_s^+ + \sum_{r\leq t+1, r\neq s} R_r} + \alpha_{t+1} \frac{n_{t+1}}{n_{t+1} + \sum_{r\leq t} R_r}\\
	&= \sum_{s\leq t}\gamma_s \alpha - \sum_{s\leq t} \alpha_s \frac{R_s^+}{R_s^+ + \sum_{r\leq t, r\neq s} R_r} + \sum_{s\leq t} \alpha_s \frac{R_s^+}{R_s^+ + \sum_{r\leq t+1, r\neq s} R_r}\\
	&\leq \alpha,
\end{align*}
where the second step follows by replacing $\alpha_{t+1}$ with the update rule from Algorithm \ref{alg:defaultbatchbh}, and the final inequality follows by the assumption that $\sum_{j=1}^\infty \gamma_j = 1$.

\section{Additional Monotone Update Rules}

In this section we provide several monotone updates for $\batchbh$ and $\batchsbh$, which control the FDR for arbitrary, possibly adversarially chosen $p$-value distributions.

\subsection{\texorpdfstring{$\batchbh$}{Batch BH} Rules}

\begin{algorithm}[H]
\SetAlgoLined
\SetKwInOut{Input}{input}
\Input{FDR level $\alpha$, non-increasing sequence $\{\gamma_t\}_{t=1}^\infty$ summing to 1, initial wealth $W_0\leq\alpha$}
Set $\alpha_1 = \gamma_1 \frac{W_0}{n_1}$\newline
 \For{$t=1,2,\dots$}{
 Run the BH procedure under level $\alpha_t$ on batch $\mathbf{P}_t$\newline
  Set $\alpha_{t+1} = \gamma_{t+1} \frac{W_0}{n_{t+1}} + \frac{\alpha}{n_{t+1}} \sum_{s=1}^t \gamma_{t+1 - s} R_s - \frac{W_0}{n_{t+1}} \sum_{s=1}^t\gamma_{t+1 - s}\One{s = \tau_1}$,\\
  where $\tau_1 = \min\{s \geq 1: R_s>0\}$
 }
 \caption{One version of the $\batchbh$ algorithm}
 \label{alg:algo1}
\end{algorithm}

\begin{fact}
\label{fact:batchbhvalid}
The algorithm given in Algorithm \ref{alg:algo1} is monotone and guarantees $\fdphat_{\batchbh}(t)\leq\alpha$.
\end{fact}

\begin{proof}
	First we prove that the algorithm guarantees $\fdphat_{\batchbh}(t)\leq\alpha$. Starting by definition, we have
\begin{align*}
    \fdphat_{\batchbh}(t) &= \sum_{j=1}^t \alpha_j \frac{R_j^+}{R_j^+ + \sum_{k\leq t, k\not=j} R_k}\\
    &\leq \sum_{j=1}^t \alpha_j \frac{n_j}{n_j + \sum_{k\leq t, k\not=j} R_k} \\
    &\leq \frac{\sum_{j=1}^t \alpha_j n_j}{1 \vee \sum_{k=1}^t R_k} \\
    &= \frac{W_0 \sum_{j=1}^t \gamma_j + \alpha \sum_{j=1}^t \sum_{l=1}^{j-1} \gamma_{j-l} R_l - W_0 \sum_{j=1}^t \sum_{l=1}^{j-1} \gamma_{j-l} \One{l=\tau_1}}{1\vee \sum_{k=1}^t R_k} \\
    &= \frac{W_0 \sum_{j=1}^t \gamma_j + \alpha \sum_{j=1}^t \sum_{l=1}^{j-1} \gamma_{j-l} R_l - W_0 \sum_{j=\tau_1+1}^t \gamma_{j-\tau_1}}{1\vee \sum_{k=1}^t R_k},
\end{align*}
where the first inequality follows because $R_j^+\leq n_j$, the second inequality follows because $n_j \geq R_j \vee 1$, the second equality follows by the definition of $\alpha_j$, as given in Algorithm \ref{alg:algo1}, and the third equality is obtained by removing the summation terms where $l\not=\tau_1$.

If $t<\tau_1$, then $R_1=R_2=\dotsb=R_t=0$ by the definition of $\tau_1$, so the bound above evaluates to $W_0\sum_{j=1}^t \gamma_j\leq W_0\leq\alpha$, which is the desired conclusion. Thus, for the remainder of the proof, we assume that $t\geq\tau_1$.

Since $R_i=0$ for $i<\tau_1$, we can remove such terms from consideration, leaving us with
\begin{align*}
    \fdphat_{\batchbh}(t)
    &\leq \frac{W_0 \sum_{j=1}^t \gamma_j + \alpha \sum_{j=\tau_1+1}^t \sum_{l=\tau_1}^{j-1} \gamma_{j-l} R_l - W_0 \sum_{j=\tau_1+1}^t \gamma_{j-\tau_1}}{\sum_{k=\tau_1}^t R_k} \\
    &= \frac{W_0 \sum_{j=1}^t \gamma_j + \alpha \sum_{j=\tau_1+2}^t \sum_{l=\tau_1+1}^{j-1} \gamma_{j-l} R_l + (\alpha R_{\tau_1} - W_0) \sum_{j=\tau_1+1}^t \gamma_{j-\tau_1}}{\sum_{k=\tau_1}^t R_k}.
\end{align*}

Since $\{\gamma_t\}_{t=1}^\infty$ is defined to be a non-negative sequence summing to $1$, then $W_0 \sum_{j=1}^t \gamma_j\leq W_0$ and $(\alpha R_{\tau_1} - W_0) \sum_{j=\tau_1+1}^t \gamma_{j-\tau_1}\leq \alpha R_{\tau_1} - W_0$. We apply this observation to obtain
\begin{align*}
    \fdphat_{\batchbh}(t) &\leq \alpha \frac{R_{\tau_1} + \sum_{j=\tau_1+2}^t \sum_{l=\tau_1+1}^{j-1} \gamma_{j-l} R_l}{\sum_{k=\tau_1}^t R_k} \\
    &\leq \alpha \frac{R_{\tau_1} + \sum_{l=\tau_1+1}^{t-1} R_l}{\sum_{k=\tau_1}^t R_k} \leq \alpha,
\end{align*}
where we again use the fact that the sequence $\{\gamma_t\}_{t=1}^\infty$ sums to one. The final inequality concludes the proof that $\fdphat_{\batchbh}(t)$ is controlled.

We now prove that the update rule is monotone. We restate the test level update rule in a more convenient form:
\begin{equation} \label{eq:lemma1proofalpha}
    \alpha_{t+1} = \gamma_{t+1} \frac{W_0}{n_{t+1}} + \frac{1}{n_{t+1}} \sum_{j=1}^t \gamma_{t+1-j}(\alpha R_j - W_0 \One{j=\tau_1}).
\end{equation}

Suppose we have two sequences of $p$-values $(P_{1,1},P_{1,2},\dotsc,P_{t,n_t})$ and $(\tilde P_{1,1}, \tilde P_{1,2}, \dotsc, \tilde P_{t,n_t})$ such that $(P_{1,1},P_{1,2},\dotsc,P_{t,n_t}) \leq (\tilde P_{1,1}, \tilde P_{1,2}, \dotsc, \tilde P_{t,n_t})$ coordinate-wise. Denote all relevant $\batchbh$ quantities on these two sequences using a similar notation.

If $\alpha_t \geq \tilde \alpha_t$, then $R_t \geq \tilde R_t$ by the definition of the $\bh$ procedure. The final observation is that the above update is monotonically increasing in $(R_1, R_2,\dotsc,R_t)$, which concludes the proof.
\end{proof}

If we know that all batches are of size at least $M$, we can also derive the following rule which is expected to be more powerful than the one above, when the batch sizes do not vary too much. Moreover, when all batches are of the same size $n_j \equiv M$, the rule is strictly more powerful.

\begin{algorithm}[H]
\SetAlgoLined
\SetKwInOut{Input}{input}
\Input{FDR level $\alpha$, non-increasing sequence $\{\gamma_t\}_{t=1}^\infty$ summing to 1, initial wealth $W_0\leq\alpha$}
Set $\alpha_1 = \gamma_1 \frac{M}{n_1}\alpha$\newline
 \For{$t=1,2,\dots$}{
 Run the BH procedure under level $\alpha_t$ on batch $\mathbf{P}_t$\newline
  Set $\alpha_{t+1} = \gamma_{t+1} \frac{M}{n_{t+1}} \alpha + \frac{\alpha}{n_{t+1}} \sum_{s=1}^t \gamma_{t+1 - s} R_s^{\text{add}}$, where $R_s^{\text{add}} = \min\{R_s,\max\{R_r:r<s\}\}$  
 }
 \caption{One version of the $\batchbh$ algorithm when $n_s\geq M$ for all $s$}
 \label{alg:bhM}
\end{algorithm}

\begin{fact}
\label{fact:batchbhMvalid}
If $n_s\geq M$ for all $s\in\N$, the algorithm given in Algorithm \ref{alg:bhM} is monotone and guarantees $\fdphat_{\batchbh}(t)\leq\alpha$.
\end{fact}

\begin{proof}

We first prove that the algorithm guarantees $\fdphat_{\batchbh}(t)\leq\alpha$. Starting by definition, we have
\begin{align*}
    \fdphat_{\batchbh}(t) &= \sum_{j=1}^t \alpha_j \frac{R_j^+}{R_j^+ + \sum_{k\leq t, k\not=j} R_k} \leq \sum_{j=1}^t \alpha_j \frac{n_j}{n_j + \sum_{k\leq t, k\not=j} R_k} \\
    &\leq \sum_{j=1}^t \alpha_j \frac{n_j}{M + \sum_{k\leq t, k\not=j} R_k} \leq \frac{\sum_{j=1}^t \alpha_j n_j}{M + \min_i\sum_{k\leq t, k\not=i} R_k},
\end{align*}
where the first inequality follows because $R_j^+\leq n_j$ and second inequality follows by the assumption that $n_j\geq M$ for all $j$. Substituting in the update rule from Algorithm \ref{alg:bhM}, we obtain
\begin{align*}
    \fdphat_{\batchbh}(t) \leq \alpha \frac{M \sum_{j=1}^t \gamma_j + \sum_{j=1}^t \sum_{l=1}^{j-1} \gamma_{j-l} R_l^{\text{add}}}{M + \min_i\sum_{k\leq t, k\not=i} R_k} \leq \alpha \frac{M + \sum_{l=1}^{t-1} R_l^{\text{add}}}{M + \min_i\sum_{k\leq t, k\not=i} R_k},
\end{align*}
where we use the fact that the sequence $\{\gamma_t\}_{t=1}^\infty$ is defined to be non-negative and summing to one. Since $\sum_{l=1}^{t-1} R_l^{\text{add}} = \min_i\sum_{k< t, k\not=i} R_k$ by the definition of $R_l^{\text{add}}$, we can conclude $\fdphat_{\batchbh}(t) \leq \alpha$, as desired. 

Monotonicity follows by the same steps as in the proof of \factref{batchbhvalid}, thus completing the proof.

\end{proof}

%

Below we give one more monotone $\batchbh$ update, based on a different idea.

\begin{algorithm}[H]
\SetAlgoLined
\SetKwInOut{Input}{input}
\Input{FDR level $\alpha$, sequence $\{\gamma_t\}_{t=1}^\infty$ summing to 1 such that $\gamma_2\geq \gamma_1$}
Set $\alpha_1 = \gamma_1 \alpha$\newline
Run the BH procedure under level $\alpha_1$ on batch $\mathbf{P}_1$\newline
Set $\alpha_2 = (\gamma_2\alpha - \alpha_1)\frac{R_1 + n_2}{n_2}$\newline
 \For{$t=2,3,\dots$}{
 Run the BH procedure under level $\alpha_t$ on batch $\mathbf{P}_t$\newline
  Set $\alpha_{t+1} = \left(\frac{R_t \left(\sum_{i=1}^{t-1}\alpha_i n_i\right)}{(\sum_{j=1}^{t-1} n_j)(\sum_{k=1}^{t-1} n_k + R_t)} + \gamma_{t+1}\alpha\right) \frac{\sum_{l=1}^t R_l + n_{t+1}}{n_{t+1}}$
 }
 \caption{One version of the $\batchbh$ algorithm}
 \label{alg:algomonotone}
\end{algorithm}

\begin{fact}
\label{lem:algomonotoneproof}
The update given in Algorithm \ref{alg:algomonotone} controls $\fdphat_{\batchbh}$ and is monotone.
\end{fact}

\begin{proof}

First we use induction to prove that for every $t\in\N$ the update controls $\sum_{i=1}^t \alpha_i \frac{n_i}{n_i + \sum_{j<t,j\neq i} R_j}$ under $\sum_{i=1}^t \gamma_i\alpha$. Then, since $\sum_{i=1}^t \alpha_i \frac{n_i}{n_i + \sum_{j<t,j\neq i} R_j} \geq \sum_{i=1}^t \alpha_i \frac{R_i^+}{R_i^+ + \sum_{j<t,j\neq i} R_j} \geq \fdphat_{\batchbh}(t)$, the first claim in the fact immediately follows.

This statement is clearly true for the two special cases when $t\in\{1,2\}$, and now assume $\sum_{i=1}^t \alpha_i \frac{n_i}{n_i + \sum_{j<t,j\neq i} R_j} \leq \sum_{i=1}^t \gamma_i \alpha$.

We can write
\begin{equation}
\label{eqn:identity}
    \frac{R_t\left(\sum_{i=1}^{t-1}\alpha_i n_i\right)}{(\sum_{j=1}^{t-1} n_j)(\sum_{k=1}^{t-1} n_k + R_t)} = \sum_{i=1}^{t-1} \alpha_i \left(\frac{n_i}{\sum_{j=1}^{t-1} n_j} - \frac{n_i}{\sum_{k=1}^{t-1} n_k + R_t} \right).
\end{equation}
We use this to rewrite $\sum_{i=1}^{t+1} \alpha_i \frac{n_i}{n_i + \sum_{j<t+1,j\neq i} R_j}$ as
\begin{align*}
    \sum_{i=1}^{t+1} \alpha_i \frac{n_i}{n_i + \sum_{j<t+1,j\neq i} R_j} &= \sum_{i=1}^{t} \alpha_i \frac{n_i}{n_i + \sum_{j\leq t,j\neq i} R_j} + \sum_{i=1}^{t-1} \alpha_i \left(\frac{n_i}{\sum_{j=1}^{t-1} n_j} - \frac{n_i}{\sum_{k=1}^{t-1} n_k + R_t} \right) + \gamma_{t+1}\alpha,\\
\end{align*}
where we apply the test level update given in Algorithm \ref{alg:algomonotone}. Further, we have
\begin{align*}
    \sum_{i=1}^{t+1} \alpha_i \frac{n_i}{n_i + \sum_{j<t+1,j\neq i} R_j}&\leq \sum_{i=1}^{t} \alpha_i \frac{n_i}{n_i + \sum_{j\leq t,j\neq i} R_j} + \sum_{i=1}^{t-1} \alpha_i \left(\frac{n_i}{ n_i + \sum_{j\leq t-1,j\neq i} R_j} - \frac{n_i}{n_i + \sum_{k\leq t,k \neq i} R_k} \right)\\
     &+ \gamma_{t+1}\alpha\\
    &= \sum_{i=1}^{t} \alpha_i \frac{n_i}{n_i + \sum_{j < t,j\neq i} R_j} + \gamma_{t+1}\alpha\\
    &\leq \sum_{i=1}^{t+1}\gamma_{i}\alpha,
\end{align*}
where the last step follows by the induction hypothesis. This completes the proof that $\fdphat_{\batchbh}$ is controlled.

Monotonicity now follows by observing that the test levels updates are increasing in the rejection counts $R_i$, as well as previous test levels. The only exception is $\alpha_2$ which is decreasing in $\alpha_1$, however because $\alpha_1$ is non-random, this does not hurt monotonicity.
\end{proof}


\subsection{\texorpdfstring{$\batchsbh$}{Batch St-BH} Rules}

\begin{algorithm}[H]
\SetAlgoLined
\SetKwInOut{Input}{input}
\Input{FDR level $\alpha$, sequence of constants $\{\lambda_t\}_{t=1}^\infty$, non-increasing sequences $\{\gamma_t\}_{t=1}^\infty$ and $\{\gamma_t'\}_{t=1}^\infty$ summing to 1, initial wealth $W_0\leq \alpha$}
Set $\alpha_1 = \gamma_1 \frac{W_0}{n_1}$\newline
\For{$t=1,2,\dots$}{
 Run the Storey-BH procedure under level $\alpha_t$ with parameter $\lambda_t$ on batch $\mathbf{P}_t$\newline
 Set $\alpha_{t+1} =  \frac{1}{n_{t+1}}\left(\gamma_{t+1} W_0 + \alpha \sum_{s=1}^t \gamma_{t+1-s} R_s - W_0 \sum_{s=1}^t \gamma_{t+1-s} \One{s = \tau_1} + \sum_{s=1}^t \gamma_{t+1-s}'\One{P_{s,\max_s}\leq\lambda_s}n_s\alpha_s\right)$,\\
 where $\tau_1 = \min\{s\geq 1: R_s > 0\}$
 }
 \caption{One version of the $\batchsbh$ algorithm}
 \label{alg:algosbh1}
\end{algorithm}

Below we prove that the update rule of Algorithm \ref{alg:algosbh1} is monotone and satisfies \defref{sbh}.

\begin{fact}
\label{fact:batchsbhvalid}
The algorithm given in Algorithm \ref{alg:algosbh1} is monotone and guarantees $\fdphat_{\batchsbh}(t)\leq\alpha$.
\end{fact}

\begin{proof}
To simplify notation, throughout the proof we denote $k_j := \One{P_{j,\max_j}>\lambda_j}$.

First we prove that the algorithm guarantees $\fdphat_{\batchsbh}(t)\leq \alpha$. It is not hard to see that $\fdphat_{\batchsbh}(t)\leq \frac{\sum_{j=1}^t n_j \alpha_j k_j}{1\vee \sum_{j=1}^t R_j}$. Therefore, it suffices to prove $\sum_{j=1}^t n_j \alpha_j k_j \leq \alpha (1\vee \sum_{j=1}^t R_j)$ for all $t$.

For all $t$, define $s(t) \defn \gamma_t W_0 + \alpha \sum_{j=1}^{t-1} \gamma_{t-j} R_j - W_0 \sum_{j=1}^{t-1}\gamma_{t-j} \One{j=\tau_1}$. With this, the test levels are equal to $n_t\alpha_t = s(t) + \sum_{j=1}^{t-1} \gamma_{t-j}' (1-k_j)n_j\alpha_j$. In \factref{batchbhvalid} we have proved that $\sum_{j=1}^t s(j)\leq \alpha (1\vee \sum_{j=1}^t R_j)$, so it suffices to prove $\sum_{j=1}^t n_j \alpha_j k_j \leq \sum_{j=1}^t s(j)$. We do so by peeling terms off one by one:
\begin{align*}
	\sum_{j=1}^t n_j\alpha_j k_j &\leq \sum_{j=1}^{t-1} n_j\alpha_j k_j + n_t \alpha_t\\
	&= \sum_{j=1}^{t-1} n_j\alpha_j k_j + s(t) + \sum_{j=1}^{t-1} \gamma_{t-j}' (1-k_j)n_j\alpha_j\\
	&\leq \sum_{j=1}^{t-2} n_j \alpha_j k_j + s(t) + \sum_{j=1}^{t-2} \gamma_{t-j}' (1-k_j)n_j\alpha_j + n_{t-1}\alpha_{t-1}.
\end{align*}
By repeating a similar argument recursively we obtain
$$\sum_{j=1}^t n_j\alpha_j k_j \leq \sum_{j=1}^t s(j).$$
Invoking \factref{batchbhvalid} now completes the proof that $\fdphat_{\batchsbh}(t)\leq \alpha$ for all $t\in\N$.

Now we prove that the update rule is additionally monotone. Take two sequences of $p$-values such that $(P_{1,1},P_{1,2},\dots,P_{t,n_t})\geq (\tilde P_{1,1},\tilde P_{1,2},\dots,\tilde P_{t,n_t})$ coordinate-wise. Denote all relevant $\batchsbh$ quantities on these two sequences using a similar notation, for example we distinguish between $k_i$ on $(P_{1,1},P_{1,2},\dots,P_{t,n_t})$ and $\tilde k_i$ on $(\tilde P_{1,1},\tilde P_{1,2},\dots,\tilde P_{t,n_t})$.

The first observation is that $\tilde k_i \leq k_i$, for all $i\in[t]$. This follows because $\One{\tilde P_{i,j} > \lambda_i}\leq \One{ P_{i,j} > \lambda_i}$ and so $\One{\tilde P_{i,\tilde \max_i} > \lambda_i}\leq \One{ P_{i, \max_i} > \lambda_i}$.
The rest of the proof follows by combining the monotonicity proof in \factref{batchbhvalid} and the fact that the update for $\alpha_t$ is non-increasing in $(k_1,\dots,k_{t-1})$.
\end{proof}

As for the $\batchbh$ family of algorithms, we also propose a rule with provable guarantees when $n_j \geq M$ for all $j$, which is expected to be more powerful when the batch sizes are roughly of the same order; moreover, if $n_j\equiv M$, the rule is strictly more powerful than the one above.

\begin{algorithm}[H]
\SetAlgoLined
\SetKwInOut{Input}{input}
\Input{FDR level $\alpha$, non-increasing sequences $\{\gamma_t\}_{t=1}^\infty$ and $\{\gamma_t'\}_{t=1}^\infty$ summing to one}
Set $\alpha_1 = \gamma_1 \frac{M}{n_1}\alpha$\newline
\For{$t=1,2,\dots$}{
 Run the Storey-BH procedure under level $\alpha_t$ with parameter $\lambda_t$ on batch $\mathbf{P}_t$\newline
 Set $\alpha_{t+1} =  \frac{M}{n_{t+1}} \gamma_{t+1} \alpha + \frac{\alpha}{n_{t+1}} \sum_{s=1}^t \gamma_{t+1-s} R^{\text{add}}_s + \frac{1}{n_{t+1}}\sum_{s=1}^t \gamma_{t+1-s}'\One{P_{s,\max_s}\leq \lambda_s}n_s\alpha_s$,\\
  where $R^{\text{add}}_s = \min\{R_s, \max\{R_r:r<s\}\}$
 }
 \caption{One version of the $\batchsbh$ algorithm when $n_s\geq M$ for all $s$}
 \label{alg:stbhM}
\end{algorithm}

\begin{fact}
\label{fact:batchsbhMvalid}
The algorithm given in Algorithm \ref{alg:stbhM} is monotone and guarantees $\fdphat_{\batchsbh}(t)\leq\alpha$.
\end{fact}

\begin{proof}
To simplify notation, we again denote $k_j := \One{P_{j,\max_j}>\lambda_j}$.

First we prove that it guarantees $\fdphat_{\batchsbh}(t)\leq \alpha$. It is not hard to see that $\fdphat_{\batchsbh}(t)\leq \frac{\sum_{j=1}^t n_j \alpha_j k_j}{M + \sum_{j=1}^t R_j^{\text{add}}}$ (recall that $R_1^{\text{add}}=0$ by definition). Therefore, it suffices to prove $\sum_{j=1}^t n_j \alpha_j k_j \leq \alpha (M + \sum_{j=1}^t R_j^{\text{add}})$.

If we denote $s(t)\defn M \gamma_{t}\alpha + \alpha \sum_{j=1}^{t-1} \gamma_{t-j} R_j^{\text{add}}$, by the same argument as in \factref{batchsbhvalid} we can conclude that
$$\sum_{j=1}^t n_j \alpha_j k_j \leq \sum_{j=1}^t s(j).$$

Following the steps of \factref{batchbhMvalid}, we can also show that $\sum_{j=1}^t s(j)\leq \alpha (M + \sum_{j=1}^t R_j^{\text{add}})$, which completes the proof that $\fdphat_{\batchsbh}(t)\leq \alpha$ for all $t\in\N$. 

The proof that the update rule is additionally monotone combines the monotonicity proofs of \factref{batchbhMvalid} and \factref{batchsbhvalid}.
\end{proof}


\section{Additional Simulation Results Against Time}

For the experiments presented in \secref{experiments}, we plot the power and FDR as functions of time, for $\pi_1\in\{0.1,0.5\}$. We observe both power and FDR to be stable across time for our default algorithms.

\begin{figure}[h]
  \centerline{
  \includegraphics[width=0.5\columnwidth]{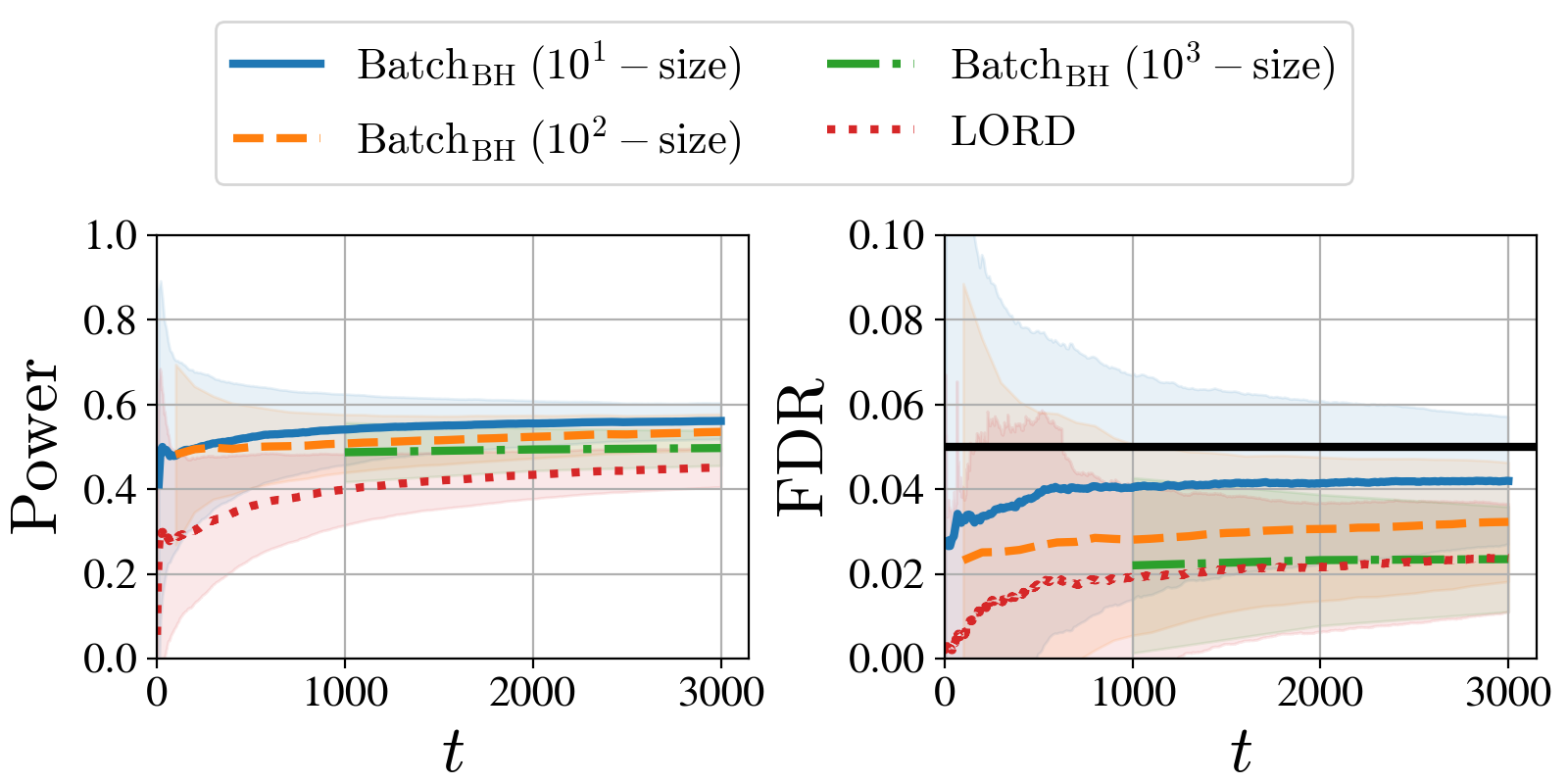}
   \includegraphics[width=0.5\columnwidth]{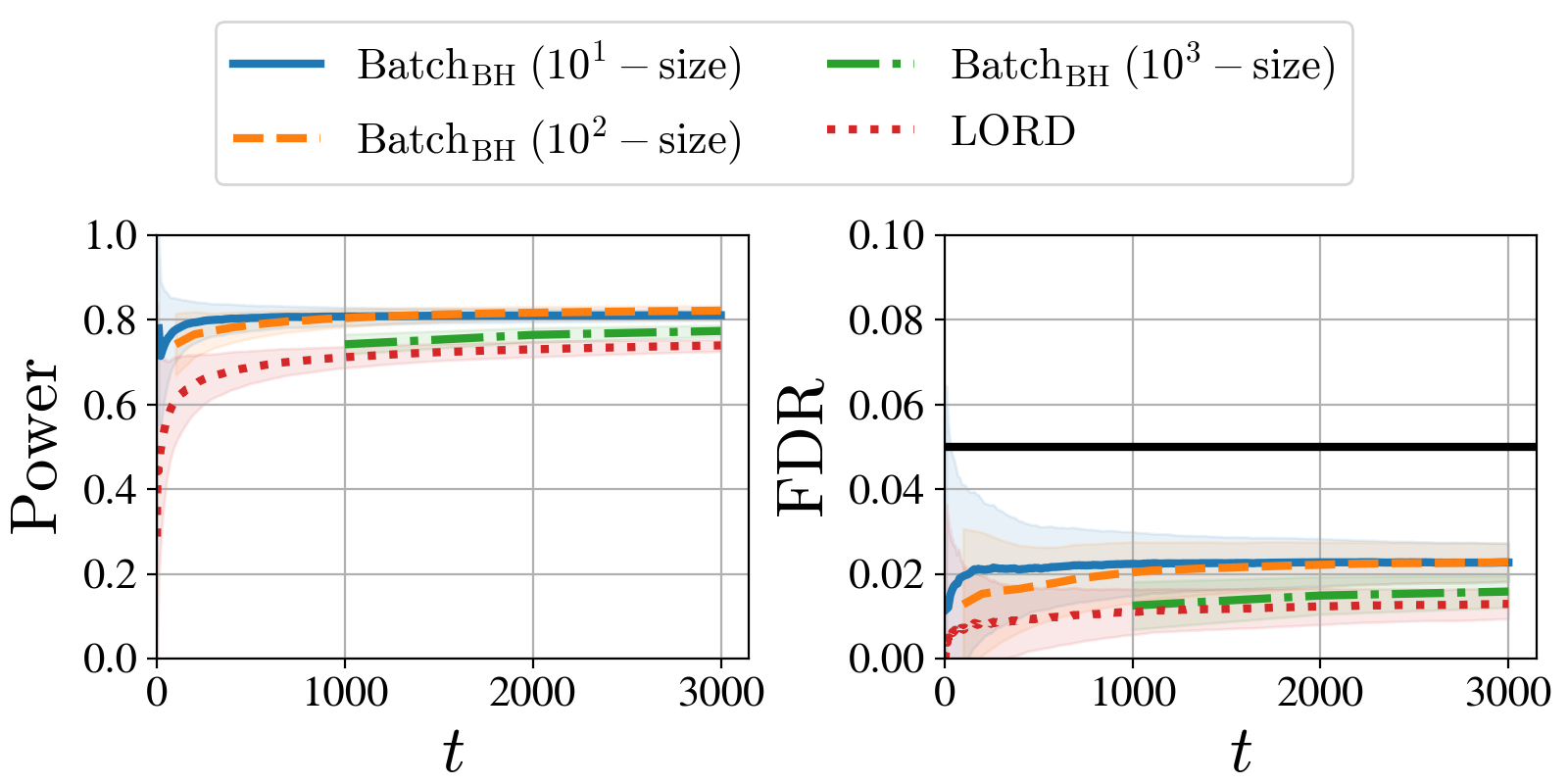}}
  \caption{Statistical power and FDR versus number of hypotheses seen $t$ for $\batchbh$ (at batch sizes 10, 100, and 1000) and LORD. We choose the probability of a non-null hypothesis to be $\pi_1=0.1$ (left) and $\pi_1=0.5$ (right). The observations under the null are $N(0,1)$, and the observations under the alternative are $N(3,1)$.}
  \label{fig:mean3_bbh_pi1_15}
\end{figure}


\begin{figure}[h]
  \centerline{
  \includegraphics[width=0.5\columnwidth]{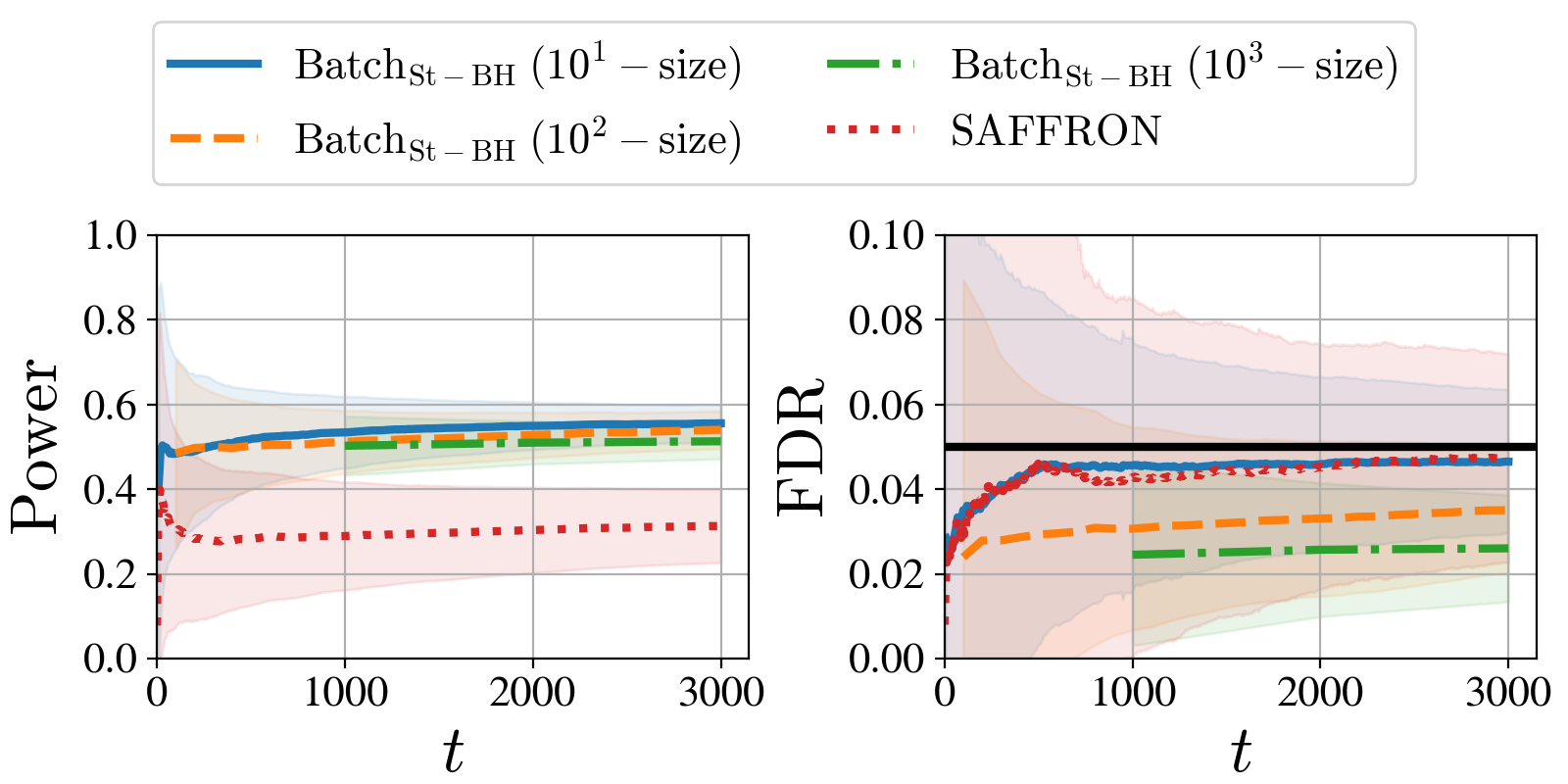}
  \includegraphics[width=0.5\columnwidth]{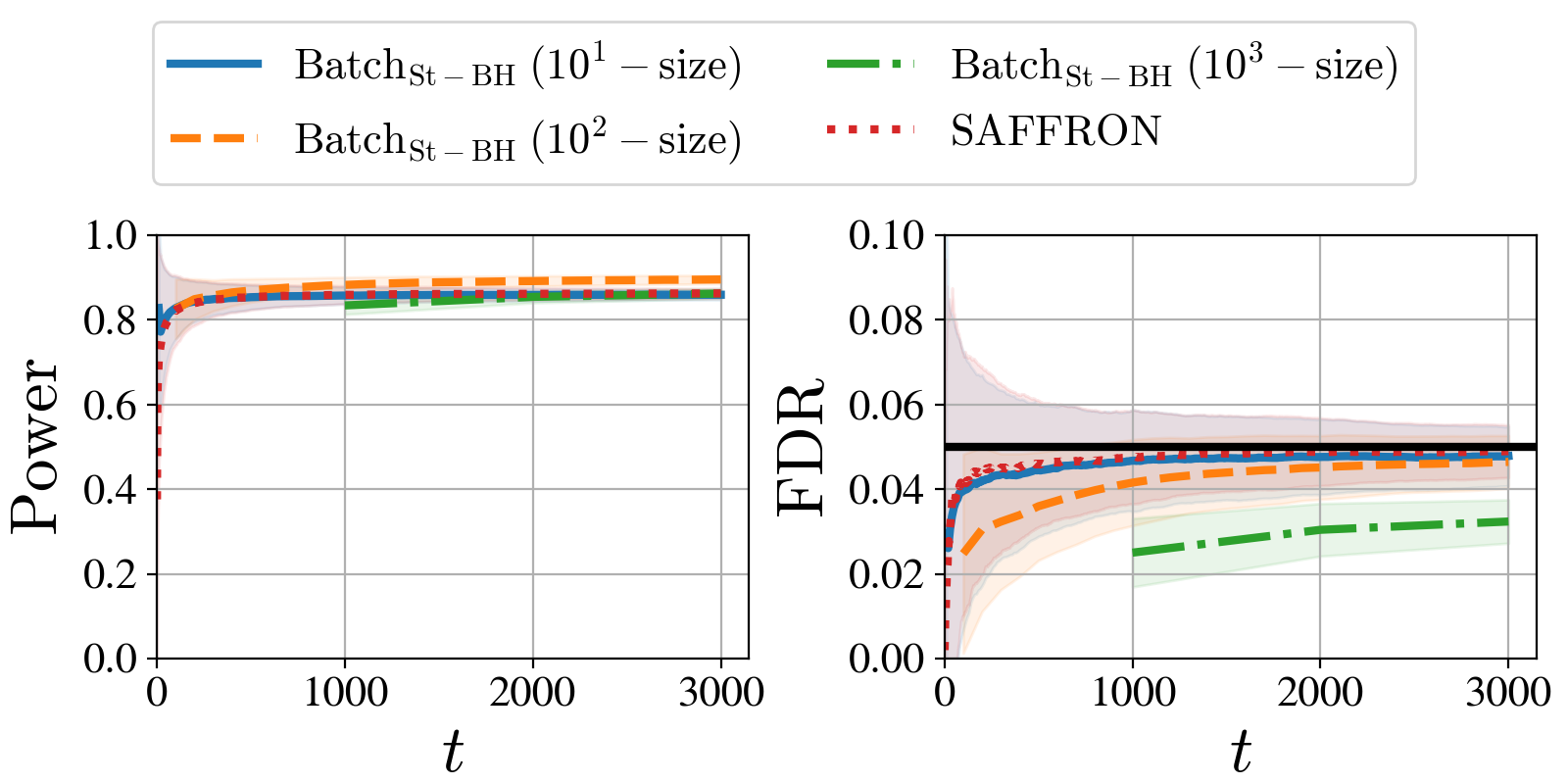}}
  \caption{Statistical power and FDR versus number of hypotheses seen $t$ for $\batchsbh$ (at batch sizes 10, 100, and 1000) and SAFFRON. We choose the probability of a non-null hypothesis to be $\pi_1=0.1$ (left) and $\pi_1=0.5$ (right). The observations under the null are $N(0,1)$, and the observations under the alternative are $N(3,1)$.}
  \label{fig:mean3_bsbh_pi1_15}
\end{figure}


\begin{figure}[h]
  \centerline{
  \includegraphics[width=0.5\columnwidth]{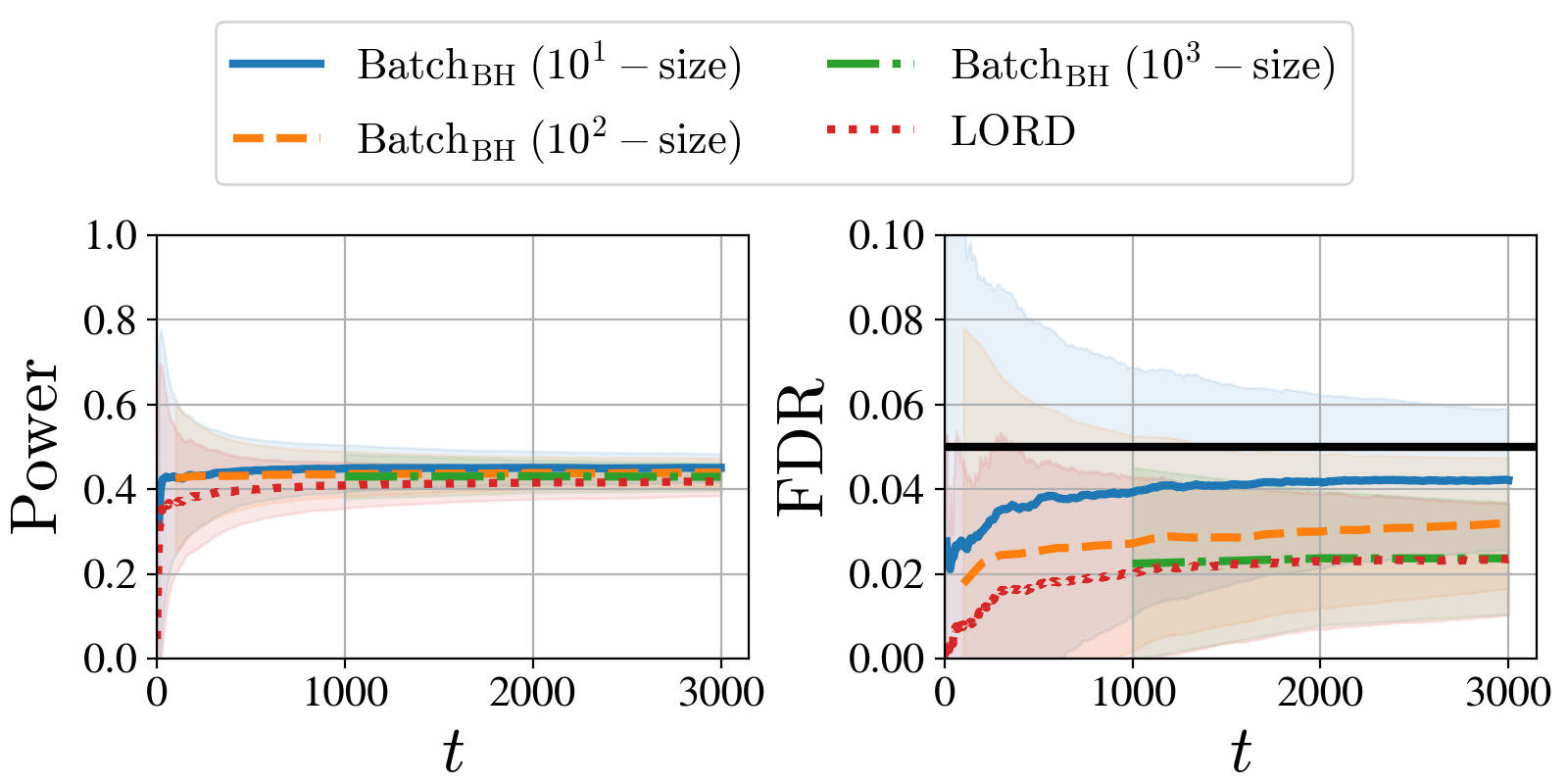}
  \includegraphics[width=0.5\columnwidth]{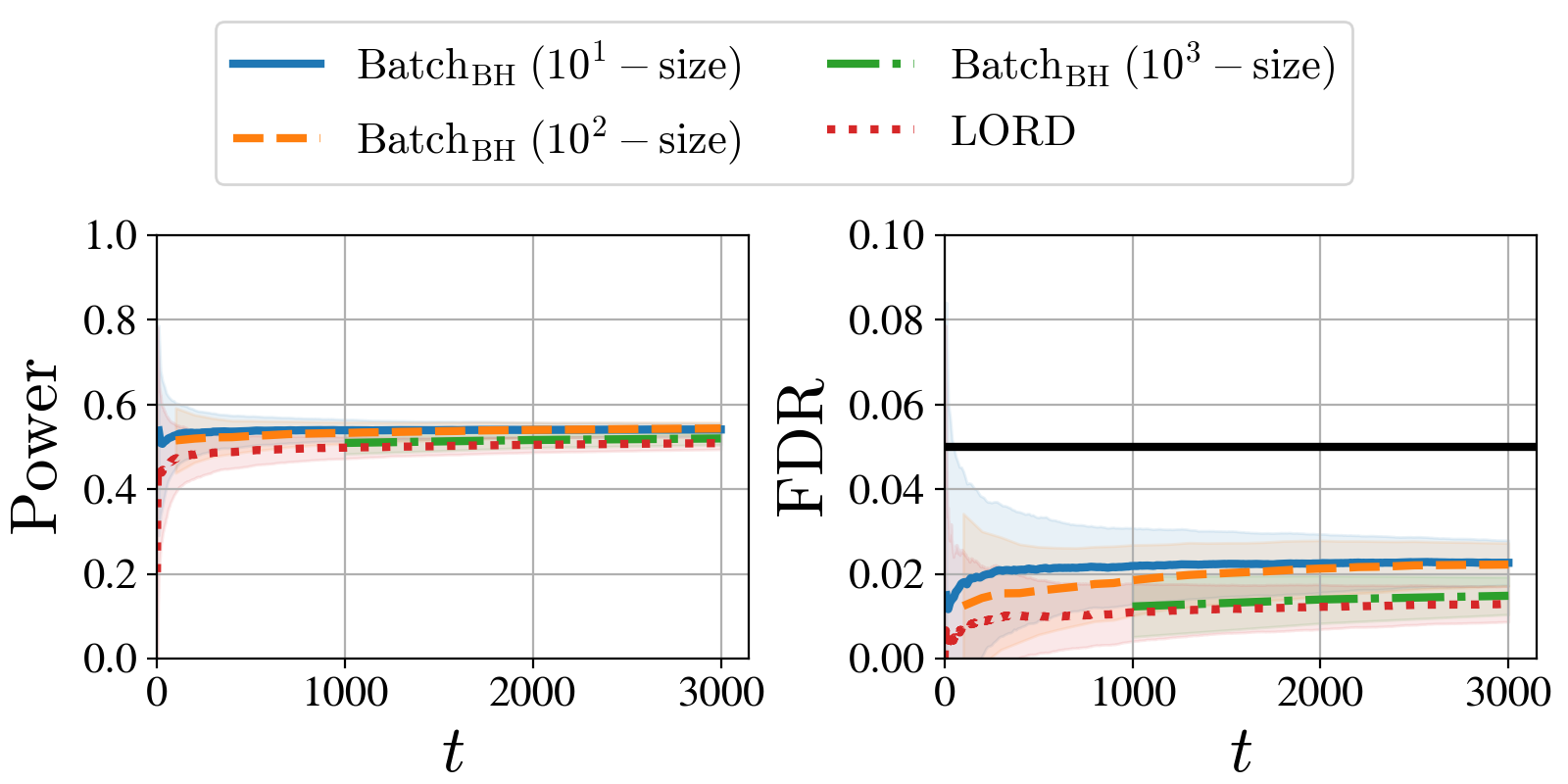}}
  \caption{Statistical power and FDR versus number of hypotheses seen $t$ for $\batchbh$ (at batch sizes 10, 100, and 1000) and LORD. We choose the probability of a non-null hypothesis to be $\pi_1=0.1$ (left) and $\pi_1=0.5$ (right). The observations under the null are $N(0,1)$, and the observations under the alternative are $N(\mu_1,1)$ where $\mu_1\sim N(0,2\log T)$.}
  \label{fig:mean0_bbh_pi1_15}
\end{figure}

\begin{figure}
  \centerline{
  \includegraphics[width=0.5\columnwidth]{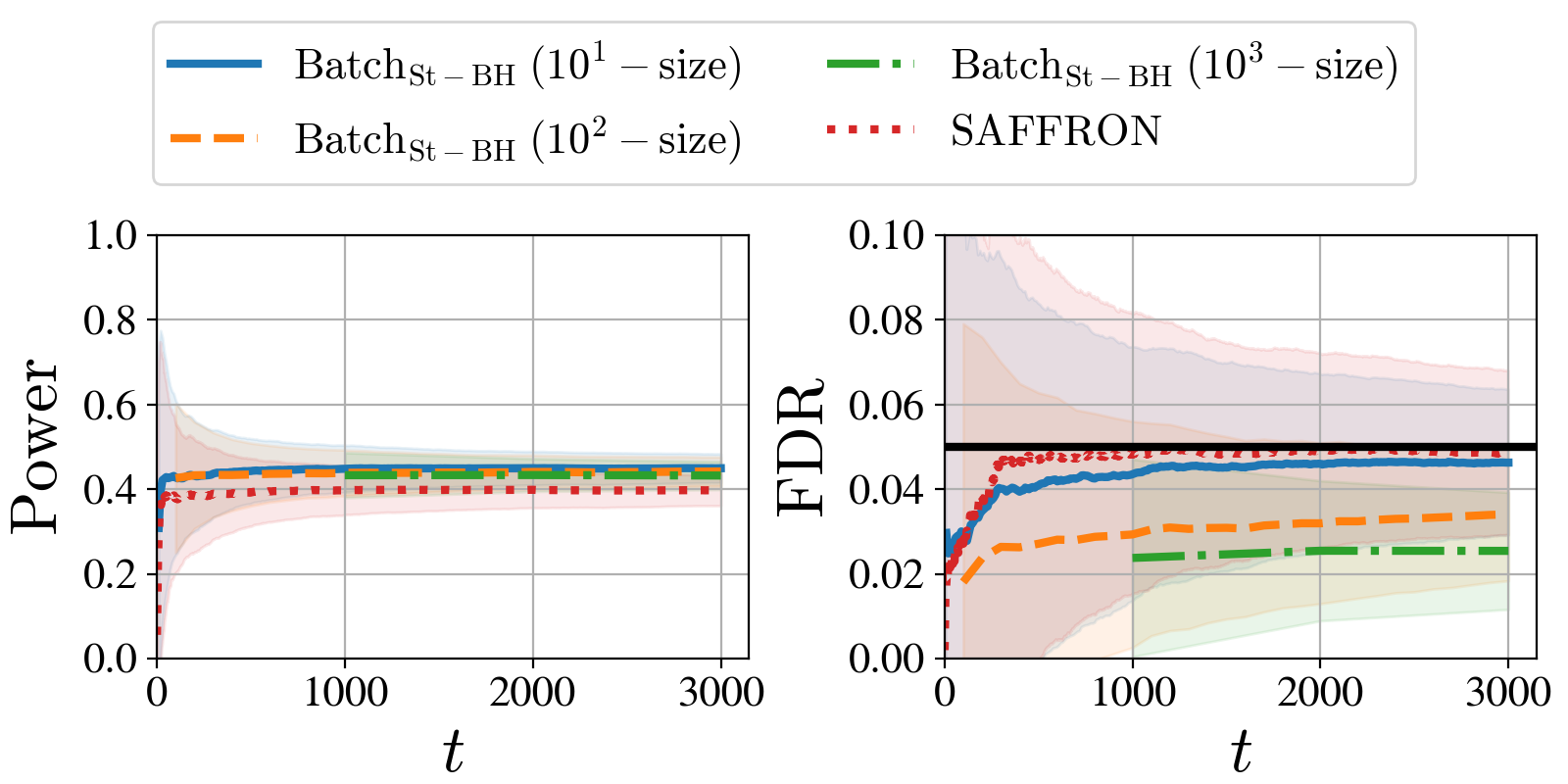}
  \includegraphics[width=0.5\columnwidth]{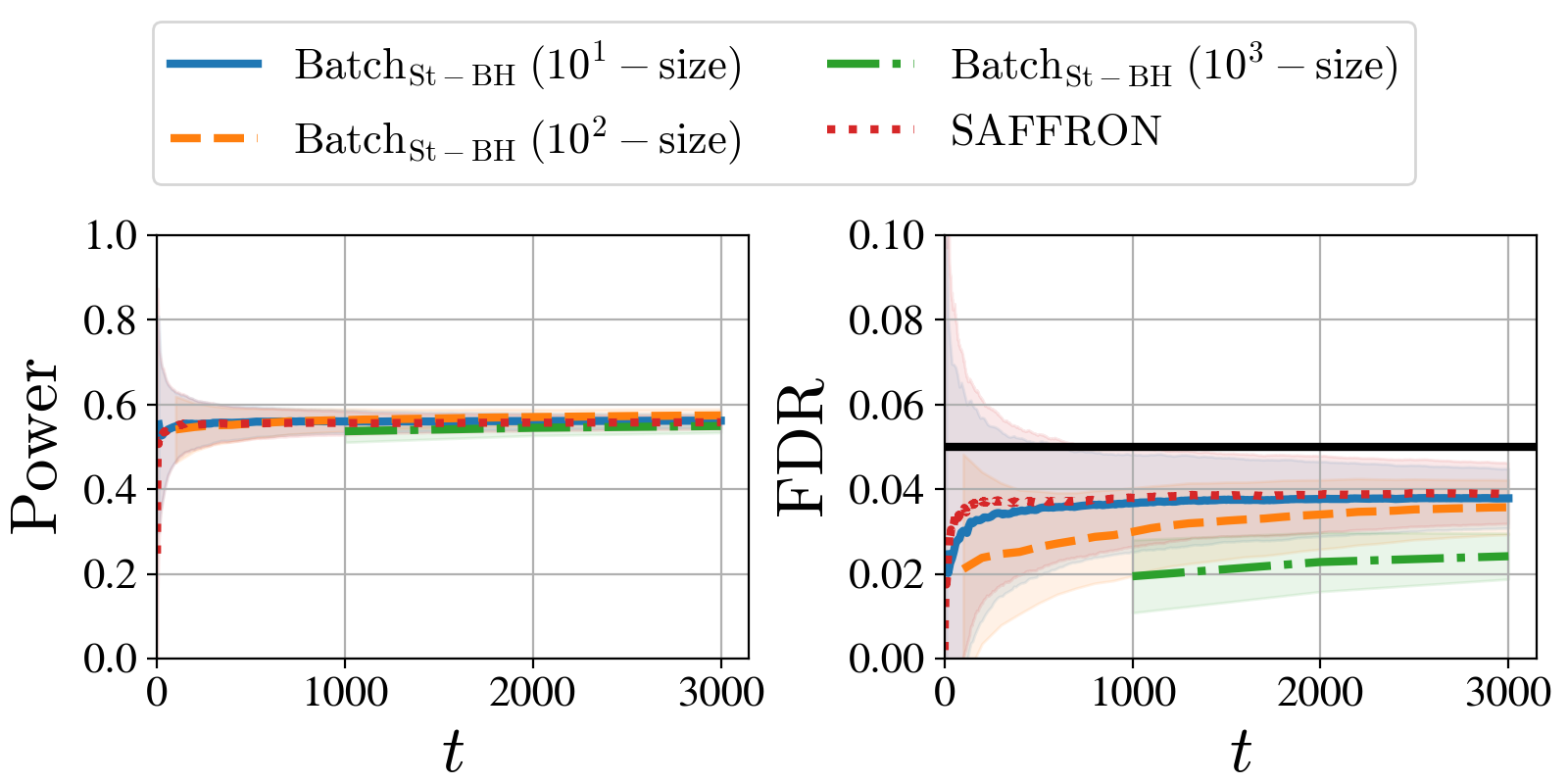}}
  \caption{Statistical power and FDR versus number of hypotheses seen $t$ for $\batchsbh$ (at batch sizes 10, 100, and 1000) and SAFFRON. We choose the probability of a non-null hypothesis to be $\pi_1=0.1$ (left) and $\pi_1=0.5$ (right). The observations under the null are $N(0,1)$, and the observations under the alternative are $N(\mu_1,1)$ where $\mu_1\sim N(0,2\log T)$.}
  \label{fig:mean0_bsbh_pi1_1}
\end{figure}

Additionally, in \figref{rdiff} we plot $R_t^+ - R_t$ for a single trial of $\batchbh$ and the first experimental setting of constant Gaussian means, at $\pi_1 = 0.1$. We observe similar behavior for $\batchsbh$ and other problem parameters as well. This experiment shows that $R_t^+ - R_t$ highly concentrates around the value 1, and in our experiments is no larger than 4. Hence, when designing new practical algorithms, it is a reasonable heuristic to assume $R_t^+ = R_t + 1$.

\begin{figure}[h]
  \centerline{
  \includegraphics[width=0.3\columnwidth]{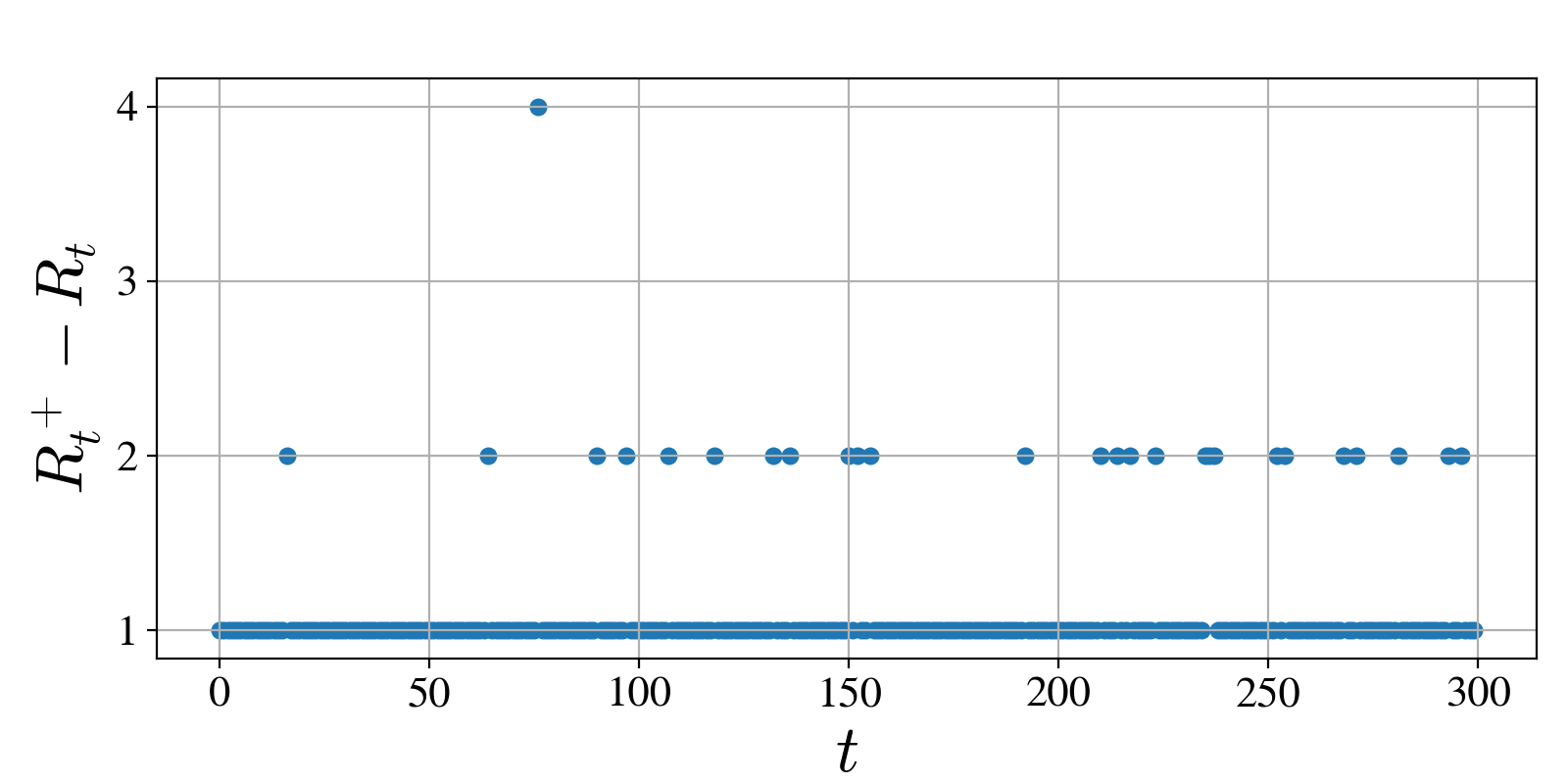}
  \includegraphics[width=0.3\columnwidth]{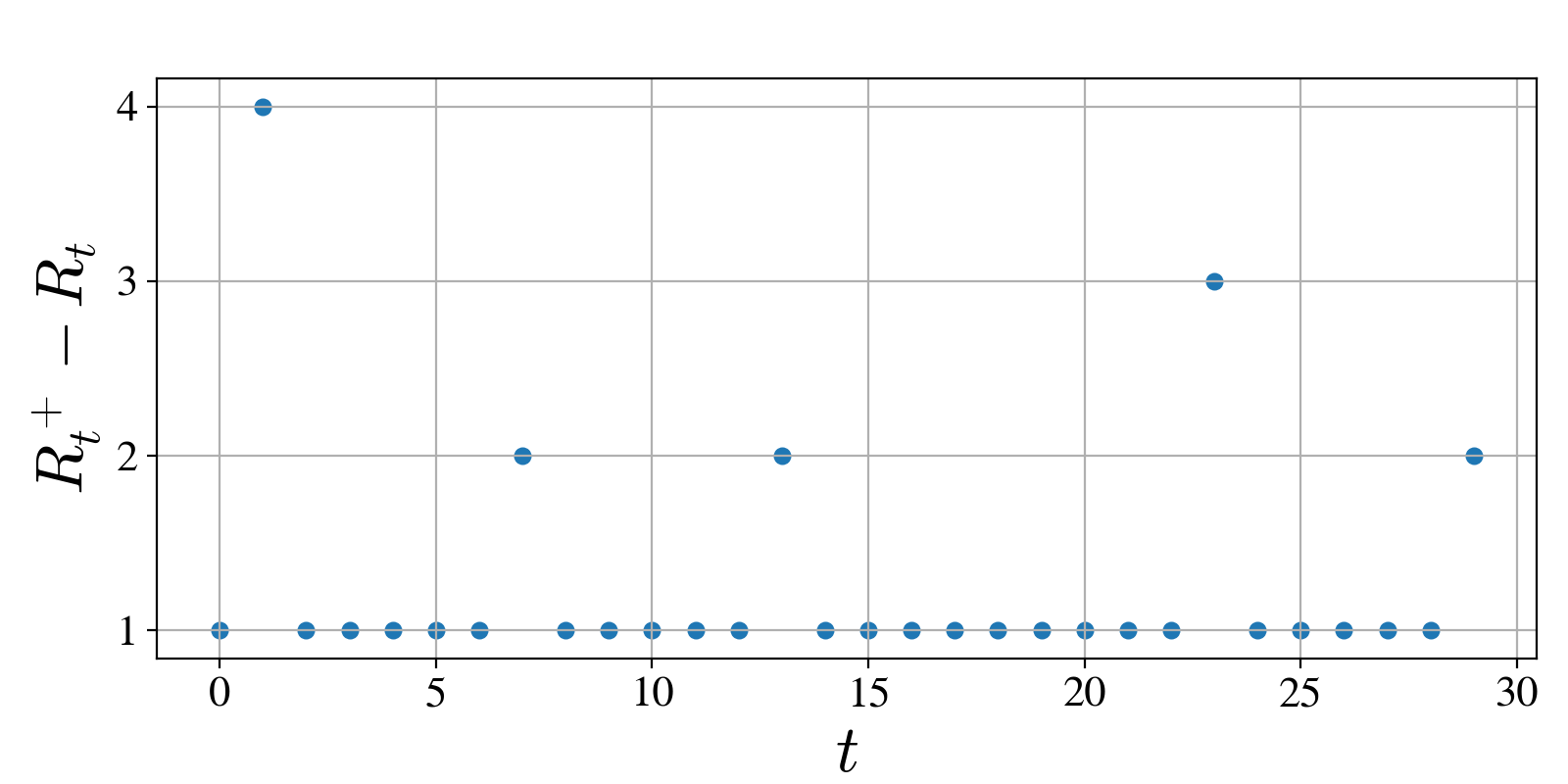}
  \includegraphics[width=0.3\columnwidth]{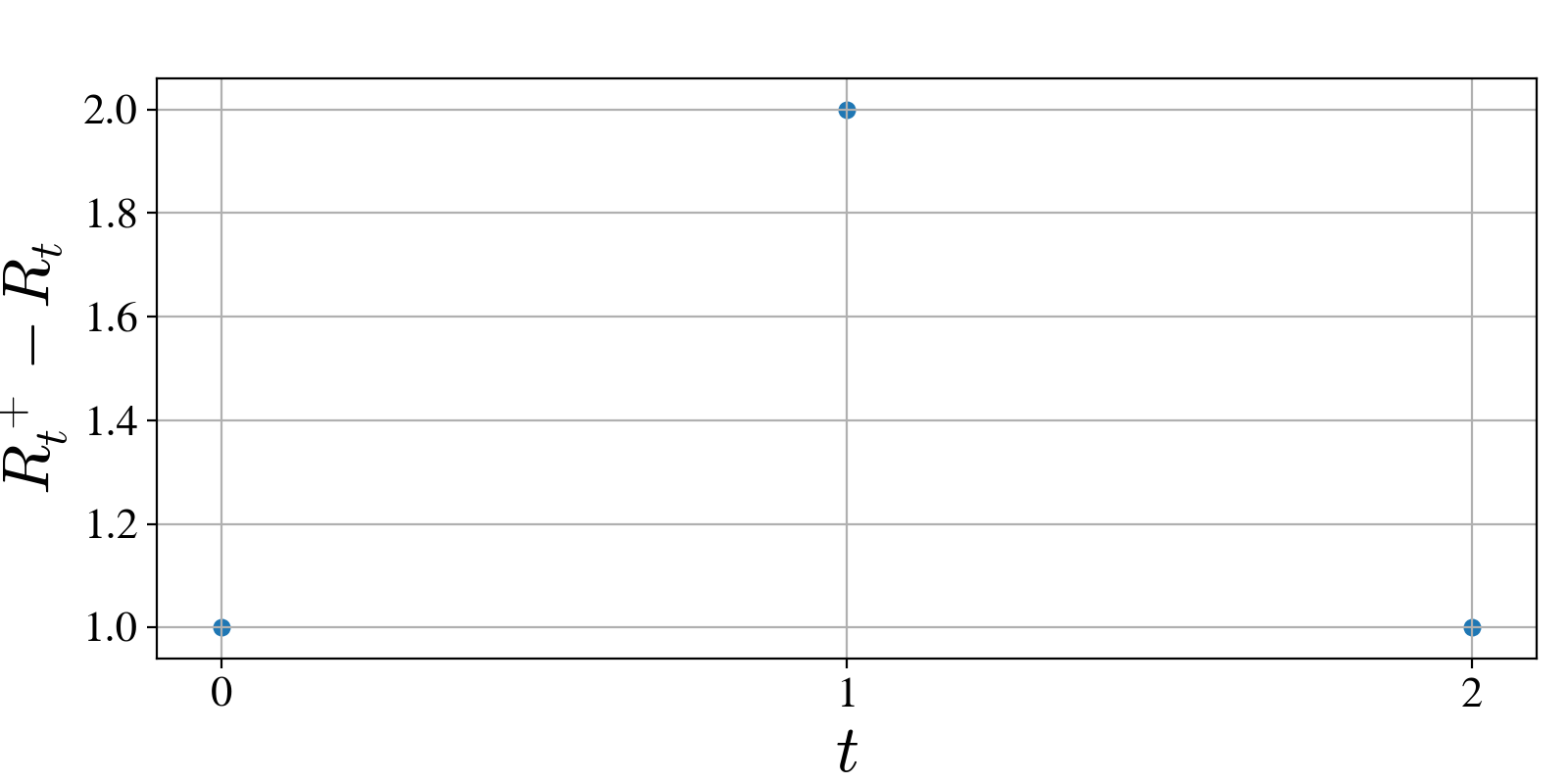}}
  \caption{$R_t^+ - R_t$ versus batch size index $t$ for $\batchsbh$, at batch sizes 10 (left), 100 (middle) and 1000 (right). We choose the probability of a non-null hypothesis to be $\pi_1=0.1$. The observations under the null are $N(0,1)$, and the observations under the alternative are $N(3,1)$.}
  \label{fig:rdiff}
\end{figure}


\section{Monotonicity in Numerical Experiments}

We verify numerically that $\batchbh$ and $\batchsbh$ are monotone with high probability, as required by \thmref{fdrbatchbh} and \thmref{fdrbatchsbh}. Although this is a heuristic way to justify the FDR control of our procedures, we found that both $\batchbh$ and $\batchsbh$ exhibit monotonicity with high probability, as well as FDR control, across various problem settings.

For a given $p$-value sequence, we first run either $\batchbh$ (or $\batchsbh$) as usual. We then randomly pick a batch $i$ and set a random $p$-value in that batch to $0$. Finally, we run $\batchbh$ (or $\batchsbh$) again on the modified $p$-value sequence and check whether the condition $\sum_{j=i+1}^t R_j \leq \sum_{j=i+1}^t \tilde R_j$ holds, where $\tilde R_j$ is the number of rejections in the $j$-th batch of the sequence in which the fixed $p$-value is set to $0$. If we find that the condition holds, then we deem $\batchbh$ (or $\batchsbh$) to be monotone on the given $p$-value sequence.

We do this for every $p$-value sequence created in \secref{mean3} and \secref{mean0}. This means that for each of the experimental settings, we perform this monotonicity check on 500 $p$-value sequences for each $\pi_1$ in $\{0.01,0.02,\dotsc,0.09\}\cup\{0.1,0.2,\dotsc,0.5\}$. For the experimental setting in \secref{mean3}, \figref{monotone_mean3} shows that $\batchbh$ is monotone on at least $97.4\%$ of the sequences, and that $\batchsbh$ is monotone on at least $96.6\%$ of the sequences. For the experimental setting in \secref{mean0}, \figref{monotone_mean0} shows that $\batchbh$ is monotone on at least $99.0\%$ of the sequences, and that $\batchsbh$ is monotone on at least $98.2\%$ of the sequences.

\begin{figure}[H]
	\centerline{
	\includegraphics[width=0.5\columnwidth]{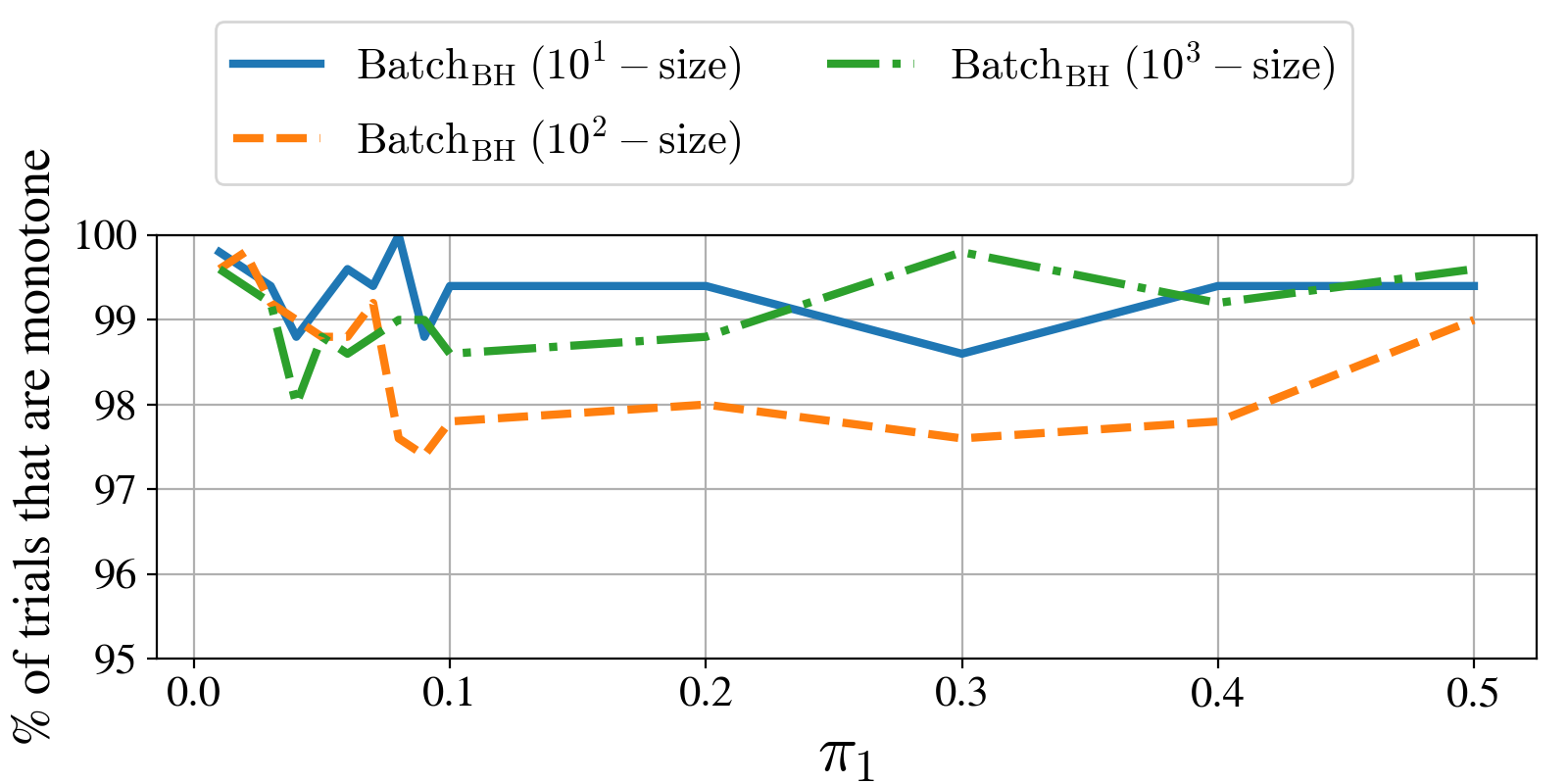}
	\includegraphics[width=0.5\columnwidth]{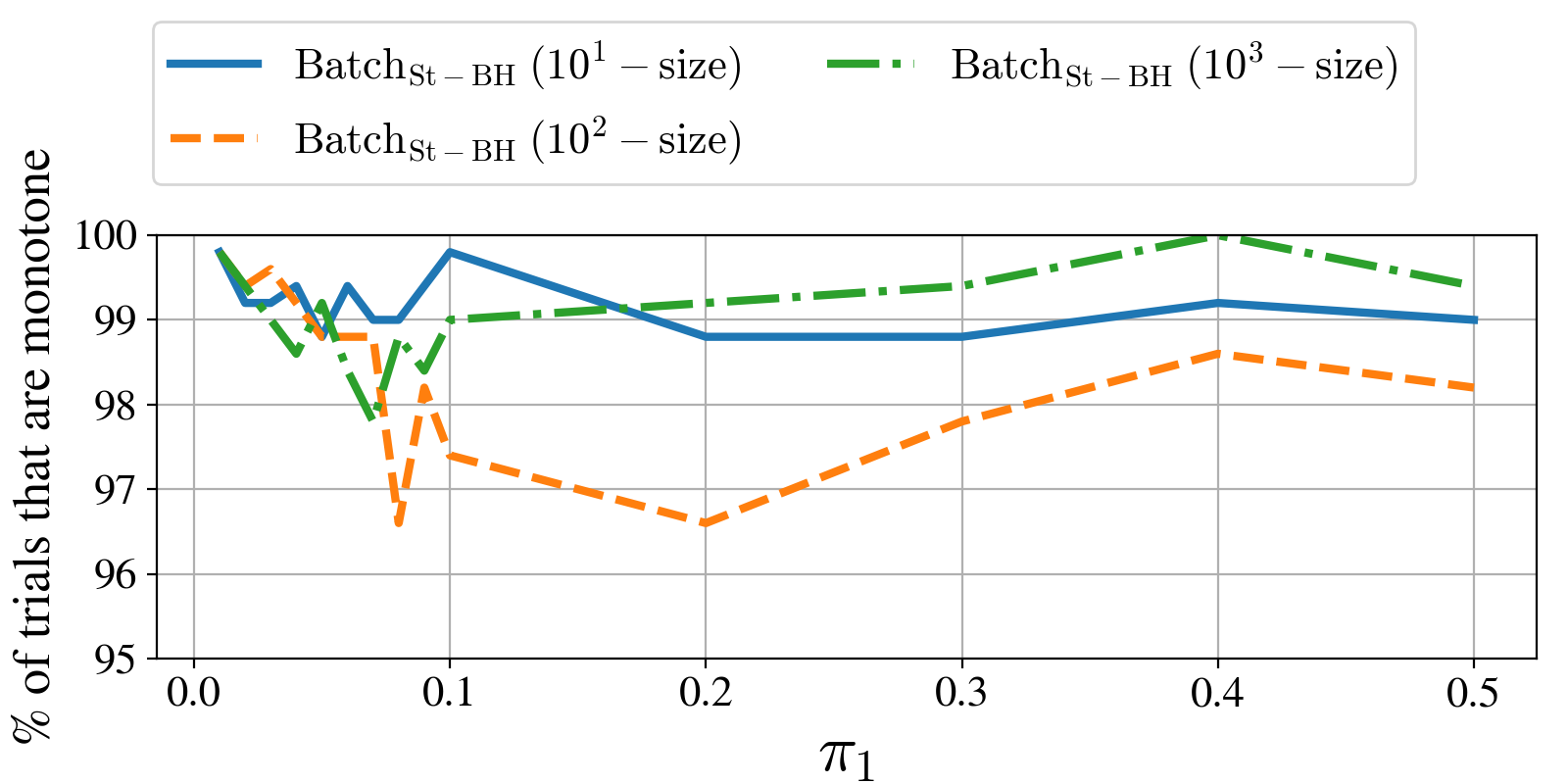}}
	\caption{For $\batchbh$, the minimum is 97.4\%. For $\batchsbh$, the minimum is 96.6\%.}
	\label{fig:monotone_mean3}
\end{figure}


\begin{figure}[H]
	\centerline{
	\includegraphics[width=0.5\columnwidth]{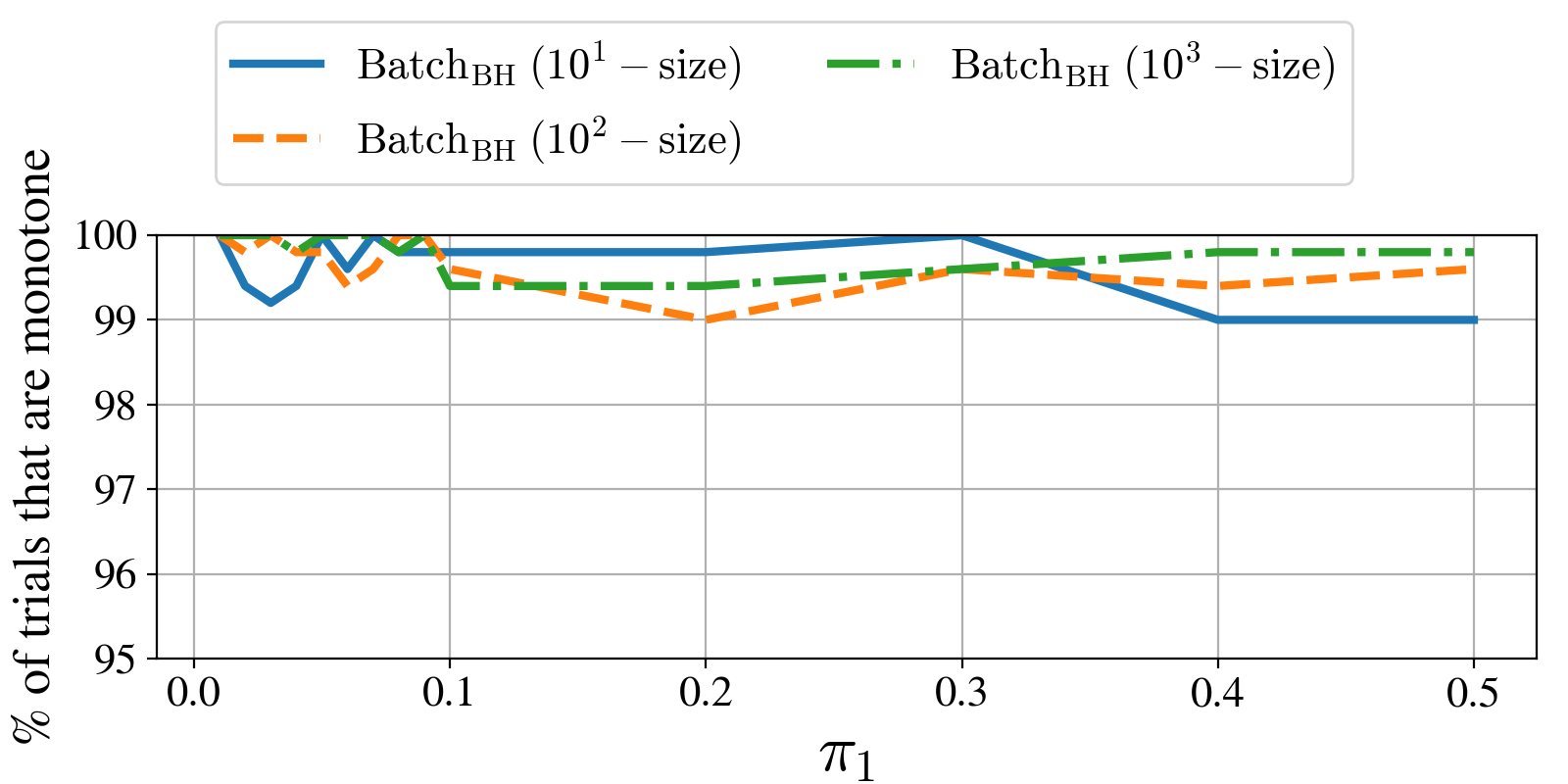}
	\includegraphics[width=0.5\columnwidth]{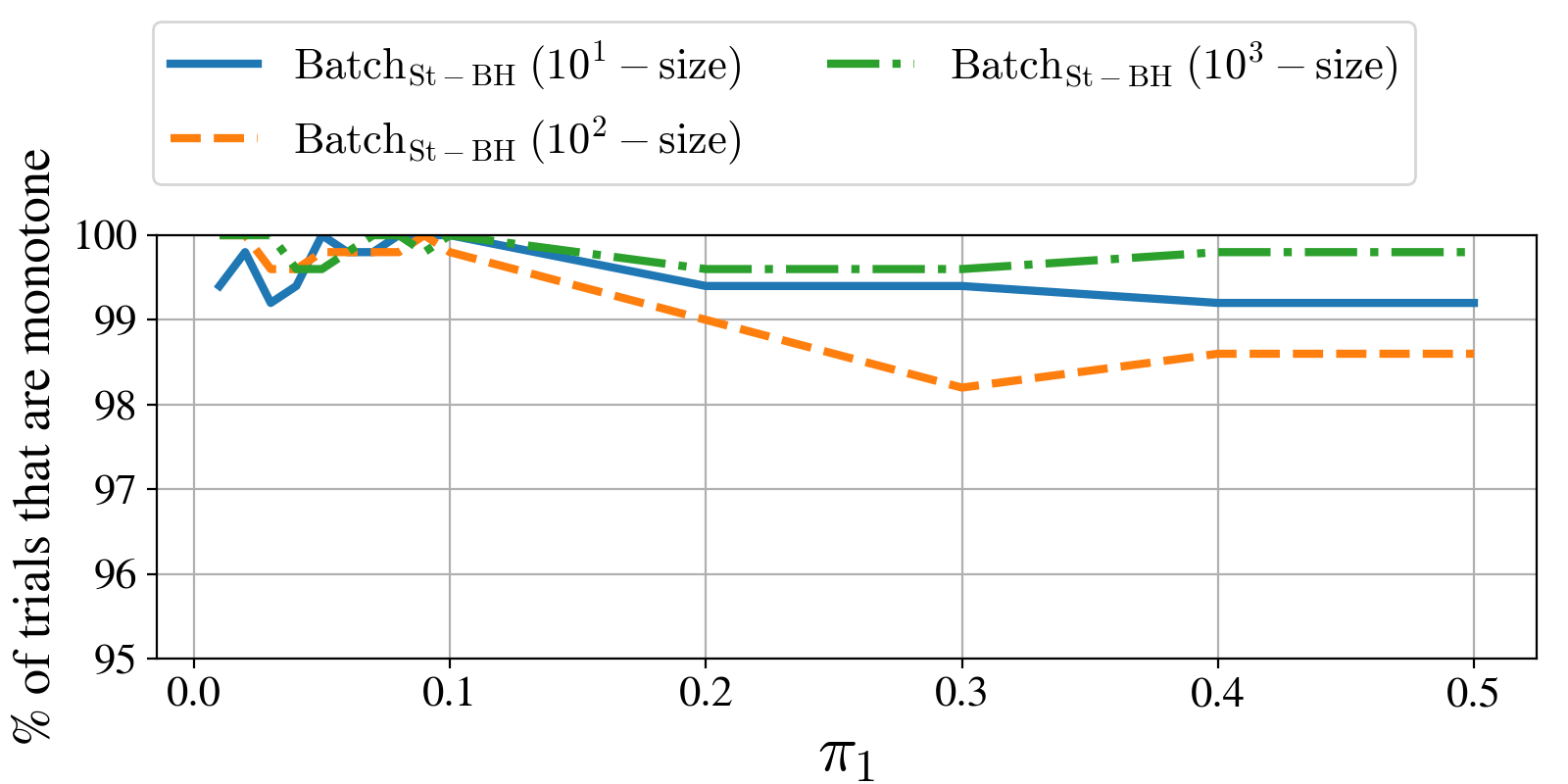}}
	\caption{For $\batchbh$, the minimum is 99.0\%. For $\batchsbh$, the minimum is 98.2\%.}
	\label{fig:monotone_mean0}
\end{figure}


\section{mFDR Control of the BH Procedure}

Recall the definition of \emph{modified}, or \emph{marginal}, false discovery rate up to time $t$:
$$\mfdr(t) = \frac{\EE{\sum_{i=1}^t |\nulls_i \cap \cR_i|}}{\EE{\sum_{i=1}^t |\cR_i|}}.$$

As discussed in \secref{discussion}, $\mfdr$ is a desirable false discovery metric due to its composition properties; ensuring $\mfdr$ control under level $\alpha$ in two disjoint batches of hypotheses guarantees $\mfdr$ at most $\alpha$ when the two batches of results are merged. It is thus natural to analyze $\mfdr$ control of the BH procedure. Unfortunately, it is not difficult to see that the BH algorithm does not imply $\mfdr$ control, although it does provide control asymptotically \citep{genovese2002operating}. Below we present a result of possibly independent interest, which shows that $\mfdr$ can be upper bounded in terms of the stability of the number of rejections. Our result implies favorable properties as the batch size tends to infinity, however it has been noted that the rejection set might be highly unstable for finite batch sizes \citep{gordon2007control}.
%

\begin{proposition}
Let the $p$-values $P_1,\dots,P_n$ be independent. Denote by $\cR$ the set of indices corresponding to the discoveries in the batch, and let $R = |\cR|$. Then, the Benjamini-Hochberg procedure at level $\alpha$ satisfies
$$\mfdr \leq \max\{1,\delta\}\alpha,$$
where $\delta \defn \sup_{i\in\nulls} \frac{\EEst{R}{P_i \in \cR}}{\EEst{R}{P_i \not\in \cR}}$.
\end{proposition}

\begin{proof}
Let $\nulls$ denote the nulls in $[n]$. Let the order statistic corresponding to $\mathbf{P}:=\{P_1,\dots,P_n\}$ be $P_{(1)},\dots, P_{(n)}$. Denote by $\mathbf{P}^{(-i)}$ the set $\mathbf{P}\setminus P_i$, and let $P^{(-i)}_{(j)}$ be the $j$-th order statistic in $\mathbf{P}^{(-i)}$. Define $R^{(-i)}$ to be the number of rejections when running \emph{modified} BH on $\mathbf{P}^{(-i)}$, which rejects the smallest $R^{(-i)}$ $p$-values in $\mathbf{P}^{(-i)}$, where $R^{(-i)} =\max\{1\leq j\leq n-1: P^{(-i)}_{(j)}\leq \frac{\alpha}{n}(j+1)\}$. For any $i,r\in[n]$, we have:
\begin{align*}
    \One{P_i\leq\frac{\alpha}{n}r, R = r} &= \One{P_i\leq\frac{\alpha}{n}r, P_{(r)}\leq \frac{\alpha}{n}r, P_{(r+1)}> \frac{\alpha}{n}(r+1),\dots,P_{(n)}> \frac{\alpha}{n}n}\\
    &= \One{P_i\leq\frac{\alpha}{n}r, P^{(-i)}_{(r-1)}\leq \frac{\alpha}{n}r, P^{(-i)}_{(r)}> \frac{\alpha}{n}(r+1),\dots,P^{(-i)}_{(n-1)}> \frac{\alpha}{n}n}\\
    &= \One{P_i\leq\frac{\alpha}{n}r, R^{(-i)} = r-1}.
\end{align*}
In words, if BH makes $r$ discoveries and a $p$-value $P_i$ is in the rejected set, then the modified BH ran on the set that drops $P_i$ will make \emph{exactly} $r-1$ discoveries.
Denote by $V$ the number of of false discoveries in $\cR$. We can express it as:\begin{align*}
    V &= \sum_{i\in\nulls} \One{P_i \leq \frac{\alpha}{n}R, R > 0}\\
    &= \sum_{i\in\nulls} \sum_{r = 1}^n \One{P_i \leq \frac{\alpha}{n}r, R = r}\\
    &= \sum_{i\in\nulls} \sum_{r = 1}^n \One{P_i \leq \frac{\alpha}{n}r, R = r, P_i\in\cR}\\
    &= \sum_{i\in\nulls} \sum_{r = 1}^n \One{P_i \leq \frac{\alpha}{n}r, R^{(-i)} = r-1, P_i\in\cR}.
\end{align*}
The third equality follows because the event $\{P_i\in\cR\}$ is implied by the event $\{P_i \leq \frac{\alpha}{n}r, R = r\}$, and the last equality just uses the first derivation in this proof. By the super-uniformity of null $p$-values, we have
$$\PP{P_i \leq \frac{\alpha}{n}r, R^{(-i)} = r-1, P_i\in\cR}\leq \frac{\alpha}{n}r \PPst{R^{(-i)} = r-1}{P_i \leq \frac{\alpha}{n}r, P_i\in\cR},$$
where we use the trivial bound $\PPst{P_i\in\cR}{P_i\leq\frac{\alpha}{n}r}\leq 1$. If the $p$-values are independent, then
$$\PPst{R^{(-i)} = r-1}{P_i \leq \frac{\alpha}{n}r, P_i\in\cR} = \PPst{R^{(-i)} = r-1}{P_i\in\cR}.$$
Combining the previous steps, we conclude
\begin{align*}
    \EE{V} & = \sum_{i\in\nulls} \sum_{r = 1}^n \PP{P_i \leq \frac{\alpha}{n}r, R^{(-i)} = r-1, P_i\in\cR}\\
    &\leq \sum_{i\in\nulls} \sum_{r = 1}^n \frac{\alpha}{n}r \PPst{R^{(-i)} = r-1}{P_i\in\cR}\\
    &= \frac{\alpha}{n} \sum_{i\in\nulls} \sum_{r = 1}^n (r-1+1) \PPst{R^{(-i)} = r-1}{P_i\in\cR}\\
    &= \frac{\alpha}{n} \sum_{i\in\nulls} \left( \EEst{R^{(-i)}}{P_i\in\cR} + \sum_{r = 1}^n \PPst{R^{(-i)} = r-1}{P_i\in\cR}\right).
\end{align*}

By the tower property and the first derivation in this proof,
$$\EEst{R^{(-i)}}{P_i\in\cR} = \EE{\EEst{R^{(-i)}}{R, P_i\in\cR}} = \EEst{R-1}{P_i\in\cR}.$$

Also, due to $\sum_{r = 1}^n \PPst{R^{(-i)}= r-1}{P_i\in\cR} = 1$:
$$\EE{V} \leq \frac{\alpha}{n} \sum_{i\in\nulls} \EEst{R}{P_i\in\cR}.$$

Denote by $\epsilon_i \defn \max\{\EEst{R}{P_i\in\cR} - \EEst{R}{P_i\not\in\cR},0\}$. Then
\begin{align*}
    \EE{R} &= \PP{P_i \in \cR}\EEst{R}{P_i\in\cR} + \PP{P_i \not\in \cR}\EEst{R}{P_i\not\in\cR}\\
   	    &\geq \EEst{R}{P_i\in\cR} - \epsilon_i\PP{P_i \not\in \cR}\\
    &\geq \EEst{R}{P_i\in\cR} - \epsilon_i.
\end{align*}

Therefore, $\EEst{R}{P_i\in\cR} \leq \EE{R} + \epsilon_i$, and with this we can conclude
$$\EE{V} \leq \frac{\alpha}{n} \sum_{i\in\nulls} (\EE{R} + \epsilon_i) \leq \alpha (\EE{R} + \max_i \epsilon_i).$$
Define $i^* := \arg\max_i \epsilon_i$. Rearranging the terms in the previous expression, we have
\begin{equation}
\label{eqn:mfdr}
    \frac{\EE{V}}{\EE{R} + \epsilon_{i^*}} = \frac{\EE{V}}{\EE{R}} \left(\frac{1}{1 + \frac{\epsilon_{i^*}}{\EE{R}}}\right) \leq \alpha.
\end{equation}
Now consider the term $\frac{\epsilon_{i^*}}{\EE{R}}$. It is strictly positive if and only if $\EEst{R}{P_{i^*}\in\cR} > \EEst{R}{P_{i^*}\not\in\cR}$, and so the maximizer in $\epsilon_{i^*}=\max\{\EEst{R}{P_{i^*}\in\cR} - \EEst{R}{P_{i^*}\not\in\cR},0\}$ is the first term if and only if $\EEst{R}{P_{i^*}\in\cR} > \EEst{R}{P_{i^*}\not\in\cR}$. Now suppose this indeed holds; then, since $\EE{R}$ is a convex combination of $\EEst{R}{P_{i^*}\in\cR}$ and $\EEst{R}{P_{i^*}\not\in\cR}$, we have
$$\frac{\epsilon_{i^*}}{\EE{R}} \leq \frac{\EEst{R}{P_{i^*}\in\cR} - \EEst{R}{P_{i^*}\not\in\cR}}{\EEst{R}{P_{i^*}\not\in\cR}} = \frac{\EEst{R}{P_{i^*}\in\cR}}{\EEst{R}{P_{i^*}\not\in\cR}} - 1.$$
All previous observations combined, we can conclude that
$$1 + \frac{\epsilon_{i^*}}{\EE{R}} \leq 1 + \max\left\{0,\frac{\EEst{R}{P_{i^*}\in\cR}}{\EEst{R}{P_{i^*}\not\in\cR}} - 1\right\} = \max\left\{1,\frac{\EEst{R}{P_{i^*}\in\cR}}{\EEst{R}{P_{i^*}\not\in\cR}}\right\}:= \max\{1,\delta\},$$
where we define $\delta \defn \sup_{i\in\nulls} \frac{\EEst{R}{P_i \in \cR}}{\EEst{R}{P_i \not\in \cR}}$.
Going back to equation \eqnref{mfdr}, this implies
$$\frac{\EE{V}}{\EE{R}} \leq \max\{1,\delta\} \alpha,$$
as desired.
\end{proof}

\end{document}